\documentclass[10pt]{iopart}

\usepackage{iopams}
\usepackage{amssymb}
\usepackage{amsfonts}
\usepackage{amstext}
\usepackage{graphicx}
\usepackage{setstack}
\usepackage{amscd}
\usepackage{mathrsfs}
\usepackage{stmaryrd}
\usepackage{slashed}
\usepackage{color}

\newcommand{\ihbar}{\imath \hbar}

\newcommand{\Sp}{\mathrm{Sp}}

\renewcommand{\Re}{\mathrm{Re}}
\renewcommand{\Im}{\mathrm{Im}}
\newcommand{\Ted}{\underset{\rightarrow}{\mathbb{T}e}}
\newcommand{\Teg}{\underset{\leftarrow}{\mathbb{T}e}}

\newcommand{\llangle}{\langle \hspace{-0.2em} \langle}
\newcommand{\rrangle}{\rangle \hspace{-0.2em} \rangle}
\newcommand{\lllangle}{\langle \hspace{-0.2em} \langle \hspace{-0.2em} \langle}
\newcommand{\rrrangle}{\rangle \hspace{-0.2em} \rangle \hspace{-0.2em} \rangle}

\newcommand{\id}{\mathrm{id}}
\newcommand{\xrightarrow}[1]{\overset{#1}{\longrightarrow}}

\newcommand{\Env}{\mathrm{Env}}
\newcommand{\Der}{\mathrm{Der}}
\newcommand{\dnc}{\mathbf{d}}
\newcommand{\Diff}{\mathrm{Diff}}
\newcommand{\Dis}{\mathbf{D}}
\newcommand{\Lie}{\mathrm{Lie}}

\newtheorem{theo}{Theorem}
\newtheorem{prop}{Property}

\newenvironment{proof}{\noindent \textit{Proof:}}{\hfill $\Box$ \\}

\begin{document}

\title[Emergent gravity and D-brane adiabatic dynamics]{Emergent gravity and D-brane adiabatic dynamics: emergent Lorentz connection}

\author{David Viennot}
\address{Institut UTINAM (CNRS UMR 6213, Universit\'e de Bourgogne-Franche-Comt\'e, Observatoire de Besan\c con), 41bis Avenue de l'Observatoire, BP1615, 25010 Besan\c con cedex, France.}

\begin{abstract}
We explore emergent geometry of the spacetime at the microscopic scale by adiabatic transport of a quasi-coherent state of a fermionic string, with quantum spacetime described by the matrix theory (BFSS matrix model). We show that the generator of the Berry phase is the shift vector of the spacetime foliation by spacelike surfaces associated with the quasi-coherent state. The operator-valued generator of the geometric phase of weak adiabatic transport is the Lorentz connection of the emergent geometry which is not torsion free at the microscopic scale. The effects of the torsion seem consistent with the usual interpretation of the Berry curvature as a pseudo magnetic field.
\end{abstract}


\section{Introduction}
Matrix models \cite{Zarembo} provide a description of the quantum spacetime as a noncommutative manifold associated with a D-brane. Such an interesting model is the  BFSS (Banks-Fischler-Shenker-Susskind) matrix theory \cite{Banks}. Some works \cite{Klammer, Steinacker, Kunter, Sahakien} show that gravity emerges at macroscopic scale from the noncommutativity of the manifold obtained by quantisation of an embedding flat spacetime. More precisely, the semi-classical limit $\frac{1}{\ihbar} [\hat f,\hat g] \xrightarrow{\hbar \to 0} \theta^{ij} \partial_i f \partial_j g$ (for $\hat f$ and $\hat g$ two observables of the quantum spacetime) induces a Poisson structure $(\theta^{ij})$ describing an emergent geometry at the macroscopic scale (the semi-classical limit being equivalent to the thermodynamical limit where the number of strings tends to infinity). The effective metric of this emergent geometry is $g^{ij} = \theta^{ik} \theta^{jl} \partial_k x^a \partial_l x^b \eta_{ab}$ where $(x^a)$ are the coordinates onto the embedding flat spacetime and which are the semi-classical limits of operators $(X^a)$ describing the D-brane. At the microscopic level (with a ``small'' number of strings), the spacetime is purely quantum since it is described in matrix models as a noncommutative manifold. To compare this one to classical spacetimes, we can use two usual methods of quantum mechanics. The first one consists to consider the mean values of the quantum spacetime observables which obey to classical laws by the Ehrenfest theorem (and so to consider the ``mean value'' of the noncommutative manifold as a classical manifold). The second one consists to consider equivalents to coherent states for the algebra of quantum spacetime observables. Such states are the quantum states closest to classical states, since they minimise the quantum uncertainties. We call emergent geometry at the microscopic level, the one defined by the quasi-coherent states $|\Lambda(x)\rrangle$ \cite{Schneiderbauer} (which are strongly related to the Perelomov coherent states of a Lie algebra \cite{Perelomov}) and which minimise the spacetime quantum uncertainties (see \cite{Schneiderbauer}). These quasi-coherent states are eigenvectors of the noncommutative Dirac operator $\slashed D_x$ of a probe fermionic string (gravity and spacetime geometry are physically revealed by test particles).\\
An interesting point of view has been presented in \cite{Berenstein,Karczmarek,Huboda}. The D-brane defines a $U(1)$-principal bundle endowed with the Berry connection $A = -\imath \llangle \Lambda|d|\Lambda \rrangle$ \cite{Shapere} and in some examples $\theta$ seems to identify with a dual tensor of the Berry curvature $F=dA$. But this point, and the precise role of the Berry connection, seem to be unclear in these previous works. This suggests a possible link between the emergent geometry and the adiabatic dynamics of the probe fermionic string. Indeed, adiabatic evolution can be defined by the adiabatic limit $U_{ad}(s) = \lim_{T \to + \infty} \Teg^{-\imath \frac{T}{\hbar} \int_0^s \hat H(s')ds'}$ where $T$ is the duration of the dynamics, $s=t/T$ is the reduced time and $\hat H$ is a time-dependent Hamiltonian ($\Teg$ denotes the time-ordered exponential, i.e. the Dyson series). By application of an adiabatic theorem \cite{Teufel} we have
\begin{equation}
  U_{ad}(s) = \sum_i e^{-\imath \frac{T}{\hbar} \int_0^s (\lambda_i(s')+A_i(s'))ds'} |\lambda_i(s)\rangle\langle \lambda_i(0)|
\end{equation}
where $(\lambda_i)$ are the instantaneous eigenvalues of $\hat H$ (supposed without crossing) with instantaneous eigenvectors $(|\lambda_i\rangle)$ and $A_i$ is the Berry connection (the Berry phase generator) for the $i$-th eigenvector. We see that for the evolution operator it is equivalent to consider the adiabatic limit $T \to +\infty$ or to consider the semi-classical limit $\hbar \to 0$ (or in other words, the relevant limit is $\frac{\hbar}{T} \to 0$). Reciprocally, if we rewrite the Heisenberg equation with the reduced time $\frac{d\hat f}{ds} = \frac{T}{\ihbar} [\hat f,\hat H]$, we see that $\frac{T}{\ihbar} [\hat f,\hat H] \xrightarrow{\frac{\hbar}{T}\to 0} \{ f,H\}$ (where the braces denote the Poisson bracket). Another heuristic argument is in favour of the relevance of the adiabatic regime in matrix model. The emergent geometry being revealed by dynamics of a test particle (probe fermionic string), the time ratio appearing is $\frac{t_P}{T}$ where $T$ is the characteristic duration of the particle transport and $t_P$ is the Planck time which characterises the inner quantum evolution of the quantum spacetime. We have clearly $\frac{t_P}{T} \ll 1$ justifying the adiabatic limit. Since the quasi-coherent states are eigenvectors of $\slashed D_x$, at the adiabatic limit they define the time-dependent mean values of the spacetime quantum observables. And so, the two approaches to compare quantum spacetimes to classical spacetimes (Ehrenfest theorem and coherent states) are in this context the same thing. The quasi-coherent state ensures a behaviour closest to a classical one for the space, and the adiabatic limit ensures a time evolution closest to a classical one ($\frac{\hbar}{T} \to 0$). The emergent geometry at the microscopic level is the one which emerges from the adiabatic limit with the quasi-coherent state.\\
The strict adiabatic limit defines strong adiabatic regimes. But weaker adiabatic regimes can be also considered (with eigenvalue crossings, with non-adiabatic transitions restricted to a small group of eigenvectors,...). In past works \cite{Viennot1,Viennot2,Viennot3}, we have studied a weak adiabatic regime where a quantum system is submitted to a competition between the adiabatic transport and the entanglement with another system (called environment). In that case, we can define a dynamics which is adiabatic with respect to the environment but not with respect to the studied system itself. In such an approach the Berry connection becomes operator-valued. In the case of the adiabatic transport of a probe fermionic string, we can have also entanglement between the spin degree of freedom with the D-brane. What does the weak adiabatic regime provide from the viewpoint of the emergent gravity? The previous works \cite{Klammer, Steinacker, Kunter, Sahakien} concerning the emergent geometry at the thermodynamical limit only focus on the emergent metric (or equivalently on the emergent tetrads). At the thermodynamical limit, to find the usual general relativity at the macroscopic scale, only torsion free geometries are considered. But is it really the case at the microscopic scale? The existence of a torsion in string theory is discussed in different models \cite{Hammond, Li, Hammond2, Sarkar}. In the context of a matrix model at the semi-classical limit, a torsion has been found in \cite{Steinacker2} which is significant at the cosmic scale and it could be a candidate to explain the dark matter problem. In this paper, we want to show that a torsion can also arise in the BFSS matrix model with manifestations at the microscopic scale (in the quasi-coherent picture). This one is related to the adiabatic limit and seems consistent with the usual interpretation of the Berry curvature.\\

After a presentation of the emergent geometry theory from the viewpoint of the adiabatic approach in section 2 (where we discuss the role of the Berry connection in matrix model), we show section 3 that the Berry connection of the weak adiabatic regime defines a Lorentz connection which completes the emergent geometry and which is not necessarily torsion free. Section 4 presents simple applications. Moreover, in \ref{adiabth} we prove a weak adiabatic theorem which is applicable in the present context, justifying mathematically the consistency of the adiabatic ansatz. \ref{NCgeom} explores in more details the relation between the adiabatic formalism developed in this paper with the noncommutative geometry and treats the diffeomorphism gauge changes. \ref{higherdim} generalises the developments of the core of this paper to fast evolving spacetimes. A last appendix presents some technical calculations used for the examples.\\

\textit{From this point, throughout this paper, we consider the unit system such that $\hbar=c=G=1$ ($\ell_P=t_P=m_P=1$ for the Planck units).\\ Moreover we denote by $\Omega^n(M,\mathfrak g)$ the set of the $\mathfrak g$-valued differential $n$-forms of the classical manifold $M$, where $\mathfrak g$ is an algebra.\\ We adopt the Einstein notations with lowercase latin indices starting from 1 and greek and capital latin indices starting from 0.}

\section{Strong adiabatic transport and emergent geometry}\label{sec1}
\subsection{Dirac-Einstein equation}
In order to fix some notations and to introduce some reference equations, we start by recalling some basic facts about the Dirac equation in curved spacetime:
\begin{equation} 
  \imath \gamma^I e^\mu_I (\partial_\mu + \omega_\mu) \Psi = 0
\end{equation}
where $\Psi$ is a massless spinor field, $(\gamma^I)$ are the Dirac matrices, $(e^\mu_I)$ are the spacetime tetrads and $\omega = \omega^{IJ}_\mu \mathscr D(L_{IJ}) dx^\mu \in \Omega^1(M,\mathfrak{sl}(2,\mathbb C))$ is the Lorentz connection ($\mathscr D$ denoting a representation of the $\mathfrak{sl}(2,\mathbb C)$ algebra, $(L_{AB})$ are the generators of the Lorentz group, $M$ denotes the spacetime manifold). The metric is defined by $g_{\mu \nu} = e^I_\mu e^J_\nu \eta_{IJ}$ (where $(\eta_{IJ})$ is the Minkowski metric of the flat spacetime viewed by an ideal local Galilean observer) and the Christoffel symbols are $\Gamma^\mu_{\rho \nu} = e^\mu_I \partial_\rho e^I_\nu + e^\mu_I \omega^{IJ}_\rho e_{J\nu}$.\\

In synchronous coordinates (i.e. metric in Gaussian normal form $ds^2 = dt^2 - g_{ab} dx^a dx^b \iff e^\mu_0=\delta^\mu_0$ and $e^0_a=0$, \cite{Landau,Cook}) the Dirac equation becomes
\begin{equation}
  \imath \gamma^0 \partial_0 \Psi + \imath \gamma^i e^a_i \partial_a \Psi + \imath \gamma^I e_I^\mu \omega_\mu \Psi = 0
\end{equation}
and by adopting the Weyl representation, we have
\begin{equation} \label{Diraceq}
  \imath \partial_0 \tilde \psi - \imath \sigma^i e^a_i \partial_a \tilde \psi + \imath(\omega_0-\sigma^i e_i^a \omega_a) \tilde \psi = 0
\end{equation}
where $(\sigma^i)$ are the Pauli matrices and $\tilde \psi$ is the half-part of the Dirac spinor. Let $\psi = U_\omega^{-1} \tilde \psi$ with $U_\omega \in \mathscr D(SL(2,\mathbb C))$ be a new representation of the spinor field such that
\begin{equation} \label{reducedDiraceq}
  \imath \partial_0 \psi - \imath \sigma^i e_i^a \partial_a \psi = 0
\end{equation}
By replacing $\tilde \psi$ by $U_\omega \psi$ in the Dirac equation, we have:
\begin{eqnarray}
  & & \left[\imath(\partial_0-\sigma^i e^a_i \partial_a)U_\omega\right]\psi+U_\omega\left[\imath(\partial_0-\sigma^i e^a_i \partial_a)\psi\right] \nonumber \\
  & & \qquad + \imath(\omega_0-\sigma^ie^a_i\omega_a)U_\omega \psi = 0 \\
  & \iff & (\partial_0 - \sigma^i e_i^a) U_\omega = -(\omega_0-\sigma^i e_i^a \omega_a) U_\omega\\
  & \iff & \sigma_e^\mu \partial_\mu U_\omega = -\sigma_e^\mu \omega_\mu U_\omega \\
  & \iff & \partial_\mu U_\omega = -\omega_\mu U_\omega \label{Uomega}
\end{eqnarray}
(with $\sigma_e^\mu = \bar \sigma^I e^\mu_I$, $\bar \sigma=(\id,-\sigma^1,-\sigma^2,-\sigma^2)$). In particular, if the spacetime dependence of $\tilde \psi$ is a ``wave packet'' strongly localised (with small width) around a classical worldline $s \mapsto x(s)$, we have $U_\omega \simeq \Teg^{- \int_0^s \omega_\mu \frac{dx^\mu}{ds}}$. Note that
\begin{eqnarray}
  \omega & =& \omega^{IJ}\mathscr D(L_{IJ}) \\
  & = & \frac{\imath}{2}\left(\begin{array}{cc} \omega^{12} & \omega^{23}-\imath \omega^{31} \\ \omega^{23}+\imath \omega^{31} & -\omega^{12} \end{array} \right) \nonumber \\
  & & \quad + \frac{1}{2}  \left(\begin{array}{cc} \omega^{03} & \omega^{01}-\imath \omega^{02} \\ \omega^{01}+\imath \omega^{02} & -\omega^{03} \end{array} \right) \label{matrixomega}
\end{eqnarray}
  
\subsection{BFSS model}
We consider a stack of $N$ D0-branes in an embedding 4D Minkowski (flat) spacetime, represented by three Hermitian operators $X^i \in \mathcal L(\mathcal H)$ (where $\mathcal H$ is a separable Hilbert space, the case $\dim \mathcal H=+\infty$ being not excluded). The reduction to 3+1 dimensions (the original BFSS model presents 9+1 dimensions) results from a truncation by taking a supersymmetric orbifold $\mathbb C^3/\mathbb Z_k$ as explained in details section II of ref. \cite{Berenstein}. In this paper, we suppose that $(X^i)$ (eventually by adding $\id_{\mathcal H}$) generates a Lie algebra $\mathfrak X$. $(X^i)$ can be assimilated to coordinates operators of a noncommutative manifold $\mathscr M$ (the $C^*$-algebra of the observables of $\mathscr M$ is generated by $(X^i)$). From the viewpoint of string theory, $\mathscr M$ is a D2-brane formed by the stack. Intuitively $X^i = \left(\begin{array}{cc} x^i_1 & s^i_{12} \\ s^i_{12} & x^i_2 \end{array} \right)$ represents a stack of two D0-branes of coordinates $x_1$ and $x_2$ in the embedding space linked by a bosonic string of oscillation radii $\{|s^i_{12}|\}_i$. We consider a massless fermionic string linking $\mathscr M$ to a probe D0-brane described by a spinor: $|\psi \rrangle = |0 \rangle \otimes |\psi^{0} \rangle + |1 \rangle \otimes |\psi^{1} \rangle \in \mathbb C^2 \otimes \mathcal H$. We can interpret the component $\psi^\alpha_a$ as $\psi^\alpha(x_a)$ where $x_a$ is the pseudo-position of the $a$th D0-brane in the embedding space (the spatial delocalisation of quantum point particle is replaced by the quantum superposition of $\dim \mathcal H$ attachment points for the string). $(|0 \rangle,|1\rangle)$ are the spin states. $\mathscr M$ can be then viewed as a quantised space (in the BFSS model the time is not quantised). The fermionic string state obeys to the following noncommutative Dirac equation \cite{Berenstein,Viennot4}:
\begin{equation}\label{BFSSeq}
  \imath \dot \Psi = \sigma_i [\mathbf X^i,\Psi]
\end{equation}
with $\mathbf X^i = \left(\begin{array}{cc} X^i & 0 \\ 0 & x^i \end{array}\right) \in \mathcal L(\mathcal H\oplus\mathbb C)$, $\Psi=\left(\begin{array}{cc} 0 & |\psi\rrangle \\ \llangle \psi| & 0 \end{array}\right) \in \mathcal L(\mathcal H\otimes \mathbb C^2 \oplus \mathbb C)$, which is equivalent to
\begin{equation}\label{NCDiraceq}
  \imath |\dot \psi \rrangle = \slashed D_x |\psi \rrangle \qquad \text{with } \slashed D_x = \sigma_i \otimes (X^i - x^i)
\end{equation}
where $x$ is the pseudo-position of the probe D0-brane in the embedding space. Formally this equation can be viewed as the space quantisation of eq.(\ref{Diraceq}-\ref{reducedDiraceq}) for a flat spacetime ($e^a_i = \delta^a_i$ and $\omega_\mu=0$). 

\subsection{Emergent geometry}\label{emergeom}
As explained in the introduction, we want to define an emergent geometry by invoking a generalisation of the notion of coherent state for $\mathscr M$. This state will constitute the foundation of the adiabatic representation of the dynamics of $|\psi\rrangle$. We start then by considering the eigenequation associated with $\slashed D_x$.\\ 
Let $M_\Lambda = \{x \in \mathbb R^3, \text{s.t. } \det\slashed D_x = 0\}$ and $|\Lambda(x)\rrangle \in \mathbb C^2 \otimes \mathcal H$ be the quasi-coherent state of $\mathscr M$, i.e. the state of the probe fermionic string such that
\begin{equation} \label{NCEVeq}
  \forall x \in M_\Lambda, \quad \slashed D_x |\Lambda(x) \rrangle = 0
\end{equation}
Eq. (\ref{NCEVeq}) can be rewritten as $\sigma_i \otimes X^i|\Lambda(x)\rrangle = E(x)|\Lambda(x)\rrangle$ with $E(x)=\sigma_ix^i$ which can be viewed as a noncommutative eigenequation in the sense where the $E(x)$ is a ``noncommutative eigenvalue'' (matrix-valued eigenvalue) of $\sigma_i \otimes X^i$. Since the solutions of the noncommutative eigenequation define the classical manifold $M_\Lambda$, this one can be viewed as an ``eigenmanifold'' of $\mathscr M$. Indeed:
\begin{eqnarray}
  \slashed D_x^2 & = & \delta_{ij} (X^i-x^i)(X^j-x^j) + \imath {\varepsilon_{ij}}^k \sigma_k \otimes X^i X^j \\
   & = & |X-x|^2 + \frac{\imath}{2} {\varepsilon_{ij}}^k \sigma_k \otimes [X^i,X^j] \label{D2}
\end{eqnarray}
$|X-x|^2$ is the operator measuring the square distance between the probe D0-brane and $\mathscr M$, and for a separable state $|\psi \rrangle = |s\rangle\otimes|\phi \rangle$ we have the Heisenberg uncertainty relation:
\begin{equation}
  {\varepsilon_{ij}}^k\langle s|\sigma_k|s\rangle \Delta_\phi X^i \Delta_\phi X^j \geq \frac{1}{2} |\llangle \psi |{\varepsilon_{ij}}^k \sigma_k \otimes [X^i,X^j]|\psi \rrangle|
\end{equation}
More precisely, ref. \cite{Schneiderbauer} shows that $|\Lambda \rrangle$ minimises the \textit{displacement energy}, intuitively the ``tension energy'' of the probe fermionic string which increases if the probe D0-brane is moved away from $\mathscr M$ or if the \textit{dispersion} $\sum_i(\Delta_{\Phi} X^i)^2$ is large:
\begin{equation}
  (\Delta_\Phi X^i)^2 = \llangle \Phi|(X^i)^2|\Psi\rrangle-\llangle \Phi|X^i|\Phi\rrangle^2
\end{equation}
$|\Lambda \rrangle$ is then the state for which the probe D0-brane ``runs'' onto $M_\Lambda$ (the probe fermionic string - the test particle - reveals the geometry). Moreover,
\begin{equation}
  \frac{1}{2}\{\sigma_i,\slashed D_x\} = \delta_{ij}(X^j-x^j) \Rightarrow \llangle \Lambda(x)|X^i|\Lambda(x) \rrangle = x^i
\end{equation}
The ``location'' of the probe D0-brane on $M_\Lambda$ indicates then the mean value of the ``location of $\mathscr M$'' . From the other side, $|\Lambda(x)\rrangle$ is also the state of a fermionic string which is close to a point particle ($|X-x|^2$ is small) strongly localised at the point $x$ (the dispersion $(\Delta_\Phi X^i)^2$ is small). To reveal the geometry at the microscopic scale, it is obvious that we need to move a strongly localised point particle.\\
For these reasons, $M_\Lambda$ is the classical manifold closest to $\mathscr M$ (it is both the ``mean value'' of $\mathscr M$ and the manifold associated with the state minimising the quantum uncertainties of the space observables). It is then the space manifold of the emergent geometry at the microscopic scale (in the sense defined in introduction, i.e. in the quasi-coherent picture).\\

As previously explained, we think that the adiabatic regime is relevant to analyse the emergent gravity of matrix models. We start here by considering the main ingredients of the (strong) adiabatic transport. Let $(u^a)$ be a local coordinates system onto  $M_\Lambda$ (and $u \mapsto x(u)$ be a parametrisation of $M_\Lambda$ embedded in $\mathbb R^3$). We consider the slow transports of the probe D0-brane along $M_\Lambda$: $t\mapsto u(t)$. The adiabatic solutions of eq. (\ref{NCDiraceq}) are local sections of the line bundle associated with a $U(1)$-bundle $\mathcal P_\Lambda$ over $M_\Lambda$ endowed with a connection described by the gauge potential $A = -\imath \llangle \Lambda|d|\Lambda \rrangle \in \Omega^1(M_\Lambda,\mathbb R)$ (where $d$ stands for the exterior derivative onto $M_\Lambda$) and the local curvature $F = dA \in \Omega^2(M_\Lambda,\mathbb R)$ (as introduced in \cite{Berenstein}). Or in other words
\begin{eqnarray}
  A_a(u) & = & -\imath \llangle \Lambda|\partial_i|\Lambda\rrangle_{|x=x(u)} \frac{\partial x^i}{\partial u^a} \\
  F_{ab}(u) & = & -\imath \left[(\partial_i \llangle \Lambda|)(\partial_j|\Lambda \rrangle) - (\partial_j \llangle \Lambda|)(\partial_i|\Lambda \rrangle)\right]_{|x=x(u)} \frac{\partial x^i}{\partial u^a} \frac{\partial x^j}{\partial u^b}\\
  & = & - \imath \llangle \Lambda| \left[\frac{\partial P_\Lambda}{\partial u^a},\frac{\partial P_\Lambda}{\partial u^b}\right]|\Lambda \rrangle
\end{eqnarray}
where $\frac{\partial x^i}{\partial u^a}$ stands for the partial derivative of the parametrisation $u \mapsto x^i(u)$ and with $P_\Lambda=|\Lambda \rrangle \llangle \Lambda|$.\\
For a slow transport $t \mapsto u(t)$ describing a closed path $\mathcal C$ on $M_\Lambda$ starting and ending at $u_0$, the adiabatic solution of eq. (\ref{NCDiraceq}) is
\begin{eqnarray}\label{psigeomphase}
  |\psi \rrangle & = & e^{-\imath \oint_{\mathcal C} A} |\Lambda(x(u_0))\rrangle \\
  & = & e^{-\imath \int_{\mathcal S} F} |\Lambda(x(u_0)) \rrangle
\end{eqnarray}
where $\mathcal S$ is the surface of $M_\Lambda$ with $\mathcal C$ as border. $F$ defines a symplectic 2-form of $M_\Lambda$, and the bivector $\tilde \theta=F^{-1}$ (where the inverse denotes the matrix inverse) defines a Poisson bracket $\{f,g\}_{\tilde \theta} = \tilde \theta^{ab} \frac{\partial f}{\partial u^a} \frac{\partial g}{\partial u^b}$ ($\forall f,g \in \mathcal C^1(M_\Lambda)$).

$M_\Lambda$ is naturally endowed with the metric of its embedding in $\mathbb R^3$: $\gamma_{ab} = \frac{\partial x^i}{\partial u^a} \frac{\partial x^j}{\partial u^b}\delta_{ij}$. Moreover note that $\gamma_{ab} = \llangle \partial_a \Lambda|\slashed D_x^2|\partial_b \Lambda \rrangle$, indeed
\begin{eqnarray}
  \slashed D_x|\Lambda\rrangle = 0 \Rightarrow  \slashed D_x \partial_a |\Lambda\rrangle & = & \frac{\partial x^i}{\partial u^a}\sigma_i|\Lambda\rrangle\\
  \Rightarrow  \llangle \partial_a \Lambda|\slashed D_x^2|\partial_b \Lambda \rrangle & = & \llangle \Lambda|\sigma_i \sigma_j|\Lambda \rrangle \frac{\partial x^i}{\partial u^a}\frac{\partial x^j}{\partial u^b} \\
  & = & \delta_{ij} \frac{\partial x^i}{\partial u^a}\frac{\partial x^j}{\partial u^b} \\
  & = & \gamma_{ab}
\end{eqnarray}
We can remark that $\llangle \Lambda|\sigma^i|\Lambda\rrangle \partial_i$ is a normal vector to $M_\Lambda$.\\
The metric $\gamma_{ab}$ is the second element of the emergent geometry, implying emergent gravitational effects (in the quasi-coherent/adiabatic picture) as manifestations of the curvature of the classical space $M_\Lambda$. By eq.(\ref{D2}) we see that $\gamma_{ab} = \gamma_{ab}^{dist}+\gamma_{ab}^{NC}$ where $\gamma^{dist} = \llangle \partial_a \Lambda||X-x|^2|\partial_b \Lambda \rrangle du^a du^b$ is the quadratic variation of the mean value of the square distance observable, and where $\gamma_{ab}^{NC} = \frac{1}{4} \llangle \partial_a \Lambda| [\sigma_i,\sigma_j] \otimes [X^i,X^j]|\partial_b \Lambda \rrangle$ is the contribution of the non-commutativity of $\mathscr M$ to the emergent metric.\\

As in \cite{Klammer,Steinacker,Kunter} we can also define a Poisson structure issued from the noncommutativity of the coordinates operators $\theta^{ij}(x)=-\imath \llangle \Lambda(x)|[X^i,X^j]|\Lambda(x)\rrangle$. But in contrast with ref. \cite{Klammer,Steinacker,Kunter} we do not consider the semi-classical limit of $\theta$. To understand the role of $\theta$ in the emergent geometry defined by the quasi-coherent state, it is necessary to formalise the operations of quantisation and of classical emergence (replacing the semi-classical limit in our approach). Let $\Env(\mathfrak X)$ be the universal $C^*$-enveloping algebra of $\mathfrak X$, which plays the role in noncommutative geometry of the algebra of ``functions'' on $\mathscr M$. Let $\omega_{\Lambda,x} : \Env(\mathfrak X) \to \mathbb R$ be the normal pure state of $\Env(\mathfrak X)$ defined by $\omega_{\Lambda,x}(Y) = \tr(P_\Lambda(x)Y)$ with $P_\Lambda = |\Lambda \rrangle \llangle \Lambda|$. We can see $\omega_\Lambda : x \mapsto \omega_{\Lambda,x}$ as a map from $\Env(\mathfrak X)$ (``noncommutative functions'' of $\mathscr M$) to $\mathcal C^\infty(M_\Lambda)$ (commutative functions of $M_\Lambda$), or also $\omega_\Lambda : \mathscr M \to M_\Lambda \subset \mathbb R^3$ with $\omega_\Lambda(X^i) = x^i$ (by identifying a point of $M_\Lambda$ with its coordinates). $\omega_\Lambda$ is then the map providing the classical analogues of the quantum space observables (as their mean values in the quasi-coherent state).\\
We recall that $\Der \mathfrak X$ (the algebra of derivatives of $\mathfrak X$) plays the role of the noncommutative tangent vector fields of $\mathscr M$. The set of $\mathcal Z(\mathfrak X)$-multilinear antisymmetric maps from ($\Der \mathfrak X)^n$ to $\Env(\mathfrak X)$, $\Omega^n_\Der \mathfrak X$, plays the role of the noncommutative $n$-forms of $\mathscr M$ ($\mathcal Z(\mathfrak X)$ is the center of $\mathfrak X$). We can introduce $\omega_{\Lambda *}$ the push-forward (tangent map) and $\omega_\Lambda^*$ the pull-back (cotangent map) of $\omega_\Lambda$:
\begin{equation}
  \omega_{\Lambda *} : \begin{array}[t]{rcl} \Der \mathfrak X & \to & T\mathbb R^3_{|M_\Lambda} = TM_\Lambda \oplus NM_\Lambda \\ L & \mapsto & \tr(P_\Lambda L(X^i)) \partial_i \end{array}
\end{equation}
where $T_xM_\Lambda$ and $N_x M_\Lambda$ denote respectively the tangent space and the normal space in $\mathbb R^3$ onto $M_\Lambda$ at the point $x$,
\begin{equation}
  \omega_{\Lambda}^* : \begin{array}[t]{rcl} \Omega^1\mathbb R^3_{|M_\Lambda} & \to & \Omega^1_\Der \mathfrak X \otimes \mathcal C^\infty(M_\Lambda) \\ \eta_i dx^i & \mapsto & \eta_i \dnc X^i \end{array}
\end{equation}
where $\dnc$ is the noncommutative derivative defined by the Koszul formula: $\dnc X^i(L) = L(X^i)$. These definitions are chosen in order to $\forall \eta \in \Omega^1\mathbb R^3_{|M_\Lambda}$, $\forall L \in \Der \mathfrak X$,
\begin{equation}
  \langle \eta, \omega_{\Lambda *}(L) \rangle_{\mathbb R^3} = \omega_\Lambda\left(\langle \omega_\Lambda^*\eta,L \rangle_{\mathfrak X} \right)
\end{equation}
where $\langle \cdot,\cdot \rangle_{\mathbb R^3} : \Omega^1 \mathbb R^3 \times T\mathbb R^3 \to \mathcal C^\infty(\mathbb R^3)$ and $\langle \cdot,\cdot \rangle_{\mathfrak X} : \Omega^1_\Der \mathfrak X \times \Der \mathfrak X \to \Env(\mathfrak X)$ are the duality brackets. $\omega_\Lambda^* : dx^i \to \dnc X^i$ can be viewed as the quantisation map, so its dual $\omega_{\Lambda *}$ should be the ``classical geometry emergence map'' as the reverse operation, except that it provides vectors not tangent to $M_\Lambda$. It is then necessary to introduce a projection. Let
\begin{equation}
  \pi_\Lambda : \begin{array}[t]{rcl} T\mathbb R^3_{|M_\Lambda} & \to & TM_\Lambda \\ u^i \partial_i & \mapsto & u^i \delta_{ij} \gamma^{ab} \frac{\partial x^j}{\partial u^b} \partial_a \end{array}
\end{equation}
be the orthogonal projection onto $TM_\Lambda$:
\begin{eqnarray}
  \pi_\Lambda(\partial_a) & = & \delta_{ij} \gamma^{cb}\frac{\partial x^i}{\partial u^b} \frac{\partial x^j}{\partial u^a}\partial_c \\
  & = & \gamma^{cb} \gamma_{ba} \partial_c \\
  & = & \partial_a
\end{eqnarray}
and
\begin{equation}
  \pi_\Lambda(\llangle \Lambda|\sigma^i|\Lambda\rrangle \partial_i) = \llangle \Lambda|\sigma_i|\Lambda\rrangle \frac{\partial x^i}{\partial u^a} \gamma^{ab} \partial_b = 0
\end{equation}
(we recall that $\llangle \Lambda|\sigma^i|\Lambda\rrangle \partial_i \in NM_\Lambda$).
Let
\begin{equation}
  \pi_\Lambda^* : \begin{array}[t]{rcl} \Omega^1M_\Lambda & \to & \Omega^1\mathbb R^3_{|M_\Lambda} \\ \eta_a du^a  & \mapsto & \eta_a \gamma^{ab} \delta_{ij} \frac{\partial x^j}{\partial u^b} dx^i \end{array}
\end{equation}
be the dual map of the projection: $\forall \eta \in \Omega^1M_\Lambda$, $\forall u \in T\mathbb R^3_{|M_\Lambda}$,
\begin{equation}
  \langle \eta, \pi_\Lambda u \rangle_{\mathbb R^3} = \langle \pi_\Lambda^* \eta, u \rangle_{\mathbb R^3}
\end{equation}
We have then $\pi_\Lambda \omega_{\Lambda *} : \Der \mathfrak X \to TM_\Lambda$ and $\omega^*_\Lambda \pi^*_\Lambda : \Omega^1M_\Lambda \to \Omega^1_\Der \mathfrak X \otimes \mathcal C^\infty(M_\Lambda)$ the maps relating the (co)tangent spaces of $\mathscr M$ and $M_\Lambda$.\\ 
We return now to the role of the Poisson structure. Let $\Theta \in \Der^2 \mathfrak X$ be the fundamental biderivative of $\mathfrak X$ defined by the commutator: $\Theta(Y,Z) = -\imath [Y,Z]$. By extension of $\omega_{\Lambda *}$ onto $\Der^2 \mathfrak X$, the pull-back of $\Theta$ defines the following bivector of $\mathbb R^3$, $\theta \in (T\mathbb R^3)^{\otimes 2}_{|M_\Lambda}$:
\begin{equation}
  \theta = \omega_{\Lambda *} \Theta = \tr(P_\Lambda \Theta(X^i,X^j)) \partial_i \otimes \partial_j = -\imath \tr(P_\Lambda[X^i,X^j]) \partial_i \otimes \partial_j
\end{equation}
Let $f:x \mapsto f_0+f_ix^i$ be a linear function, we have $f(X)=f_0 \id_{\mathcal H}+f_iX^i \in \mathfrak X$, $\omega_{\Lambda,x}(f(X)) = f(x)$ and $\omega_{\Lambda}(L(f(X))) = (\omega_{\Lambda *}L)f$. Let $f$ and $g$ be two linear functions, we have
\begin{equation}
  -\imath \omega_\Lambda([f(X),g(X)]) = \theta^{ij} \partial_i f \partial_j g = \{f,g\}_\theta
\end{equation}
$\omega_\Lambda$ transforms then the commutator $\Theta$ of $\mathfrak X$ into a Poisson bracket of $\mathbb R^3$ restricted onto linear functions of $M_\Lambda$ (note that this is not a Poisson bracket of $M_\Lambda$ since it includes derivatives in $\mathbb R^3$ and not only derivatives tangent to $M_\Lambda$)\footnote{Note that we cannot extend this correspondance to $\Env(\mathfrak X)$ and analytical functions of $M_\Lambda$. For example with $f(x)=f_0+f_ix^i+f_{ij}x^ix^j$:
  \begin{equation}
    \omega_{\Lambda,x}(f(X)) = f_0 + f_ix^i +f_{ij} \llangle \Lambda(x)|X^iX^j|\Lambda(x)\rrangle \not= f(x)
  \end{equation}
  and
  \begin{eqnarray} (\omega_{\Lambda *}L_Y)f(x) & = & f_i \llangle \Lambda(x)|[Y,X^i]|\Lambda(x)\rrangle + f_{ij} \llangle \Lambda(x)|[Y,X^i]|\Lambda(x)\rrangle x^j \nonumber \\
    & & + f_{ij}x^i \llangle \Lambda(x)|[Y,X^j]|\Lambda(x)\rrangle \\
    & \not= & f_i \llangle \Lambda(x)|[Y,X^i]|\Lambda(x)\rrangle + f_{ij} \llangle \Lambda(x)|[Y,X^i]X^j|\Lambda(x)\rrangle \nonumber \\
    & & + f_{ij} \llangle \Lambda(x)|X^i[Y,X^j]|\Lambda(x)\rrangle=\omega_{\Lambda,x}(L_Yf(X))
  \end{eqnarray}
  (with $L_Y(Z)=[Y,Z]$); and then
  \begin{equation}
    -\imath \omega_\Lambda([g(X),f(X)]) = -\imath g_i \omega_\Lambda(L_{X^i} f(X)) \not= \{g,f\}_\theta
  \end{equation}
  (with $g$ a linear function)}.
Note moreover that:
\begin{eqnarray}
  & & \frac{1}{2}\{\sigma_k,\slashed D_x^2\} = \sigma_k \otimes |X-x|^2 + \frac{\imath}{2} {\varepsilon^{ij}}_k [X^i,X^j] \\
  & \Rightarrow &  -\imath [X^i,X^j] = \varepsilon^{ijk} \sigma_k \otimes |X-x|^2-\frac{1}{2}\varepsilon^{ijk}\{\sigma_k,\slashed D_x^2\} \\
  & \Rightarrow & \theta^{ij}(x) = \varepsilon^{ijk} \llangle \Lambda(x)|\sigma_k \otimes |X-x|^2|\Lambda(x)\rrangle
\end{eqnarray}
$\pi_\Lambda \theta = \theta^{ij} \gamma^{ac} \gamma^{bd} \delta_{ik} \delta_{jl} \frac{\partial x^k}{\partial u^c} \frac{\partial x^l}{\partial u^d} \partial_a \otimes \partial_b$ is a bivector of $M_\Lambda$, but it does not define a Poisson structure on $M_\Lambda$ because it does not satisfy the Jacobi identity.\\
To summarise and compare with the semi-classical emergence approach we have:
\begin{center}
  \begin{tabular}{r|ccc}
    & \textit{quantum geometry} & & \textit{classical geometry} \\
    \hline
    \textit{semi-classical} & $(\Env(\mathfrak X),-\imath[\cdot,\cdot]) $ & $\overset{N \to +\infty}{\longrightarrow}$ & $(\mathcal C^\infty(M),\{\cdot,\cdot\})$ \\
    \textit{approach} & $(\Env(\mathfrak X),-\imath[\cdot,\cdot])$ & $\underset{\text{Moyal quantisation}}{\longleftarrow}$ &  $(\mathcal C^\infty(M),\{\cdot,\cdot\})$ \\
    \hline
    \textit{quasi-coherent} & $(\Env(\mathfrak X),\Theta) $ & $\overset{\omega_{\Lambda}}{\longrightarrow}$ & $(\mathcal P^1_{|M_\Lambda}(\mathbb R^3),\theta)$ \\
    \textit{approach} & $\Der \mathfrak X$ & $\overset{\pi_\Lambda \omega_{\Lambda *}}{\longrightarrow}$ & $TM_\Lambda$ \\
    & $\Omega^1_{\Der} \mathfrak X$ & $\underset{\omega_\Lambda^* \pi_\Lambda^*}{\longleftarrow}$ & $\Omega^1 M_\Lambda$
  \end{tabular}
\end{center}
$\mathcal P^1_{|M_\Lambda}(\mathbb R^3)$ denoting the linear functions of $\mathbb R^3$ restricted to $M_\Lambda$ and $N$ being the number of strings. The semi-classical approach focus mainly on the algebras of functions, but in the quasi-coherent approach we focus mainly on the tangent vectors fields and the differential forms. This is because a part of the emergent geometry is inherited from the geometric properties of the adiabatic bundle which are encoded in the differential forms $A$ and $F$. The discussion concerning the role of $A$ is in the next section.\\

It can be surprising to have two different structures, $\theta^{ij}$ associated with the commutator of $\mathfrak X$, and $\tilde \theta^{ab}$ associated with the Berry curvature, which seems to be not directly related. To understand this point, it is necessary to compare the present formalism with usual general relativity. In a synchronous frame, the space defined by $t=c^{st}$ is a curved three dimensional commutative manifold $M$. An ideal local Galilean observer see a flat three dimensional commutative space $\mathbb R^3$. The triads $(e^a_i)$ (spacial part of the tetrads) define a map $e:T\mathbb R^3 \to TM$ (with $e(\eta^i \partial_i) = \eta^i e^a_i \partial_a$) transforming the tangent vectors viewed by the ideal observer to true tangent vectors of the curved space. In particular the reduced Dirac equation in a curved space, eq. (\ref{reducedDiraceq}), can be rewritten as:
\begin{equation}
  \imath \partial_0 \psi - \imath \sigma^i e(\partial_i) \psi = 0
\end{equation}
and is then just the transformation of the Dirac equation in a flat space. Moreover, $(e^i_a)$ define a dual inverse map $e^{*-1}: \Omega^n \mathbb R^3 \to \Omega^n M$ (with $e^{*-1}(\eta_i dx^i) = \eta_i e^i_a du^a$) transforming infinitesimal variations viewed by the ideal observer to true infinitesimal variations of the curved space. In particular, the space metric can be rewritten as $\gamma_{ab} du^a du^b = e_a^i e_b^j \delta_{ij} du^a du^b = e^{*-1}(\delta_{ij}dx^i dx^j)$. The duality condition $\langle dx^i, \partial_j \rangle_{\mathbb R^3} = \delta^i_j$ induces by invariance that $\langle e^{*-1}(dx^i),e(\partial_j)\rangle_M = \delta^i_j \iff e^i_a e^a_j = \delta^i_j$, the dual triads $(e^a_i)$ are then the inverse matrix of the triads $(e^i_a)$.\\
In emergent geometry, $\mathfrak X$ defines a non-commutative manifold $\mathscr M$ and a target space $\mathbb R^3$; the emergent manifold $M_\Lambda$ is a curved two dimensional commutative manifold. The triads are defined by $\pi_\Lambda \omega_{\Lambda *} : \Der \mathfrak X \to TM_\Lambda$ the map transforming ``noncommutative tangent vectors'' of $\mathscr M$ to tangent vectors of the emergent manifold:
\begin{eqnarray}
& &  e^a_i \partial_a = -\imath \pi_\Lambda \omega_{\Lambda *} (L_{X^i}) \\
  & \iff & e^a_i = \delta_{il} \theta^{lj} \delta_{jk} \gamma^{ab} \frac{\partial x^k}{\partial u^b}
\end{eqnarray}
(where $L_\bullet:\mathfrak X \to \Der \mathfrak X$ denotes the Lie derivative, $L_Y(Z) = [Y,Z]$).
In particular, eq. (\ref{BFSSeq}):
\begin{equation}
  \imath \partial_0 \Psi - \sigma^iL_{\mathbf X^i} \Psi=0
\end{equation}
generates a reduced Dirac equation on $M_\Lambda$:
\begin{eqnarray}
  & & \imath \partial_0 \psi - \sigma^i \pi_{\Lambda} \omega_{\Lambda *}(L_{\mathbf X^i}) \psi = 0\\
  & \iff & \imath \partial_0 \psi - \imath \sigma^i e_i^a \partial_a \psi = 0
\end{eqnarray}
In contrast with the case of usual general relativity, the dual triads $(\tilde e^i_a)$ are not the inverse matrix of the triads $(e^a_i)$. Firstly, $\dim \mathfrak X = 3$ and $\dim M_\Lambda=2$, so $(\tilde e^i_a)$ and $(e^a_i)$ are not square matrices. Moreover, there is an important difference between inner derivatives as $L_{X^i} \in \Der \mathfrak X$ and outer derivatives as $\frac{\partial}{\partial x^i} \in T\mathbb R^3$, which can be expressed by the failure of the duality relation $\langle \mathbf dX^i, L_{X^j} \rangle_{\mathfrak X} \not= \delta^i_j$, indeed:
\begin{equation}
  \langle \mathbf dX^i, L_{X^j} \rangle_{\mathfrak X} = [X_j,X^i] \iff \omega_\Lambda(\langle \mathbf dX^i, L_{X^j} \rangle_{\mathfrak X}) = \imath \delta_{jk} \theta^{ki}
\end{equation}
So we should have $\tilde e^i_a e^a_j = \delta_{jk} \theta^{ki}$. Because of $\delta_{ip} \gamma^{cd} \frac{\partial x^p}{\partial u^d} \frac{\partial x^i}{\partial u^a} = \gamma^{cd} \gamma_{da} = \delta^c_a$ we have
\begin{eqnarray}
  & & \frac{\partial x^i}{\partial u^a} e^a_j = \delta_{jl} \theta^{li}  \\
  & \Rightarrow &  e^a_j = \delta_{jl} \theta^{li} \delta_{ip} \gamma^{ad} \frac{\partial x^p}{\partial u^d}
\end{eqnarray}
We have then $\langle \tilde e_\Lambda^* (\mathbf dX^i),\pi_\Lambda \omega_{\Lambda *}(-\imath L_{X^j}) \rangle_{M_\Lambda} = \omega_{\Lambda}(\langle \mathbf dX^i, -\imath L_{X^j} \rangle_{\mathfrak X})$, with $\tilde e_\Lambda^*(\mathbf dX^i) = \frac{\partial x^i}{\partial u^a} du^a$. It follows that the choice of defining the dual triads $(\tilde e^i_a = \frac{\partial x^i}{\partial u^a})$ to obtain the embedding metric of $M_\Lambda$ whereas the triads $(e^a_i)$ are defined with the commutation relations of $\mathfrak X$, ensures the consistency of the duality relations.

\subsection{Role of the geometric phase in matrix model}\label{geomphaseingravity}
Now we want to examine the role of the geometric phase (or of its generator $A$) in the emergent geometry, since until now we have examined only the aspects related to the quasi-coherent picture and not the aspects directly related to the adiabatic assumption. We start by a re-examination of the dynamical equations. The dynamics of the fermionic string in the strong adiabatic approximation is then $|\psi(t)\rrangle = e^{-\imath \int_0^t A_a \dot u^a dt} |\Lambda(x(u(t))) \rrangle$. By injecting this expression in eq.(\ref{NCDiraceq}) we find
\begin{equation}
  \imath \partial_0|\Lambda \rrangle - \slashed D_x |\Lambda \rrangle + A_a \dot u^a |\Lambda \rrangle = 0
\end{equation}
Consider the Dirac-Einstein equation (in the representation eliminating the Lorentz connection) with $e^0_i=0$, $e^0_0=1$ but with $e^a_0 \not=0$:
\begin{equation}
  \imath \partial_0 \psi - \imath \sigma^i e^a_i \partial_a \psi + \imath e^a_0 \partial_a \psi = 0
\end{equation}
In the WKB approximation $\psi \simeq \zeta \sqrt{\varrho} e^{\imath S}$ with $S$ the classical action, $\zeta \in \mathbb C^2$ ($\|\zeta\|=1$) and $\varrho = \psi^\dagger \psi$. If $\psi$ is strongly localised around a classical trajectory ($\partial_a \sqrt \varrho = \frac{u_a(t)-u_a}{\Delta u^2} \sqrt \varrho \simeq 0$) we have
\begin{equation}
  \imath \partial_0 \psi - \imath \sigma^i e^a_i \partial_a \psi - e^a_0 k_a \psi = 0
\end{equation}
with $k_a = \partial_a S$. By identifying $k^a$ with $\dot u^a$, the comparison of the equations provides
\begin{equation}
  e^a_0 = - A^a =  -\imath \gamma^{ab} \llangle \Lambda|\partial_b|\Lambda \rrangle
\end{equation}
The generator of the geometric phase is then the triad $(e^a_0)$. The analysis of the previous section provided only the emergent geometry for the space part ($e^a_i$). The analysis of the geometric phase provides then the emergent geometry for the time part $(e^a_0)$. To continue the analysis, it is necessary to find the associated dual tetrads and so to extend to spacetime the analysis of the (co)tangent vector fields.\\
The $C^*$-algebra of time-dependent observables of $\mathscr M$ (or of the Heisenberg representations of the observables of $\mathscr M$) is $\Env \mathfrak X \otimes \mathcal C^\infty(\mathbb R)$. Its algebra of derivatives is $\Der (\mathfrak X \otimes \mathcal C^\infty(\mathbb R)) = \Der \mathfrak X \oplus T\mathbb R$ and $\Omega^1_{\Der}(\mathfrak X \otimes \mathcal C^\infty(\mathbb R)) = \Omega^1_{\Der}\mathfrak X \oplus \Omega^1\mathbb R$. It follows that $\langle \dnc X^i,\partial_0 \rangle_{\mathfrak X \otimes \mathcal C^\infty(\mathbb R)} = \langle dx^0,L_{X^j} \rangle_{\mathfrak X \otimes \mathcal C^\infty(\mathbb R)} = 0$ and $\langle dx^0,\partial_0 \rangle_{\mathcal C^\infty(\mathbb R)} = 1$. It follows that the dual triads must satisfy
\begin{eqnarray}
  \tilde e^i_\mu e^\mu_j = \delta_{kj} \theta^{ki} & \Rightarrow & 0+ \tilde e^i_a e^a_j = \delta_{kj} \theta^{ki} \Rightarrow \tilde e^i_a = \frac{\partial x^i}{\partial u^a} \\
  \tilde e^i_\mu e^\mu_0 = 0 & \Rightarrow & \tilde e^i_0 - \tilde e^i_a A^a = 0 \Rightarrow \tilde e^i_0 = A^a \frac{\partial x^i}{\partial u^a} \\
  \tilde e^0_\mu e^\mu_j = 0 & \Rightarrow & 0 + \tilde e^0_a e^a_j = 0 \Rightarrow \tilde e^0_a = 0 \\
  \tilde e^0_\mu e^\mu_0 = 1 & \Rightarrow & \tilde e^0_0 - \tilde e^0_a A^a = 1 \Rightarrow \tilde e^0_0 = 1
\end{eqnarray}
Finally, the metric of the spacetime $\mathbb R \times M_\Lambda$ is
\begin{eqnarray}
  ds^2 & = & \tilde e^I_\mu \tilde e^J_\mu \eta_{IJ} du^\mu du^\nu \\
  & = & dt^2-(A^adt+du^a)(A^bdt+du^b) \gamma_{ab}\\
  & = & (1-A^aA^b \gamma_{ab}) dt^2 - 2 A^a \gamma_{ab} du^b dt - \gamma_{ab} du^a du^b \\
  & = & (1-A_a A_b \gamma^{ab}) dt^2 + 2 A_a du^a dt - \gamma_{ab} du^a du^b
\end{eqnarray}
which is the usual metric for a spacetime defined by a foliation of spacelike surfaces where $A^a \partial_a$ is the shift vector of the foliation ($dt A^a(u) \partial_a = \overrightarrow{P'P''}$ where $P$ is the point of coordinates $u$ on the leaf $M_\Lambda^{(t)}$, $P'$ is the point at the intersection of the normal vector at $P$ and $M_\Lambda^{(t+dt)}$; and $P''$ is the intersection of the line with constant spacial coordinates passing by $P$ and $M_\Lambda^{(t+dt)}$) \cite{Nakamura,Gourgoulhon}. $(A_a,\gamma_{ab})$ constitutes then the variables of the ADM formulation of the gravity (the lapse function is identically equal to $1$ here). This is in accordance with the fact that $A$ is the connection of the $U(1)$-principal bundle $\mathcal P_\Lambda$ describing the adiabatic transport and defines then the holonomy in this bundle. $|\Lambda \rrangle$ is a local section of the line bundle associated with $\mathcal P_\Lambda$ which is associated with the eigenmanifold $M_\Lambda$. The parallel transport of $|\Lambda \rrangle$ along an infinitesimal displacement $du$ (during the interval $dt$) is then $|\Lambda(u)\rrangle \mapsto e^{-A_adu^a} |\Lambda(u+du)\rrangle$.\\
The role of the geometric phase generator $A$ for the point of view of the non-commutative geometry is considered \ref{NCgeom}.

\subsection{Epistemological digression}\label{epistemo}
In comparison with other approaches of quantum gravity, emergent gravity presents a very special point of view. We cannot try to quantify the ``gravitational field'' ($(e^\mu_I,\omega^{IJ}_\mu)$ or gravitational waves field $h_{\mu \nu}$ in $g_{\mu \nu} = \eta_{\mu \nu}+h_{\mu\nu} +\mathcal O(h^2)$) nor the curved spacetime. In emergent gravity we start from the epistemological argument that the quantisation rules used in quantum mechanics are written from the point of view of an observer (they consist to describe the quantum \textbf{observables} as operators from the classical observables) and not from the point of view of the quantised physical system. In a classical general relativity context, this corresponds to the point of view of an ideal local Galilean observer. But such an observer seems to see a flat Minkowski spacetime (the spacetime of metric $\eta_{IJ}$). It is then epistemologically consistent to consider the quantisation of the observed flat Minkowski spacetime, with for example the rules used in the BFSS model: $x^i \leadsto X^i$ and $\partial_i \leadsto L_{X^i}$. And it can be consistent to consider, as in the BFSS ansatz, that time is not quantised since this is the one measured by the classical observer's clock and not the proper time of the test particle (the synchronous frame being defined in the neighbourhood of the observer). At a second step, micro gravitational effects result from the noncommutativity of the quantised spacetime, as viewed on the emergent curved manifold $M_\Lambda$ (``mean value of the space''): properties issued from the noncommutativity become micro gravitational effects in the sense of the Ehrenfest theorem. So gravity (at the Planck scale) is not a force, is not directly the manifestation of the spacetime curvature, but is the manifestation of spacetime noncommutative structure (the curvature emerging at the quantum averaging or at the semi-classical limit for the macroscopic scale). For this point of view, the quantisation consists well to the transformation of the Dirac operator (for massless fermions) in flat spacetime to the BFSS Dirac operator: $\imath \partial_0 - \imath \sigma^i \partial_i \leadsto \imath \partial_0 - \sigma^i L_{\mathbf X^i}$. A Dirac-Einstein operator of the emerging gravity is then obtained by application of the ``geometric emergence map'' $\pi_\Lambda \omega_{\Lambda *} : \Der \mathfrak X \to TM_\Lambda$ (which can be viewed as dual to the quantisation map $\omega_{\Lambda}^* : dx^i \mapsto \mathbf dX^i$). The emergent manifold $M_\Lambda$ being associated with an eigenvector of the BFSS Dirac operator $\slashed D_x$, it is consistent that the dynamics of the test particle (the probe D0 brane) revealing the emergent geometry be an adiabatic dynamics (it is only in the adiabatic regime that the quantum dynamics remains projected onto an instantaneous eigenvector). Adiabatic dynamics is characterised by a $U(1)$-principal bundle over $M_\Lambda$ with the connection defined by the gauge potential $A \in \Omega^1(M_\Lambda,\mathbb R)$. It is natural in these conditions to identify this $U(1)$-gauge theory with an emerging gravity gauge theory. More precisely, it is well known that quantum dynamics can be assimilated to the transport of a particle in a space (here $M_\Lambda$) where the field $F$ (usually assimilated to a virtual magnetic field) lives. The singularities of $F$ (corresponding to the level crossings) are usually assimilated to virtual Dirac magnetic monopoles. In this analogy, the classical analogue of the virtual particle is submitted to a ``pseudo Laplace force'' associated with $F$ which can be viewed here as an emergent gravity force. From the principle of the Einstein's elevator, the free falling ideal local Galilean observer does not see the gravity force, and then this one defines its frame. This can epistemologically explain that the local Galilean frame in the emergent geometry is such that $\tilde e^i_0 = A^a \frac{\partial x^i}{\partial u^a}$. The role of the gauge structure associated with the Berry connection which is unclear in \cite{Berenstein} is then explained in the present context. We will see section \ref{examples} that the interpretation of $F$ as a pseudo magnetic field in $M_\Lambda$ is also consistent in matrix model.\\
Nevertheless we can see that the adiabatic approximation does not seem to make emerge a Lorentz connection. One could imagine that this one is totally determined by the knowledge of $g_{\mu\nu}$ and $\tilde e^i_\mu/e^\mu_i$ by supposing that the geometry is torsion free (Levi-Civita connection). If it is the case at the macroscopic scale for the usual general relativity theory, nothing force that it is the case at the microscopic scale. It is usual to consider a non zero torsion in string theory \cite{Hammond, Li, Hammond2, Sarkar} for example. It can be then not impossible that the emergent geometry be not torsion free. We need then to see how a Lorentz connection can emerge from the adiabatic approach.

\section{Weak adiabatic transport and emergent geometry}
In this section we want to show that it is possible to define an emergent Lorentz connection with an adiabatic transport. In strong adiabatic transport we have a $\mathfrak u(1)$-connection defined by the generator of the geometric phase $A$. We have viewed that this one generates the shift vector of the emergent geometry. The idea of this section consists to find a $\mathfrak{gl}(2,\mathbb C)$-connection (sum of the $\mathfrak u(1)$-connection of the previous section and of a $\mathfrak{sl}(2,\mathbb C)$-connection playing the role the Lorentz connection) generalising the geometric phases. The representation of the group $GL(2,\mathbb C) \subset \mathfrak M_{2\times 2}(\mathbb C)$ is a set of spin operators. We search then an operator-valued geometric phase. In previous works \cite{Viennot1,Viennot2,Viennot3}, we have prove that such an operator-valued geometric phase arises in weak adiabatic regimes. Translated to the present context, such a regime is adiabatic with respect to the quantum space $\mathscr M$ but not with respect to spin degree of freedom (the operator valued geometric phase modifying the spin state with respect to the eigenstate); whereas the strong adiabatic regime used in the previous section is adiabatic with respect both $\mathscr M$ and the spin.

\subsection{Emergent metric and emergent Lorentz connection}
Now we want to consider a weak adiabatic regime as in \cite{Viennot1,Viennot2,Viennot3}. To this, it is necessary to change the point of view concerning the mathematical structure of the state space. Firstly, we consider $\mathbb C^2 \otimes \mathcal H$ not as an Hilbert space, but as a left Hilbert $C^*$-module over the $C^*$-algebra $\mathfrak a = \mathcal L(\mathbb C^2) = \mathfrak M_{2 \times 2}(\mathbb C)$ (spin observable algebra). The representation of $\mathfrak a$ onto $\mathbb C^2 \otimes \mathcal H$ being defined by:
\begin{equation}
  S|\psi\rrangle = S \otimes \id_{\mathcal H} |\psi \rrangle, \quad \forall S \in \mathfrak a$, $\forall |\psi \rrangle \in \mathbb C^2 \otimes \mathcal H
\end{equation}
  The inner product of the $C^*$-module ($\langle \cdot | \cdot \rangle_* : (\mathbb C^2\otimes \mathcal H)^2 \to \mathfrak a$) is defined by $\forall |\psi\rrangle,|\phi\rrangle \in \mathbb C^2 \otimes \mathcal H$:
\begin{equation}
  \langle \phi|\psi \rangle_* = \tr_{\mathcal H} |\psi \rrangle \llangle \phi|
\end{equation}
where $\tr_{\mathcal H}$ denotes the partial trace over $\mathcal H$. Eq.(\ref{NCEVeq}) can be viewed as a noncommutative eigenequation in the sense of this $C^*$-module, $|\psi\rrangle$ is eigenvector of $\sigma_i \otimes X^i$ with eigenvalue $E(x)=\sigma_ix^i \in \mathfrak a$. Note that the square $C^*$-module norm $\|\psi\|_*^2 = \tr_{\mathcal H} |\psi \rrangle \llangle \psi|$ is the density matrix (mixed state) of the fermionic string's spin (if $\|\psi\|_*^2$ is not pure, it represents the spin state entangled with states of $\mathscr M$, the noncommutative manifold playing the role of a quantum environment for the spin of the fermionic string).\\
For a transport $t \mapsto u(t)$ on $M_\Lambda$ in a weak adiabatic regime (see \ref{adiabth} and \cite{Viennot1,Viennot2,Viennot3}) the solution of eq.(\ref{NCDiraceq}) is
\begin{equation}
  |\psi(t)\rrangle = \Ted^{-\imath \int_0^t \mathfrak A_a(u(t'))\dot u^a(t')dt'} |\Lambda(x(u(t)))\rrangle
\end{equation}
where $\Ted$ denotes the time counter-ordered exponential and $\mathfrak A \in \Omega^1(M_\Lambda,\mathfrak a)$ is the operator-valued geometric phase generator defined by
\begin{equation}
  \mathfrak A \rho_\Lambda = -\imath \langle \Lambda |\partial_i|\Lambda\rangle_* \frac{\partial x^i}{\partial u^a} du^a
\end{equation}
where $\rho_\Lambda = \langle \Lambda|\Lambda\rangle_* = \tr_{\mathcal H}|\Lambda \rrangle \llangle \Lambda|$ is the density matrix of the spin of the fermionic string in the quasi-coherent state. The operator-valued geometric phase $\Ted^{-\imath \int_0^t \mathfrak A_a\dot u^adt'} \in \mathfrak a$ represents changes of spin orientation occurring during the adiabatic transport of $|\Lambda \rrangle$ viewed, in accordance with the interpretation of section \ref{emergeom}, as the transport of a strongly localised particle in the curved spacetime $M_\Lambda \times \mathbb R$ (this one does not affect the localisation shape of the particle but can induce changes of the spin state by Thomas or de Sitter precession phenomenon for example). The adiabatic solutions are local sections of a vector bundle associated with a non-abelian bundle gerbe (a categorical principal bundle) $\mathscr P_\Lambda$ over $M_\Lambda$ endowed with a 2-connection with gauge 2-potential (a gauge potential specific to bundle gerbes) $\mathfrak A$ (see \cite{Viennot1} for the details). Note that $\tr (\rho_\Lambda \mathfrak A) = A$ and then the quantum statistical mean value of $\mathfrak A$ corresponds to the discussion about the strong adiabatic regime. It is then interesting to split the gauge potential onto abelian and purely non-abelian parts: $\mathfrak A = \mathfrak A^{off} + \frac{1}{2} \tilde A$ with $\tr \mathfrak A^{off} = 0$ and $\tilde A = \tr \mathfrak A$. This splitting is associated with the central extension of groups $1 \to \mathbb C^* \to GL(2,\mathbb C) \to GL(2,\mathbb C)/\mathbb C^* = SL(2,\mathbb C) \to 1$, where $\mathfrak A^{off}$ is a gauge potential of a $SL(2,\mathbb C)$-principal bundle over $M_\Lambda$ supporting the non-abelian bundle gerbe $\mathscr P_\Lambda$ (see \cite{Mackaay,Murray}).\\

We have then
\begin{equation}
  |\psi(t)\rrangle = e^{-\frac{\imath}{2}\int_0^t \tilde A_a \dot u^a dt'} \Ted^{-\imath \int_0^t \mathfrak A_a^{off} \dot u^a dt'} |\Lambda(x(u(t)))\rrangle
\end{equation}
Along a specific path $t\mapsto u^a$, eq.(\ref{Uomega}) becomes $\dot U_\omega = -\omega_a \dot u^a U_\omega$ and then $U_\omega = \Teg^{-\int_0^t \omega_a \dot u^a dt'}$. But following eq.(\ref{reducedDiraceq}), $|\psi \rrangle= U_\omega^{-1} |\tilde \psi \rrangle$. By identification, it follows that
\begin{equation}
  U_\omega = \Teg^{-\int_0^t \omega_a \dot u^a dt'} = \left(\Ted^{-\imath \int_0^t \mathfrak A_a^{off} \dot u^a dt'}\right)^{-1} = \Teg^{\imath \int_0^t \mathfrak A_a^{off} \dot u^a dt'}
\end{equation}
and then that $\omega = -\imath \mathfrak A^{off}$. The non-abelian gauge potential of the weak adiabatic regime generates a Lorentz connection.

Finally the emergent geometry is described by a metric associated with the abelian gauge potential $\frac{1}{2}\tilde A \in \Omega^1(M_\Lambda,\mathbb R)$ and a Lorentz connection defined by a the non-abelian gauge potential $\mathfrak A^{off} \in \Omega^1(M_\Lambda,\mathfrak{sl}(2,\mathbb C))$. In \cite{Miao}, it is proved that gravity in noncommutative space needs to extend the gauge symmetry from $SL(2,\mathbb C)$ to $GL(2,\mathbb C)$, adding a $U(1)$-gauge potential to the Lorentz connection. This one is in the present model the Berry phase generator $\tilde A$. This gauge theory in $M_\Lambda$ provides then from the Lorentz connection during the central extension of the symmetry group needed by the noncommutative origin of the gravity.\\
The curvature tensor ${R^{IJ}}_{ab}$ is obtained by
\begin{equation}
  R = \frac{1}{2} {R^{IJ}}_{ab} \mathscr D(L_{IJ}) du^a \wedge du^b = d\omega + \omega \wedge \omega
\end{equation}
  and is related to the curving of the 2-connection of $\mathscr P_\Lambda$ (see \cite{Viennot1,Mackaay,Murray}): $B = d\mathfrak A - \mathfrak A \wedge \mathfrak A = -R+\frac{1}{2} \tilde F$ (where the Berry curvature $\tilde F=d\tilde A$ is the fake curvature of the 2-connection).\\ 

Since $\llangle \Lambda(x)|X^i|\Lambda \rrangle = x^i$ (with minimal dispersion $\sum_i (\Delta_\Lambda X^i)^2$), $|\Lambda(x)\rrangle$ can be viewed as a fermion state strongly localised at $x$. In this strongly localised state, $U_\omega$ appears as the evolution operator of the fermion's spin. The dynamics of the spin is then governed by the Hamiltonian $-\imath \omega_a \dot u^a$. This is in accordance with previous results concerning the dynamics of localised qubit in curved spacetimes \cite{Viennot5} (a detailed study of this spin dynamics can be found in \cite{Viennot5}). The localised qubit model comes from a WKB analysis. We find the heuristic argument of the introduction, the (weak) adiabatic approximation at the spacetime microscopic scale (i.e. in a quasi-coherent state) is similar to the semi-classical limit at the spacetime macroscopic scale.\\

At this stage, the details of the emergent geometry depends on the separable or entangled nature of $|\Lambda \rrangle$.

\subsubsection{Separable quasi-coherent state case:}
In this section we suppose that the quasi-coherent state is separable: $|\Lambda(x)\rrangle = |0_x\rangle \otimes |\lambda_0(x)\rangle \in \mathbb C^2 \otimes \mathcal H$. In that case $\rho_\Lambda = |0_x \rangle \langle 0_x|$ is not invertible, and is its own pseudo-inverse. We have then
\begin{equation}
  \mathfrak A = -\imath |d0_x\rangle \langle 0_x| - \imath \langle \lambda_0|d|\lambda_0 \rangle |0_x \rangle \langle 0_x|
\end{equation}
(due to the non-invertible character of $\rho_\Lambda$ this expression is a specific gauge choice of $\mathfrak A$, see \cite{Viennot1} for a discussion about the gauge changes of gauge 2-potentials). It follows that
\begin{eqnarray}
  \tilde A = \tr \mathfrak A & = & -\imath \langle 0_x|d|0_x\rangle - \imath \langle \lambda_0|d|\lambda_0\rangle \\
  & = & A
\end{eqnarray}
  which defines then the same metric than the strong adiabatic regime.\\

Let $|1_x\rangle \in \mathbb C^2$ be a state orthogonal to $|0_x\rangle$ (forming both a basis of $\mathbb C^2$).
\begin{eqnarray}
  \mathfrak A^{off} & = &  \mathfrak A - \frac{1}{2}A \\
  & = & \frac{\imath}{2} \llangle \Lambda|d|\Lambda\rrangle(|1_x\rangle\langle 1_x|-|0_x\rangle\langle 0_x|) -\imath \langle 1_x|d|0_x\rangle|1_x\rangle\langle 0_x|
\end{eqnarray}
Let $|a_x\rangle = a_x^0|0\rangle + a_x^a|1\rangle$ be the decomposition of the $x$-dependent basis into the canonical basis of $\mathbb C^2$. We choose $|1_x\rangle = \bar 0_x^1|0\rangle - \bar 0_x^0|1\rangle$. In this canonical basis we have
\begin{eqnarray}
  \mathfrak A^{off} & = & -\frac{\imath}{2} \left(\begin{array}{cc} \bar 0_x^0 d0_x^0-\bar 0_x^1d0_x^1 & 2\bar 0_x^1 d0_x^0 \\ 2\bar 0_x^0 d0_x^1 & -\bar 0_x^0d0_x^0+\bar 0_x^1d0_x^1 \end{array} \right) \nonumber \\
  & & \quad -\frac{\imath}{2}\langle \lambda_0|d|\lambda_0\rangle \left(\begin{array}{cc} |0_x^0|^2-|0_x^1|^2 & 20_x^0\bar 0_x^1 \\ 2\bar 0_x^00_x^1 & |0_x^1|^2-|0_x^0|^2 \end{array} \right) \\
  & = & -\frac{\imath}{2}\left(\begin{array}{cc} \llangle \Lambda|\sigma^3d|\Lambda\rrangle & \llangle \Lambda|(\sigma^1-\imath \sigma^2)d|\Lambda \rrangle \\ \llangle \Lambda|(\sigma^1+\imath \sigma^2)d|\Lambda \rrangle & -\llangle \Lambda|\sigma^3d|\Lambda\rrangle \end{array} \right)
\end{eqnarray}
It follows that the Lorentz connection is
\begin{eqnarray}
  \omega^{ij}_a & = & - {\varepsilon^{ij}}_k \Re \llangle \Lambda|\sigma^k\partial_a|\Lambda \rrangle \\
  \omega^{0i}_a & = & \Im \llangle \Lambda|\sigma^i\partial_a|\Lambda \rrangle
\end{eqnarray}

\subsubsection{Entangled quasi-coherent state case:}
Now we suppose that the quasi-coherent state is entangled and we write $|\Lambda \rrangle = |0\rangle \otimes |\Lambda^0 \rangle + |1\rangle \otimes |\Lambda^1 \rangle = (|\Lambda^0 \rangle\ |\Lambda^1\rangle)$ in the canonical basis of $\mathbb C^2$. We have then
\begin{eqnarray}
  \rho_\Lambda & = & \left(\begin{array}{cc} \langle \Lambda^0|\Lambda^0\rangle & \langle \Lambda^1|\Lambda^0 \rangle \\ \langle \Lambda^0|\Lambda^1 \rangle & \langle \Lambda^1|\Lambda^1 \rangle \end{array} \right) \\
  \rho_\Lambda^{-1} & = & \frac{1}{|\Lambda|} \left(\begin{array}{cc} \langle \Lambda^1|\Lambda^1\rangle & -\langle \Lambda^1|\Lambda^0 \rangle \\ -\langle \Lambda^0|\Lambda^1 \rangle & \langle \Lambda^0|\Lambda^0 \rangle \end{array} \right)
\end{eqnarray}
with $|\Lambda|=\det \rho_\Lambda \not=0$ because $|\Lambda\rrangle$ is entangled. It follows that
\begin{eqnarray} \label{Aentangled}
  \mathfrak A & = & -\imath \tr_{\mathcal H}|d\Lambda\rrangle \llangle \Lambda| \rho_\Lambda^{-1} \\
  & = & -\imath \left(\begin{array}{cc} \langle \Lambda^0_*|d|\Lambda^0\rangle & \langle \Lambda_*^1|d|\Lambda^0 \rangle \\ \langle \Lambda_*^0|d|\Lambda^1 \rangle & \langle \Lambda_*^1|d|\Lambda^1 \rangle \end{array} \right)
\end{eqnarray}
where the dual quasi-energy state $\llangle \Lambda_*|$ is defined by
\begin{eqnarray}
  \langle \Lambda_*^0| & = & \frac{1}{|\Lambda|} (\langle \Lambda^1|\Lambda^1\rangle \langle \Lambda^0| - \langle \Lambda^0|\Lambda^1\rangle \langle \Lambda^1|)\\
  \langle \Lambda_*^1| & = & \frac{1}{|\Lambda|} (-\langle \Lambda^1|\Lambda^0\rangle \langle \Lambda^0| + \langle \Lambda^0|\Lambda^0\rangle \langle \Lambda^1|)
\end{eqnarray}
Note that $(\langle \Lambda_*^0|,\langle \Lambda_*^1|)$ is bi-orthorgonal to $(|\Lambda^0\rangle, |\Lambda^1\rangle)$, i.e. $\langle \Lambda_*^\alpha|\Lambda^\beta\rangle = \delta^{\alpha \beta}$. Eq.(\ref{Aentangled}) is then similar to a non-abelian geometric phase generator associated with a non-hermitian Hamiltonian (see for example \cite{Leclerc}, and see \cite{Viennot5} for an application to a localised spin adiabatically transported in a curved spacetime). The difference is that here the dissipative process onto the spin is not induced by the non-selfadjointness of the spin Hamiltonian but by the entanglement between the spin with the noncommutative manifold.\\

It follows that the metric is defined with $\frac{1}{2} \tilde A = \frac{1}{2} \tr \mathfrak A = -\frac{\imath}{2} \llangle \Lambda_*|d|\Lambda \rrangle$. When the density matrix is the microcanonical distribution $\rho_\Lambda = \frac{1}{2} \id$ ($|\Lambda\rrangle$ is maximally entangled) then $\frac{1}{2} \tilde A = A$ and the metric is the same than in the strong adiabatic regime.
$\mathfrak A^{off} = \mathfrak A - \frac{1}{2} \tilde A$ and by comparison with eq.(\ref{matrixomega}) the Lorentz connection is
\begin{eqnarray}
  \omega^{ij}_a & = & - {\varepsilon^{ij}}_k \Re\llangle \Lambda_*|\sigma^k\partial_a|\Lambda \rrangle \\
  \omega^{0i}_a & = & \Im \llangle \Lambda_*|\sigma^i \partial_a|\Lambda \rrangle
\end{eqnarray}

The expressions in the two cases are totally similar, with in the separable case $\llangle \Lambda_*|\equiv\llangle \Lambda|$.  

\subsection{Emergent geometry: Christoffel symbols and torsion}\label{torsion}
The tetrads of the emergent geometry have be found in section \ref{geomphaseingravity} and the Lorentz connection have be found in the previous section. We are now able to compute the Christoffel symbols of the emergent geometry:
\begin{eqnarray}
  \Gamma^\mu_{0\nu} & = & 0 \\
  \Gamma^0_{b0} & = & \omega^0_{bj} \tilde e^j_0 \\
   \Gamma^a_{b0} & = & \Xi^a_e\left(e^e_i \partial_b \tilde e^i_0 + e^e_i \omega^i_{b0} + e^e_i \omega^i_{bj} \tilde e^j_0\right) \\
  \Gamma^0_{bc} & = & \omega^0_{bj} \tilde e^j_c \\
  \Gamma^a_{bc} & = & \Xi^a_e \left(e^e_i \partial_b \tilde e^i_c + e^e_i \omega^i_{bj} \tilde e^j_c \right)
\end{eqnarray}
where $\Xi$ is the inverse matrix of $(e^a_i \tilde e^i_b)$ (since $e^a_i \tilde e^i_b \not= \delta^a_b$, it is necessary to renormalise the product between triads and dual triads) and with
\begin{equation}
  \begin{array}{ll}
    \tilde e^i_a = \frac{\partial x^i}{\partial u^a} & e^a_i = {\varepsilon_{ij}}^l \llangle \Lambda|\sigma_l \otimes |X-x|^2|\Lambda \rrangle \gamma^{ab} \frac{\partial x^j}{\partial u^b} \\
    e^a_0 = \frac{\imath}{2} \gamma^{ab} \llangle \Lambda_*|\partial_b|\Lambda \rrangle & \tilde e^i_0 = \frac{\imath}{2} \gamma^{ab} \llangle \Lambda_*|\partial_b|\Lambda \rrangle \frac{\partial x^i}{\partial u^a}
  \end{array}
\end{equation}
We see that the emergent geometry is in general not torsion free: $T^\alpha_{\beta \gamma} = \Gamma^\alpha_{\beta \gamma} - \Gamma^\alpha_{\gamma \beta} \not=0$. More precisely, since ${\varepsilon^{ij}}_k \sigma^k = \frac{1}{2 \imath}[\sigma^i,\sigma^j]$ we have
\begin{eqnarray}
  T^i_{bc} & = & \frac{1}{2}\mathrm{Im} \tr \left([\sigma^i,\tau_c] \mathfrak A_b - [\sigma^i,\tau_b] \mathfrak A_c \right) \label{torsion1}\\
  & = & \frac{1}{2} \mathrm{Im} \tr\left(\sigma^i([\tau_c,\mathfrak A_b]-[\tau_b,\mathfrak A_c])\right) \label{torsion2}\\
  T^0_{bc} & = & \mathrm{Im} \tr\left(\tau_b \mathfrak A_c - \tau_c \mathfrak A_b \right) \label{torsion3}\\
  T^0_{b0} & = & -\frac{1}{2} \mathrm{Im}\tr (\underline A \mathfrak A_b) \label{torsion4}\\
  T^i_{b0} & = & \frac{1}{2} \partial_b A^i - \mathrm{Im} \tr(\sigma^i \mathfrak A_b) + \frac{1}{4} \mathrm{Im} \tr\left([\sigma^i,\underline A]\mathfrak A_b\right) \label{torsion5}
\end{eqnarray}

with $\tau_a \equiv \sigma_i \frac{\partial x^i}{\partial u^a}$, $\underline A \equiv \gamma^{ab} A_a \tau_b$, $A^i \equiv \gamma^{ab} A_a \frac{\partial x^i}{\partial u^b}$ and $T^a_{\beta \gamma} = \Xi^a_e e^e_i T^i_{\beta \gamma}$.\\
Since these formulae are strongly dependent on the spin degree of freedom of the fermionic string, we can think that the origin of this non-zero torsion is the same than in Einstein-Cartan theory. By considering $\Gamma^a_{b0} = T^a_{b0}$, we see that this part of the torsion essentially depends on two elements: the Lorentz connection (which is by construction related to the effects of the spin of the fermionic string) and derivatives of $\tilde A$ (which permits to think that the torsion is related to $\tilde F$). In section \ref{examples}, we will see that the effects of the torsion in some examples are consistent with the usual interpretation of the Berry curvature as a pseudo magnetic field in $M_\Lambda$. In \cite{Steinacker2}, a torsion on a quantum spacetime described by the IKKT matrix model is considered at the semi-classical limit. A comparison between this torsion and the one of the present work can found \ref{NCtorsion}.\\

The extrinsic curvature of $M_\Lambda$ (see \cite{Nakamura,Gourgoulhon}) is $K_{ab} = \frac{1}{2}(\partial_a A_b - \Gamma^c_{ba} A_c + \partial_b A_a - \Gamma^c_{ab} A_c$). And in a same way, we can compute the Riemann tensor with:
\begin{eqnarray}
  {\mathcal R^0}_{0cd} & = & {R^0}_{jcd} \tilde e^j_0 + {R^0}_{0cd}\\
  {\mathcal R^0}_{bcd} & = & {R^0}_{jcd} \tilde e^j_b\\
  {\mathcal R^a}_{0cd} & = & e^a_0{R^0}_{jcd}\tilde e^j_0+ \Xi^a_e(e^e_i{R^i}_{0cd}+e^e_i{R^i}_{jcd}\tilde e^j_0)+e^a_0{R^0}_{0cd})\\
  {\mathcal R^a}_{bcd} & = & e^a_0{R^0}_{jcd}\tilde e^j_b+\Xi^a_ee^e_i{R^i}_{jcd}\tilde e^j_b
\end{eqnarray}
with ${R^I}_{Jcd} = {\left[\partial_c \omega_d - \partial_d \omega_c + \frac{1}{2}[\omega_c,\omega_d] \right]^I}_J$.

\subsection{Generalisation}
Until now we have considered an algebra $\mathfrak X$ of three dimensions, involving an eigenmanifold $M_\Lambda$ of two dimensions. To obtain a three dimensional eigenmanifold it is necessary to consider an algebra with four generators. Moreover we have considered only time independent generators $\{X^i\}$, but in some cases these operators can be time dependent with fast evolutions hampering the use of an adiabatic limit. In fact these two questions are related. To apply an adiabatic limit with fast evolutions, we use the Schr\"odinger-Koopman approach \cite{Viennot6} in which the degrees of freedom of the fast evolution becomes new quantum variables associated with a new observable, which is here a fourth coordinate operator $X^4$. With this one, we can have a three dimensional eigenmanifold. It is natural than $X^4$ emerges from the dynamics of $\{X^i\}_{i=1,2,3}$ without introducing another operator. We start with a quantum space $\mathscr M$ of three dimensions (generated by $\{X^i\}_{i=1,2,3}$), which defines a three dimensional classical space $M_\Lambda$, the fourth operator $X^4$ resulting from the inner dynamics of $\mathscr M$. In this approach, the time dimension is split into a time of the slow evolution which remains classical, and a time of the fast evolution which is quantised as $X^4$. So the signature of the new embedding spacetime is $(+,-,-,-,+)$. As interpreted in section \ref{epistemo}, $t$ the time of slow evolutions can be viewed as the time indicated by the clock of the ideal Galilean observer. $X^4$ as the manifestation of the fast inner evolution of $\mathscr M$ can be interpreted as a quantum time observable.\\

We have considered only massless fermions in the previous sections. With massive fermions, in the Weyl representation considered here, the evolution of the state $|\psi\rrangle$ is submitted to fast chiral oscillations. We can also used the same approach to treat this fast oscillations by introducing a quantum variable $\theta$ on the oscillation phase circle $\mathbb S^1$ and an associated fourth operator $-\imath \omega \partial_\theta$ (where $\omega/2$ is the mass) which can be interpreted as the coordinate operator on a compact dimension.\\

These two generalisations are treated in details in \ref{higherdim}. An interesting fact with the generalisation to massive fermions, is that the approach deals not only with the ground quasi-energy state, solution of $\slashed D_x |\Lambda_0\rrangle = 0$, as in the previous sections, but also with generalised excited states, solutions of $\slashed D_x |\Lambda_n \rrangle = n\omega |\Lambda_n \rrangle$, where $n \in \mathbb N$ defines a mode of chiral oscillations. As $|\Lambda_0\rrangle$ these states define classical geometries $(M_{\Lambda,n},g_n,\omega_n)$ (eigenmanifold, spacetime metric, Lorentz connection) in the same manner that in the previous sections. We can interpret $(M_{\Lambda,n},g_n,\omega_n)$ as the geometry ``warped'' by the mass $\omega$ of the test particle when this one is in the mode $n$ of chiral oscillations. The justification of the use of the excited states can be found in \ref{chiral}.

\section{Examples}\label{examples}
Now we apply the formalism developed in the previous sections to specific examples of Lie algebras $\mathfrak X$. The structure of $\mathfrak X$ is defined by the contents of the spacetime in the neighbourhood of the local Galilean observer, but this is not the subject of this paper, where we suppose that $\mathfrak X$ is previously known.\\

The noncommutative manifold $\mathscr M$ defines a quantum spacetime (with no quantisation of time). To compare it with the usual theory of gravity (and then with classical geometries), the quasi-coherent state, as being the quantum state of $\mathscr M$ closest to a classical state (since it minimises the quantum dispersion), provides the classical geometry closest to $\mathscr M$. This geometry is described by three classical entities:
\begin{itemize}
\item the eigenmanifold $M_\Lambda$, described by the set of noncommutative eigenvalues $E(x) = x^i \sigma_i$: $M_\Lambda = \{x \in \mathbb R^3, \det(\sigma_i \otimes X^i - E(x))=0\}$;
\item the metric $g_{\mu \nu}$ on $M_\Lambda$, provided by the derivatives of the associated eigenvector (the quasi-coherent state):
  \begin{eqnarray}
    g_{00} & = & 1-\hat A_a\hat A_b \gamma^{ab} \qquad \text{with } \hat A_a = \llangle \Lambda_*|\partial_a\Lambda \rrangle \\
    g_{0a} & = & \hat A_a \\
    g_{ab} & = & \gamma_{ab} = \llangle \partial_a \Lambda|\slashed D_x^2|\partial_b \Lambda \rrangle
  \end{eqnarray}
\item the Lorentz connection $\omega_a^{IJ}$ (or equivalently the Christoffel symbols $\Gamma^\mu_{\rho \nu}$), defined with $\llangle \Lambda_*|\sigma^i| \partial_a \Lambda \rrangle$.  
\end{itemize}
Since the Lorentz connection does not correspond to the Levi-Civita connection associated with $g_{\mu \nu}$, the geometry is not torsion free ($\Gamma^\mu_{\rho \nu} \not= \Gamma^\mu_{\nu \rho}$). This is the geometry defined by the classical quantities $(M_\Lambda,g_{\mu \nu},\omega^{IJ}_a)$ that we call emergent geometry, since it emerges from the purely quantum eigenequation $\slashed D_x |\Lambda \rrangle = 0$. In the following examples, we solve this equation for different models of $\mathscr M$ (different algebras $\mathfrak X$) to compute these classical geometric entities. In the spirit of the meaning of the adiabatic transport of the probe brane, the geometry can be revealed by the movement of a test classical particle. A useful manner to illustrate the emergent geometry $(M_\Lambda,g_{\mu \nu},\omega^{IJ}_a)$ consists then to draw the involved geodesics. Due to the non vanishing torsion $T^\mu_{\rho \nu} = \Gamma^\mu_{\rho \nu} - \Gamma^\mu_{\nu \rho}$, there are two notions of geodesics. The minimising geodesics are defined as curves on $M_\Lambda \times \mathbb R$ which are of minimal length (with respect to $g_{\mu \nu}$) between two any closed points. Since $\gamma_{ab}$ is the metric induced by the embedding of $M_\Lambda$, these ones are obvious with regard to the shape of $M_\Lambda$ in $\mathbb R^3$. The auto-parallel geodesics are defined as curves on $M_\Lambda \times \mathbb R$ such that the tangent vectors at any two infinitesimally closed points be parallel. They are computed with the non symmetric Christoffel symbols $\Gamma^\mu_{\nu \rho}$. In torsion free geometry, the two notions are the same. The comparison of the auto-parallel geodesics with the minimising geodesics provide then the direct effect of the torsion which is the main souvenir of the quantum nature of $\mathscr M$ in the emergent classical geometry of the quasi-coherent picture. The deviation of the geodesics by the torsion can be then interpreted as the irreducible spacetime quantum effect, since this one remains in the state closest to a classical one.\\

The methodology to treat a concrete example is the following. After the definition of the algebra $\mathfrak X$ defining the quantum spacetime $\mathbb R \times \mathscr M$, we compute the quantum entities $(|\Lambda\rrangle, \hat A, \mathfrak A)$ with the quasi-coherent state solution of $\slashed D_x|\Lambda \rrangle=0$ (or its generalisation for excited states in the case a massive fermion). $\hat A$ and $\mathfrak A$ are respectively geometric phase generators of the strong and weak adiabatic transports. The three kinds of data $(|\Lambda\rrangle, \hat A, \mathfrak A)$ encode the behaviour of the quantum spacetime $\mathbb R \times \mathscr M$ at the adiabatic limit. Since this behaviour is the quantum regime closest to a classical one, we can compute the associated classical entities $(M_\Lambda,g_{\mu \nu},\omega^{IJ}_a)$ and study with the auto-parallel geodesics the effects of the torsion, signature of the quantum nature of $\mathscr M$ surviving in the more classical regime.

\subsection{The noncommutative plane}\label{NCplane}
We consider the case where $\mathscr M$ is defined by $X^1 = \frac{a+a^+}{2}$, $X^2=\frac{a-a^+}{2\imath}$ and $X^3=0$ where $a$ and $a^+$ are the harmonic oscillator annihilation and creation operators ($\dim \mathcal H = +\infty$).
\subsubsection{Ground state:}
The solutions of eq.(\ref{NCEVeq}) are (see \ref{CCR})
\begin{equation}
  M_\Lambda = \mathbb R^2 = \{(x^1,x^2,0), x^1=\Re(\alpha), x^2=\Im(\alpha)\}_{\alpha \in \mathbb C}
\end{equation}
\begin{equation}
  |\Lambda(\alpha)\rrangle = |0\rangle \otimes |\alpha \rangle
\end{equation}
where $|\alpha \rangle = e^{-|\alpha|^2/2} \sum_{n=0}^{+\infty} \frac{\alpha^n}{\sqrt{n!}} |n\rangle$ is an harmonic oscillator coherent state \cite{Perelomov,Puri}. It follows that $A = -\imath \langle \alpha|\partial|\alpha \rangle - \imath \langle \alpha |\bar \partial|\alpha \rangle$ (where $\partial$ and $\bar \partial$ denote the Dolbeault derivative operators).
\begin{eqnarray}
  \frac{\partial}{\partial \bar \alpha} |\alpha \rangle & = & - \frac{\alpha}{2} |\alpha \rangle \\
  \frac{\partial}{\partial \alpha}|\alpha \rangle & = & -\frac{\bar \alpha}{2}|\alpha \rangle + e^{-|\alpha|^2/2} \sum_{n=1}^{+\infty} \frac{\alpha^{n-1}}{\sqrt{(n-1)!}} |n\rangle \\
  & = & -\frac{\bar \alpha}{2}|\alpha \rangle + a^+|\alpha \rangle
\end{eqnarray}
We have then
\begin{eqnarray}
  \langle \alpha|\bar \partial|\alpha \rangle & = & - \frac{\alpha}{2} d\bar \alpha\\
  \langle \alpha|\partial|\alpha \rangle & = & (-\frac{\bar \alpha}{2} + \langle \alpha|a^+|\alpha\rangle)d\alpha =  \frac{\bar \alpha}{2} d\alpha
\end{eqnarray}
  and finally
\begin{equation}
  A = -\imath \frac{\bar \alpha d\alpha - \alpha d\bar \alpha}{2} = x^1dx^2-x^2dx^1
\end{equation}
\begin{equation}
  F = \partial A + \bar \partial A = \frac{\imath}{2} d\alpha \wedge d\bar \alpha = dx^1 \wedge dx^2
\end{equation}
The metric of $\mathbb R \times M_\Lambda$ is then
\begin{equation}
  ds^2 = (1-(x^1)^2-(x^2)^2)dt^2+2x^1dx^2dt-2x^2dx^1dt - d(x^1)^2-d(x^2)^2
\end{equation}
Moreover we have simply $\mathfrak A = A|0\rangle \langle 0|$ and $\mathfrak A^{off} = \frac{A}{2} \sigma^3$, and the unique non-zero component of the Lorentz connection is $\omega^{03}_a = A_a$.\\
The triads and the dual triads are 
\begin{equation}
  \begin{array}{lllll}
    e^0_0=1 & e^0_i=0 & e^1_0=-x^2 & e^2_0=x^1 & e^a_i = \frac{1}{2} {\varepsilon_i}^a \\
    \tilde e^0_0=1 & \tilde e^0_a=0 & \tilde e^1_0 = x^2 & \tilde e^2_0=-x^1 & \tilde e^i_a = \delta^i_a
  \end{array}
  \end{equation}
    The non-zero Christoffel symbols are only $\Gamma^1_{20}= -\Gamma^2_{10}=1$. It follows that the (auto-parallel) geodesic equations are
\begin{equation}
  \left\{ \begin{array}{l} \ddot t=0 \\ \ddot x^1+\dot x^2 \dot t=0 \\ \ddot x^2-\dot x^1 \dot t=0 \end{array} \right.
\end{equation}
(where the dots denote here the derivative with respect to the proper time $s$). It follows that
\begin{equation}
  \left\{ \begin{array}{l} 
    t(s)=\beta s \\
    x^1(s)=\frac{v_0}{\beta}\sin(\beta s+\varphi)+x^1(0)-\frac{\dot x^2(0)}{\beta} \\
    x^2(s)=-\frac{v_0}{\beta}\cos(\beta s+\varphi)+x^2(0)-\frac{\dot x^1(0)}{\beta}
  \end{array} \right.
  \end{equation}
with $v_0=\sqrt{\dot x^1(0)+\dot x^2(0)}$ and $\tan \varphi= \frac{\dot x^2(0)}{\dot x^1(0)}$. The geodesics on $M_\Lambda$ are then circles: the effect of the non-zero torsion $T^a_{b0} = \Gamma^a_{b0}$ twists the geodesics.\\
The geodesic equations can be rewritten as $\ddot x^a= \beta {F^a}_b \dot x^b$ which are the equations of a classical particle of charge $-\beta$ moving on the plane $M_\Lambda$ where lives a normal magnetic field $F$, in accordance with the usual interpretation of the Berry curvature. This example is interesting since because of the zero curvature ($M_\Lambda$ is flat) it exhibits a pure effect of torsion. In the general case, the torsion has complicated expressions eq.(\ref{torsion1}-\ref{torsion5}) due to the non-trivial geometry and the non-trivial operator-valued Berry phase generator. But in this pure torsion example, the torsion reduced to be $T^a_{b0}={F^a}_b$ the Berry curvature, providing an interpretation of this one in this matrix model.

\subsubsection{Excited states:} We consider now a fermion of mass $\frac{\omega}{2}$ and then the excited states of $\slashed D_x$ (see \ref{chiral}). The solutions of $\slashed D_x|\Lambda_{p,n}\rrangle = p\omega |\Lambda_{p,n} \rrangle$ with $p \in \mathbb Z^*$ are (see \ref{CCR2}):
\begin{equation}
  M_{\Lambda,p,n} = \mathbb R^2 = \{(x^1,x^2,\sqrt{p^2\omega^2-n}), x^1=\Re(\alpha), x^2=\Im(\alpha)\}_{\alpha \in \mathbb C}
\end{equation}
with $n \in \{0,1,...,\lfloor p^2\omega^2 \rfloor\}$ which is the degeneracy index of the Floquet value $p\omega$.
\begin{equation}
  |\Lambda_n \rrangle = \frac{1}{\sqrt 2} (|0\rangle \otimes |n\rangle_\alpha + |1\rangle \otimes |n-1\rangle_\alpha)
\end{equation}
(the excited quasi-coherent states do not depend of $p$), where $|n\rangle_\alpha = \frac{(a^+-\bar \alpha)^n}{\sqrt{n!}}|\alpha \rangle$ ($\alpha = x^1 + \imath x^2$) [in the case $n=0$, we have $|\Lambda_0(w=1)\rrangle = |0\rangle \otimes |\alpha \rangle$]. The excited quasi-coherent states are maximally entangled and the density matrices of the spin are the microcanonical distribution $\rho_{\Lambda,n} = \frac{1}{2} \id$ ($n \not=0$).
\begin{eqnarray}
  \frac{\partial}{\partial \alpha}|n\rangle_\alpha & = & -\frac{\bar \alpha}{2}|n\rangle_\alpha + a^+|n\rangle_\alpha \\
  & = & \frac{\bar \alpha}{2}|n\rangle_\alpha + b^+_\alpha|n\rangle_\alpha \\
  & = & \frac{\bar \alpha}{2}|n\rangle_\alpha + \sqrt{n+1}|n+1\rangle_\alpha
\end{eqnarray}
\begin{eqnarray}
  \frac{\partial}{\partial \bar \alpha}|n\rangle_\alpha & = & -\frac{\sqrt{n}(a^+-\bar \alpha)^{n-1}}{\sqrt{(n-1)!}}|\alpha\rangle + \frac{(a^+-\bar \alpha)^n}{\sqrt{n!}} \frac{\partial}{\partial \bar \alpha}|\alpha \rangle \\
  & = & -\sqrt{n}|n-1\rangle_\alpha - \frac{\alpha}{2}|n\rangle_\alpha
\end{eqnarray}
It follows that the Berry phase generator is
\begin{eqnarray}
  A_n & =& -\imath \llangle \Lambda_n|d|\Lambda_n\rrangle \\
  & =&  -\imath \frac{\bar \alpha d\alpha - \alpha d\bar \alpha}{2} \\
  & = & x^1dx^2-x^2dx^1
\end{eqnarray}
  and then the metric of $M_{\Lambda,n,p} \times \mathbb R$ are the same than for the ground state.\\
  $\llangle \Lambda_{n*}| = \sqrt{2}(\langle 0|\otimes{_\alpha}\langle n|+\langle 1|\otimes{_\alpha}\langle n-1|)$, and then
\begin{eqnarray}
  \llangle \Lambda_{n*}|\sigma^i \frac{\partial}{\partial \alpha}|\Lambda_n\rrangle & = & \sqrt n \langle 0|\sigma^i|1\rangle\\
  \llangle \Lambda_{n*}|\sigma^i \frac{\partial}{\partial \bar \alpha}|\Lambda_n\rrangle & = & -\sqrt n \langle 1|\sigma^i|0\rangle
\end{eqnarray}
It follows that
\begin{eqnarray}
  \llangle \Lambda_{n*}|\sigma^i \partial_1|\Lambda_n\rrangle & = & \frac{\sqrt n}{2} (\langle 0|\sigma^i|1\rangle - \langle 1|\sigma^i|0\rangle) \\
  \llangle \Lambda_{n*}|\sigma^i \partial_2|\Lambda_n\rrangle & = & \imath \frac{\sqrt n}{2} (\langle 0|\sigma^i|1\rangle + \langle 1|\sigma^i|0\rangle)
\end{eqnarray}
  The Lorentz connection is then $\omega^{ij}_n=0$, $\omega^{01}_n = \sqrt n dx^2$, $\omega^{02}_n = -\sqrt n dx^1$ and $\omega^{03}_n=0$. Finally the non-zero Christoffel symbols are
\begin{equation}
  \begin{array}{ll}
    \Gamma^0_{10} = \sqrt n x^1 & \Gamma^0_{20} = \sqrt n x^2 \\ \Gamma^0_{12} = -\sqrt n & \Gamma^0_{21} = \sqrt n \\
    \Gamma^2_{10} = -1+\sqrt n & \Gamma^1_{20} = 1-\sqrt n \\
  \end{array}
\end{equation}
The (auto-parallel) geodesics are then solutions of
\begin{equation}
  \left\{\begin{array}{c} \ddot t + \frac{\sqrt{n}}{2}\frac{d|x|^2}{ds} \dot t = 0 \\ \ddot x^1+(1-\sqrt n) \dot x^2 \dot t = 0 \\ \ddot x^2-(1-\sqrt n) \dot x^1 \dot t = 0 \end{array} \right.
\end{equation}
The geodesics are straight lines on the plane $M_\Lambda$ for $n=1$. In that case, the fact to find usual geodesics is accidental, resulting from the killing of the contribution to the torsion of $A$ by the contribution of $\mathfrak A^{off}$. We recall that $n$ is a label defining an excited state, which is, in accordance with the interpretation of \ref{chiral}, a quasi-coherent state in a mode of chiral oscillations of the massive fermion.\\
We have $\dot t = \beta e^{-\frac{\sqrt n}{2}|x(s)|^2}$ and the geodesics in the plane are drawn figure \ref{geoplan}.
\begin{figure}
  \includegraphics[width=5cm]{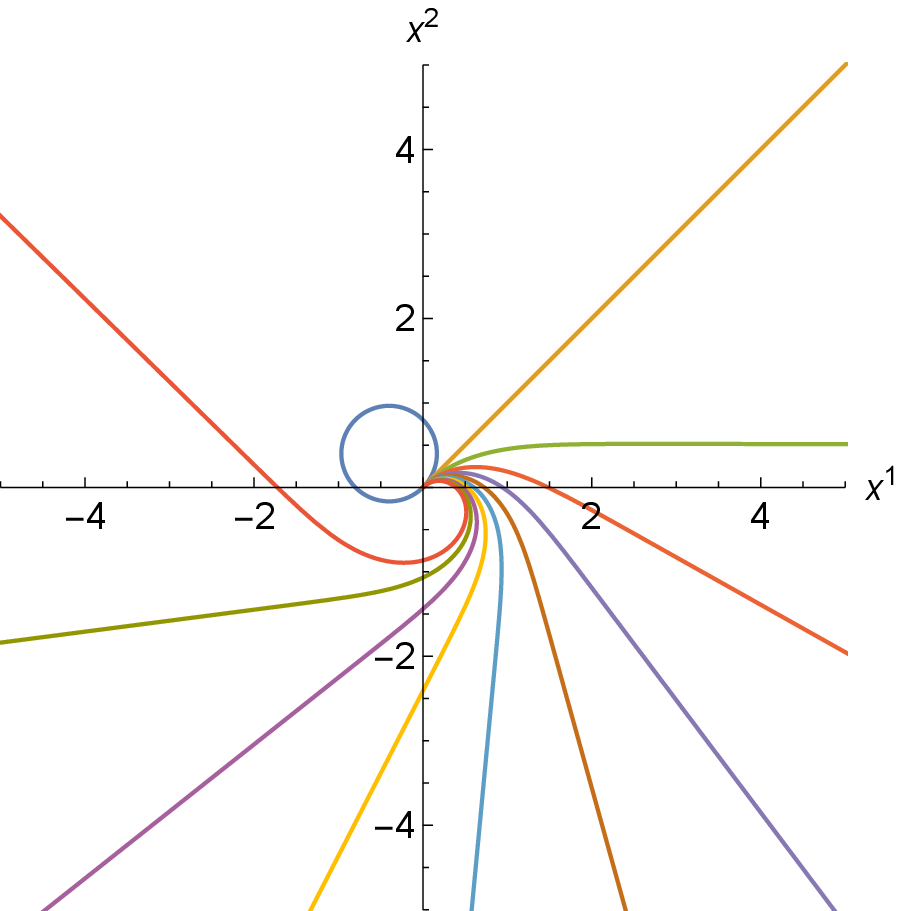} \includegraphics[width=5cm]{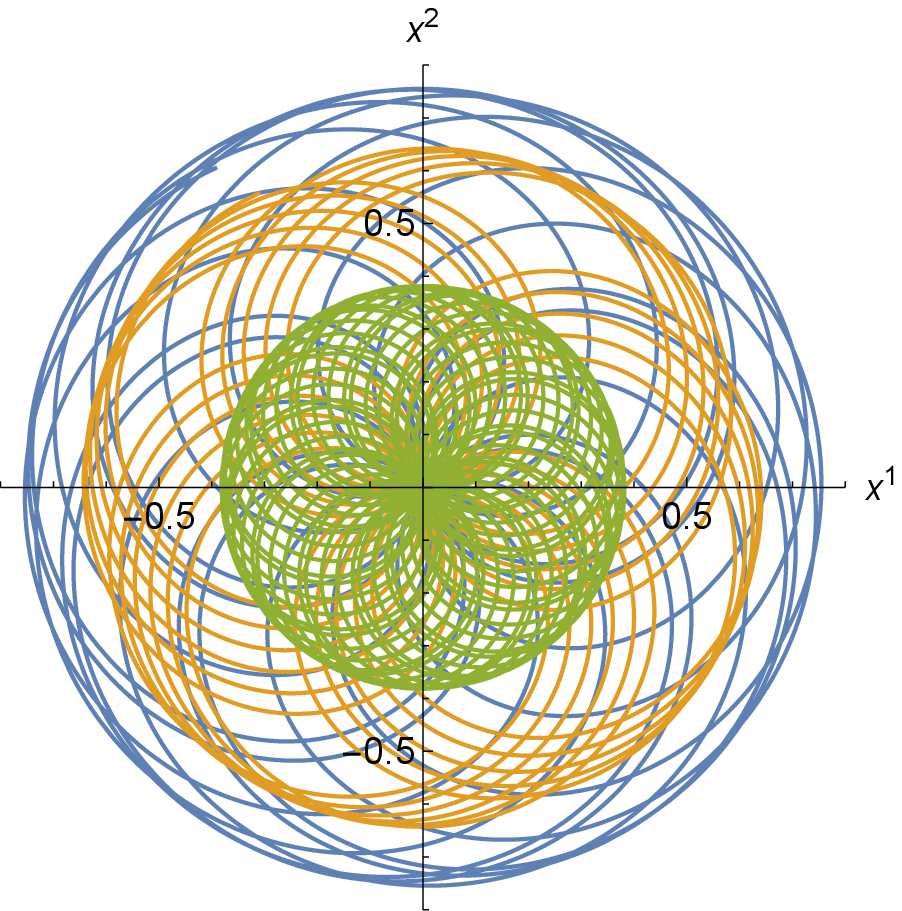}\\
  \caption{\label{geoplan} Geodesics in the emerging gravity for the noncommutative plane for $n=0$ to $n=10$ (left) and for $n=11, 12$ and $20$ (right) with $x^1(0)=x^2(0)=0$, $\dot x^1(0)=\dot x^2(0)=0.4$ and $\beta=1$.}
\end{figure}
For small values of $|x|^2$, the effect of the torsion is anew similar to that of a magnetic field normal to $M_\Lambda$. For large values of $|x|^2$ (with $n\not=0$), $\ddot x^a \simeq 0$ and the geodesics become straight lines (the effect of the torsion becomes negligible). $M_\Lambda \times \mathbb R$ is the emergent spacetime at the microscopic scale. But its space part is non-compact, it is then consistent than at large distance on $M_\Lambda$ (reaching the macroscopic scale) we find the behaviour of a (torsion free) classical spacetime. This is for large values of $|x|^2$ that we find the purely classical behaviour. We can interpret this by the structure of the coherent state $|\alpha \rangle = e^{-|\alpha|^2/2} \sum_{m=0}^{+\infty} \frac{\alpha^m}{\sqrt{m!}} |m\rangle$. For small values of $|\alpha|$, a few number of modes $|m\rangle$ are significantly populated. We can roughly  identify each mode to a string linking two D0-branes of the stack, in the sense than $a = \left(\begin{array}{cccc} \ddots & & &  \\ & 0 & \sqrt m & \\ & 0 & 0 & \\ & & & \ddots \end{array}\right)$ can be interpreted as representing a string of oscillation radius $\sqrt m$ linking two D0-branes at 0 in the complex plane. For $|\alpha|$ small, the number of strings contributing to the coherent geometry is then small. Conversely, for $|\alpha|$ large the number of modes significantly populated in $|\alpha\rangle$ is large and the number of strings contributing to the coherent geometry is large. In this sense, the macroscopic limit is identified with $|x|\to \infty$. 

\subsection{The noncommutative Minkowski space}
We consider anew the noncommutative plane but with time dependent operators: $X^1(t)=w(t)\frac{a+a^+}{2}$, $X^2(t)=w(t)\frac{a-a^+}{2\imath}$ and $X^3=0$. The dynamics of $\mathscr M$ is described by
\begin{eqnarray}
    \ddot X^i - [X_j,[X^i,X^j]]=0 & \iff & \ddot w=0 \\
    & \iff & \left\{\begin{array}{c} \dot w = p \\ \dot p = 0 \end{array} \right.
\end{eqnarray}
$p$ being a constant of motion, we can consider it as a parameter. Because we want to consider the case of a fast evolution of $\mathscr M$, $p$ is large. We promote $w$ as a new quantum variable associated with a fourth coordinate observable defined as being the Koopman generator (see \ref{fastdyn}). This one is $X^4 = -\imath p \frac{\partial}{\partial w}$, with $\Sp(X^4) = \mathbb R$ ($X^4 e^{-\imath x^4 \frac{w}{p}} = x^4 e^{-\imath x^4 \frac{w}{p}}$), in the Hilbert space $L^2(\mathbb R,\varrho(w)dw)$. Since the phase space is not compact, we do not have a natural choice for the density $\varrho$. In order to the Koopman eigenfunctions $e^{-\imath x^4 \frac{w}{p}}$ be normalisable, it needs that $\int_{-\infty}^{+\infty} \varrho(w)dw = 1$. Physically, the emergent geometry is valid in the neighbourhood measured by the Galilean observer in which the synchronous coordinates are defined. We can then consider that $\varrho$ is uniform on a domain corresponding to this neighbourhood and outside it decreases fastly.\\
We consider the total Dirac operator $\slashed D_x^K = \sigma_I \otimes (X^I-x^I)$ (with $I \in \{1,2,3,4\}$ and $\sigma_4=\id_2$). The quasi-coherent states are then such that
\begin{eqnarray}
  \slashed D_x^K |\Lambda_n \rrrangle = 0 & \iff & \imath p \frac{\partial}{\partial w} |\Lambda_n(w)\rrangle = \slashed D_x(w) |\Lambda_n(w) \rrangle \\
  & \iff & \Teg^{-\frac{\imath}{p} \int_1^w \slashed D_x(w) dw} |\Lambda_n(w=1)\rrangle
\end{eqnarray}
with
\begin{equation}
  \slashed D_x(w) = \left(\begin{array}{cc} -x^3-x^4 & wa^+-\bar \alpha \\ wa - \alpha & x^3-x^4 \end{array} \right)
\end{equation}
and with by continuity for $p \to 0$:
\begin{equation}
  \slashed D_x(w=1)|\Lambda_n(w=1)\rrangle = 0
\end{equation}
The solutions of this equation are (see \ref{CCR2}):
\begin{equation}
  M_{\Lambda,n} = \{x \in \mathbb R^4, x^4 = \pm \sqrt{(x^3)^2+n} \} \qquad n\in \mathbb N
\end{equation}
\begin{equation}
  |\Lambda_n(w=1)\rrangle = \frac{1}{\sqrt 2} (|0\rangle \otimes |n\rangle_\alpha + |1\rangle \otimes |n-1\rangle_\alpha)
\end{equation}
where $|n\rangle_\alpha = \frac{(a^+-\bar \alpha)^n}{\sqrt{n!}}|\alpha \rangle$ ($\alpha = x^1 + \imath x^2$) [in the case $n=0$, we have $|\Lambda_0(w=1)\rrangle = |0\rangle \otimes |\alpha \rangle$]. The quasi-coherent eigenspace is degenerate, and it follows that we have a collection of disconnected three dimensional eigenmanifolds $\{M_{\Lambda,n}\}_{n\in \mathbb N}$ in $\mathbb R^4$. The hypersurfaces $M_{\Lambda,n}$ can be viewed as ``parallel'' brane worlds in the ``bulk'' $\mathbb R^4$ with matter confined onto them.\\
$M_{\Lambda,n}$ is parametrised by $\mathbb R^3 \ni u \mapsto x(u) = (u^1,u^2,u^3,\pm \sqrt{(u^3)^2+n}) \in \mathbb R^4$. The metric of $M_\Lambda$ is then
\begin{eqnarray}
  d\ell^2_n & = & - \frac{\partial x^I}{\partial u^a}\frac{\partial x^J}{\partial u^b} \eta_{IJ} du^a du^b \\
  & = & d(u^1)^2 + d(u^2)^2 + \frac{d(u^3)^2}{1+(u^3)^2/n}
\end{eqnarray}
For $n=0$ we have $d\ell^2_0 =  d(u^1)^2 + d(u^2)^2$, and $M_{\Lambda,0}$ can be considered as being two dimensional and is the eigenmanifold of the static noncommutative plane treated at section \ref{NCplane}. $\lim_{n \to +\infty} M_{\Lambda,n}$ is the flat space, $M_{\Lambda,n}$ is also the flat space with the coordinates change $\tilde u^3 = \sqrt{n} \mathrm{ar}\!\sinh(u^3/\sqrt{n})$.\\

The Berry phase generator is
\begin{eqnarray}
  A_n & = & -\imath \lllangle \Lambda_n|d|\Lambda_n \rrrangle \\
  & = & -\imath \int_{-\infty}^{+\infty} \llangle \Lambda_n|U_x(w)^{-1}dU_x(w)|\Lambda_n\rrangle \varrho(w)dw - \imath \llangle \Lambda_n|d|\Lambda_n \rrangle
\end{eqnarray}
with $U_x(w) = \Teg^{-\frac{\imath}{p} \int_1^w \slashed D_x(w) dw}$ and by denoting $|\Lambda_n(w=1)\rrangle$ simply by $|\Lambda_n\rrangle$.
\begin{equation}
  -\imath \llangle \Lambda_n|d|\Lambda_n\rrangle = -\imath \frac{\bar \alpha d\alpha - \alpha d\bar \alpha}{2} = u^1du^2-u^2du^1
\end{equation}
This part is the Berry phase generator of the static noncommutative plane.
\begin{equation}
  \slashed D_x(w) = w\left(\begin{array}{cc} 0 & a^+ \\ a & 0\end{array}\right) - \sigma_Ix^I \Rightarrow \frac{\partial \slashed D_x(w)}{\partial x^I} = -\sigma_I
\end{equation}
We have then
  \begin{equation}
    U_x(w)^{-1}dU_x(w) = \frac{\imath}{p} (w-1) \sigma_I \frac{\partial x^I}{\partial u^a}du^a + \mathcal O(1/p^2)
  \end{equation}
  $\llangle \Lambda_n|\sigma_I|\Lambda_n\rrangle = \tr \sigma_I$, it follows that
  \begin{eqnarray}
    & & -\imath \int_{-\infty}^{+\infty} \llangle \Lambda_n|U_x(w)^{-1}dU_x(w)|\Lambda_n\rrangle \varrho(w)dw \nonumber \\
    & & \qquad = \pm \frac{C}{p} \frac{u^3}{\sqrt{(u^3)^2+n}} du^3 + \mathcal O(1/p^2)
  \end{eqnarray}
  where $C = \int_{-\infty}^{+\infty} (w-1)\varrho(w)dw$. Finally:
\begin{equation}
  A_n = u^1du^2-u^2du^1\pm\frac{C}{p}\sinh(\tilde u^3/\sqrt n) d\tilde u^3 + \mathcal O(1/p^2)
\end{equation}
  The Berry curvature is only $F=du^1\wedge du^2$ (independent of $n$). Finally the metric of the spacetime $M_{\Lambda,n} \times \mathbb R$ is
\begin{eqnarray}
  ds^2_n & = & (1-(u^1)^2-(u^2)^2)dt^2 \nonumber \\
  & & \quad +2u^1du^2dt-2u^2du^1dt \mp \frac{2C}{p}\sinh(\tilde u^3/\sqrt n) d\tilde u^3 dt \nonumber \\
  & & \quad - d(u^1)^2 - d(u^2)^2 - d(\tilde u^3)^2 + \mathcal O(1/p^2)
\end{eqnarray}

\begin{eqnarray}
  \lllangle \Lambda_{n*}|\sigma^i \partial_a|\Lambda_n\rrrangle & = & \int_{-\infty}^{+\infty} \llangle \Lambda_{n*}|U_x^{-1}(w)\sigma^i \frac{\partial U_x(w)}{\partial x^J}|\Lambda_n\rrangle \varrho(w)dw \frac{\partial x^J}{\partial u^a} \nonumber \\
  & & \qquad + \llangle \Lambda_{n*}|\sigma^i \partial_a|\Lambda_n \rrangle \\
  & = & \frac{\imath C}{p} \tr(\sigma^i \sigma_J) \frac{\partial x^J}{\partial u^a} + \llangle \Lambda_{n*}|\sigma^i \partial_a|\Lambda_n \rrangle + \mathcal O(1/p^2)\\
  & = & \frac{\imath C}{p} \frac{\partial x^i}{\partial u^a} + \llangle \Lambda_{n*}|\sigma^i \partial_a|\Lambda_n \rrangle + \mathcal O(1/p^2)\\
  & = & \frac{\imath C}{p} \delta^i_a + \llangle \Lambda_{n*}|\sigma^i \partial_a|\Lambda_n \rrangle + \mathcal O(1/p^2)
\end{eqnarray}  

The Lorentz connection is then (at order $1/p^2$)
\begin{eqnarray}
  \omega^{ij}_n  =  0 & \qquad & \omega^{03}_n  =  \frac{C}{p} du^3 \\
  \omega^{01}_n  =  \frac{C}{p} du^1 + \sqrt n du^2 & \qquad & \omega^{02}_n  =  -\sqrt n du^1 +  \frac{C}{p} du^2
 \end{eqnarray}

The nonzero Christoffel symbols are then (at order $1/p^2$)
\begin{equation}
  \begin{array}{llll}
    \Gamma^0_{10} = \sqrt n u^1+\frac{C}{p}u^2 & \Gamma^0_{20} = \sqrt n u^2 - \frac{C}{p}u^1 \\
    \Gamma^1_{20} = 1-\sqrt n  & \Gamma^2_{10} = -1+\sqrt n \\
    \Gamma^0_{21} = \sqrt n & \Gamma^0_{12} = -\sqrt n \\
    \Gamma^1_{10} = -\frac{C}{p} & \Gamma^2_{20} = -\frac{C}{p} \\
    \Gamma^0_{11} = \frac{C}{p}  & \Gamma^0_{22} = \frac{C}{p} \\
    \multicolumn{2}{l}{\Gamma^{\tilde 3}_{\tilde 30} = -\frac{C}{p} \cosh(\tilde u^3/\sqrt n)(1/\sqrt n +\cosh(\tilde u^3/\sqrt n))} \\
    \multicolumn{2}{l}{\Gamma^0_{\tilde 3 \tilde 3}= \frac{C}{p} \cosh^2(\tilde u^3/\sqrt n)}
  \end{array}
  \end{equation}

Some (auto-parallel) geodesics are drawn fig. \ref{geominkowski}
\begin{figure}
  \includegraphics[width=8cm]{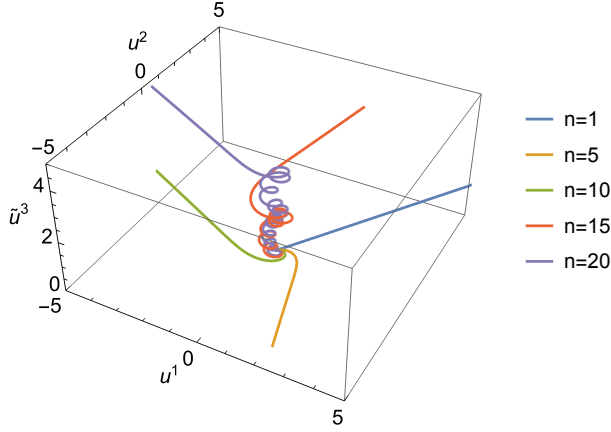}
  \caption{\label{geominkowski} Geodesics in emergent gravity for a noncommutative Minkowski space (with $p=100$ ans $C=1$) for different values of $n$ with $u^1(0)=u^2(0)=\tilde u^3(0)=0$, $\dot u^1(0)=\dot u^2(0)=0.4$, $\dot {\tilde u^3}(0) = 0.1$, $t(0)=0$ and $\dot t(0)=1$.}
\end{figure}
Anew the effect of the torsion is similar to the effect of a magnetic field in the direction $\partial_{\tilde 3}$, the same comments that those of the previous example occur.

\subsection{The Fuzzy sphere}
\subsubsection{Ground state:} 
We consider the case where $\mathscr M$ is defined by $X^i = rJ^i$ with $r\in \mathbb R^{+*}$ and $(J^i)$ the angular momentum operators in an irreducible representation ($\dim \mathcal H=2j+1$). The solution of eq.(\ref{NCEVeq}) are (see \ref{su2})
\begin{equation}
  M_\Lambda = \{x \in \mathbb R^3, |x|=rj\}
\end{equation}
\begin{equation}
  |\Lambda(x) \rrangle = |\zeta\rangle_{1/2} \otimes |\zeta \rangle_j
\end{equation}
with $\zeta = e^{\imath \varphi} \tan \frac{\theta}{2}$, the embedding of the sphere $M_\Lambda$ into $\mathbb R^3$ being
\begin{equation}
  x(\theta,\varphi) = rj(\sin \theta \cos \varphi, \sin \theta \sin \varphi, \cos \theta)
\end{equation}
and $|\zeta\rangle_j$ is a Perelomov $\mathfrak{su}(2)$ coherent state \cite{Perelomov,Puri}. The abelian gauge potential is then
\begin{equation}
  A =  -\imath ({_{1/2}}\langle \zeta|d|\zeta\rangle_{1/2} + {_j}\langle \zeta|d|\zeta\rangle_j) \\
\end{equation}
with $d=\partial + \bar \partial$.
\begin{eqnarray}
  & & d|\zeta\rangle_j = -j\frac{\bar \zeta d\zeta + \zeta d\bar \zeta}{1+|\zeta|^2} |\zeta \rangle_j \nonumber \\
  & & + \frac{1}{(1+|\zeta|^2)^j} \sum_{m=-j}^{j-1} \sqrt{\frac{(2j)!}{(j+m)!(j-m)!}}(j-m)\zeta^{j-m-1}|jm\rangle d\zeta
\end{eqnarray}
\begin{eqnarray}
  & & {_j}\langle \zeta|d|\zeta\rangle_j = -j \frac{\bar \zeta d\zeta + \zeta d\bar \zeta}{1+|\zeta|^2} \nonumber \\
    & & + \frac{\bar \zeta d\zeta}{1+|\zeta|^2} \sum_{m=-j}^{j-1} \frac{(2j)!}{(j+m)!(j-m-1)!}|\zeta|^{2(j-m-1)}
\end{eqnarray}
But by the binomial theorem:
\begin{eqnarray}
  & & \sum_{m=-j}^{j-1} \frac{(2j)!}{(j+m)!(j-m-1)!}|\zeta|^{2(j-m-1)} \nonumber \\
  & & = \sum_{m=0}^{2j-1} \frac{(2j)!}{m!(2j-m-1)!}|\zeta|^{2(2j-1-m)} \\
  & & = 2j(1+|\zeta|^2)^{2j-1}
\end{eqnarray}
  It follows that
\begin{equation}
  {_j}\langle \zeta|d|\zeta\rangle_j = -j \frac{\bar \zeta d\zeta + \zeta d\bar \zeta}{1+|\zeta|^2}+2j\frac{\bar \zeta d\zeta}{1+|\zeta|^2} = j \frac{\bar \zeta d\zeta - \zeta d\bar \zeta}{1+|\zeta|^2}
\end{equation}
Finally
\begin{equation}
  A = -\imath(\frac{1}{2}+j)\frac{\bar \zeta d\zeta-\zeta d\bar \zeta}{1+|\zeta|^2} = (2j+1) \sin^2\frac{\theta}{2} d\varphi
\end{equation}
The Berry curvature is then
\begin{equation}
  F_{\theta \varphi} = \partial_\theta A_\varphi = \frac{2j+1}{2} \sin \theta
\end{equation}
$\gamma_{ab} du^a du^b = (rj)^2(d\theta^2 + \sin^2\theta d\varphi^2)$ is the usual metric of a sphere of radius $rj$. $\mathbb R \times M_\Lambda$ is endowed with the metric
\begin{eqnarray}
  ds^2 & = & \left(1-\left(\frac{2j+1}{2rj}\right)^2 \tan^2 \frac{\theta}{2}\right)dt^2 \nonumber \\
  & & + 2(2j+1) \sin^2\frac{\theta}{2} d\varphi dt \nonumber \\
  & & - (rj)^2(d\theta^2+\sin^2\theta d\varphi^2)
\end{eqnarray}

The triads are then
\begin{equation}
  \begin{array}{lll}
    \tilde e^1_\theta = rj \cos \theta \cos \varphi & \tilde e^2_\theta = rj \cos \theta \sin \varphi & \tilde e^3_\theta = -rj \sin \theta \\
    \tilde e^1_\varphi = -rj \sin \theta \sin \varphi & \tilde e^2_\varphi = rj \sin \theta \cos \varphi & \tilde e^3_\varphi = 0
  \end{array}
\end{equation}
and $e^\theta_i = \theta_{ij} g^{\theta \theta} \frac{\partial x^j}{\partial \theta}$ and $e^\varphi_i = \theta_{ij} g^{\varphi \varphi} \frac{\partial x^j}{\partial \varphi}$ are then
\begin{equation}
  \begin{array}{lll}
    e^\theta_1 = r \sin \varphi & e^\theta_2 = -r \cos \varphi & e^\theta_3 =0 \\
    e^\varphi_1 = r\cos \varphi\, \mathrm{cotan}\, \theta & e^\varphi_2 = r \sin \varphi\, \mathrm{cotan}\, \theta & e^\varphi_3 = -r
  \end{array}
\end{equation}
and 
\begin{equation}
  \begin{array}{lll}
    e^\varphi_0 = \frac{2j+1}{4(rj)^2} (\tan^2 \frac{\theta}{2}+1) & e^\theta_0 = 0 \\
    \tilde e^1_0 = \frac{2j+1}{2rj} \tan \frac{\theta}{2} \sin \varphi & \tilde e^1_0 = - \frac{2j+1}{2rj} \tan \frac{\theta}{2} \cos \varphi & \tilde e^3_0=0
  \end{array}
\end{equation}

$\llangle \Lambda|\sigma^id|\Lambda \rrangle = {_{1/2}}\langle \zeta|\sigma^id|\zeta\rangle_{1/2} + {_{1/2}}\langle\zeta|\sigma^i|\zeta\rangle_{1/2} {_j}\langle \zeta |d|\zeta\rangle_j$. ${_j}\langle \zeta |d|\zeta\rangle_j = 2\imath j \sin^2 \frac{\theta}{2} d\varphi$, ${_{1/2}}\langle \zeta|\vec \sigma|\zeta\rangle_{1/2} = \vec n$ and
\begin{eqnarray}
  {_{1/2}}\langle \zeta|\sigma^1d|\zeta\rangle_{1/2} & = & (\cos \varphi \cos \theta + \imath \sin \varphi)\frac{d\theta}{2} + \imath e^{\imath \varphi} \sin \theta \frac{d\varphi}{2} \\
  {_{1/2}}\langle \zeta|\sigma^2d|\zeta\rangle_{1/2} & = & (\sin \varphi \cos \theta - \imath \cos \varphi)\frac{d\theta}{2} + e^{\imath \varphi} \sin \theta \frac{d\varphi}{2} \\
  {_{1/2}}\langle \zeta|\sigma^3d|\zeta\rangle_{1/2} & = & -\sin \theta \frac{d\theta}{2} - \imath \sin^2\frac{\theta}{2} d\varphi
\end{eqnarray}
The Lorentz connection is then
\begin{eqnarray}
  \omega^{12} & = & \sin \theta d\theta \\
  \omega^{23} & = & -\cos\varphi \cos \theta d\theta + \sin\varphi \sin \theta d\varphi \\
  \omega^{31} & = & -\sin \varphi \cos \theta d\theta - \cos\varphi \sin \theta d\varphi \\
  \omega^{01} & = & \sin \varphi d\theta + 2\cos\varphi \sin\theta (\frac{1}{2}+2j\sin^2\frac{\theta}{2}) d\varphi \\
  \omega^{02} & = & -\cos \varphi d\theta + 2\sin\varphi \sin\theta (\frac{1}{2}+2j\sin^2\frac{\theta}{2}) d\varphi \\
  \omega^{03} & = & 2\sin^2\frac{\theta}{2}(-1+2j\cos\theta)d\varphi
\end{eqnarray}
$\Xi^\theta_\varphi = -\frac{1}{r^2j} \sin \theta$, $\Xi^\varphi_\theta = \frac{1}{r^2j \sin \theta}$ and $\Xi^\theta_\theta=\Xi^\varphi_\varphi=0$, the non zero Christoffel symbols are then
\begin{equation}
  \begin{array}{ll}
    \Gamma^0_{\theta 0} = \frac{2j+1}{2rj} \tan \frac{\theta}{2} & \Gamma^0_{\theta \varphi} = - \Gamma^0_{\varphi \theta} = -rj \sin \theta \\
    \Gamma^\theta_{\varphi 0} = -\frac{\sin \theta}{rj} + \frac{2j+1}{2(rj)^2}\tan\frac{\theta}{2}\cos\theta & \Gamma^\varphi_{\theta 0} = \frac{1}{rj \sin \theta} - \frac{2j+1}{4(rj)^2}\frac{1+\tan^2 \frac{\theta}{2}}{\sin \theta} \\
    \Gamma^\theta_{\varphi \varphi} = - \cos \theta \sin \theta & \Gamma^\varphi_{\theta \varphi} =  \Gamma^\varphi_{\varphi \theta} =  \mathrm{cotan}\, \theta
  \end{array}
\end{equation}
The restriction to the space components corresponds to the Levi-Civita connection on the sphere. Induced (auto-parallel) geodesics are drawn fig. \ref{geosphere} and \ref{geosphereproj}.
\begin{figure}
  \includegraphics[width=6cm]{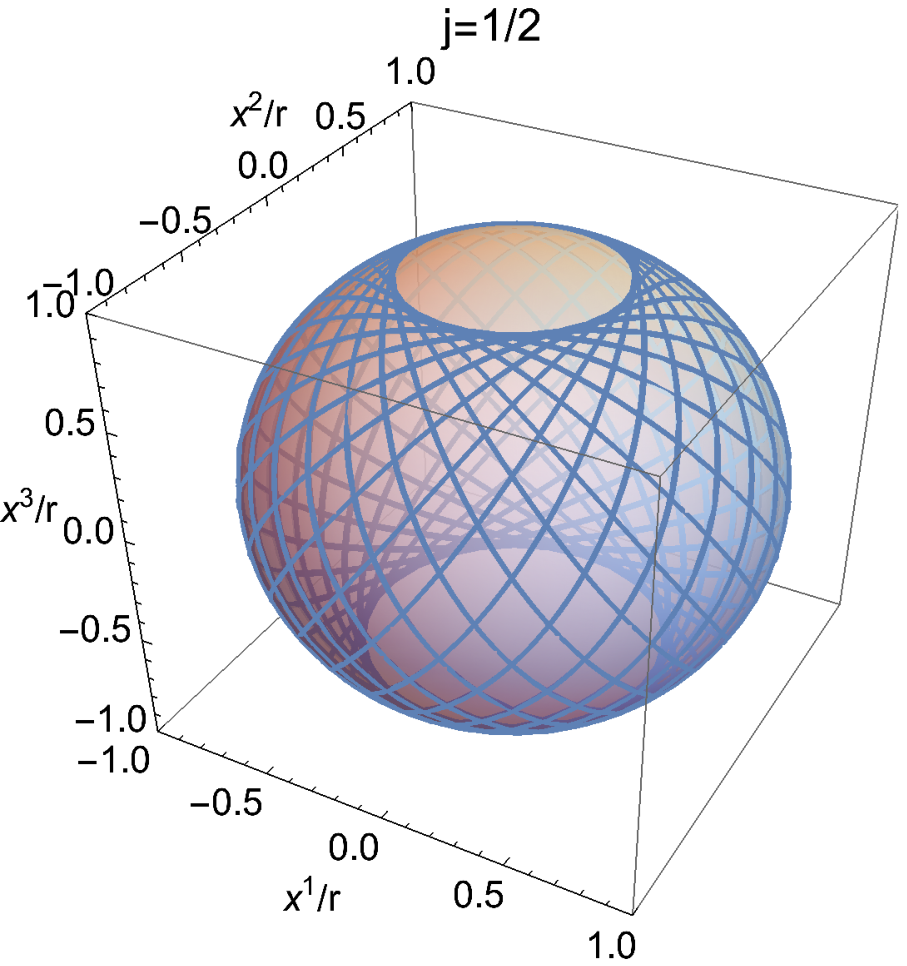} \includegraphics[width=6cm]{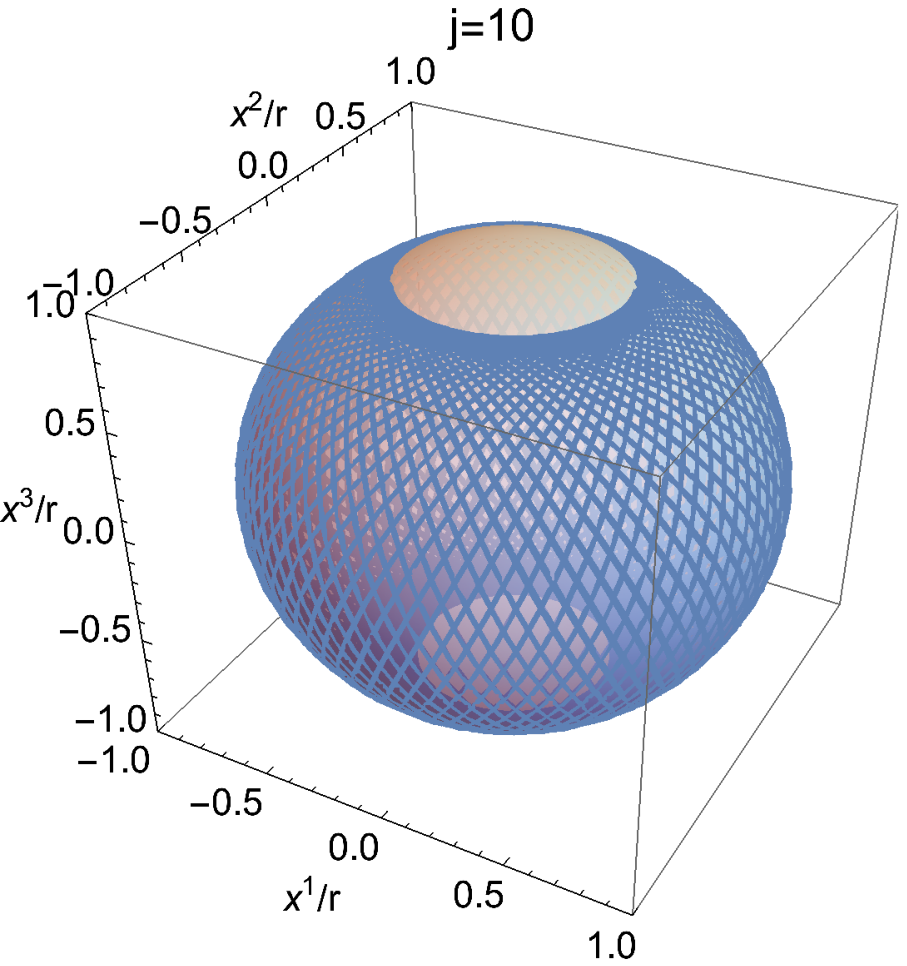} \\ \includegraphics[width=6cm]{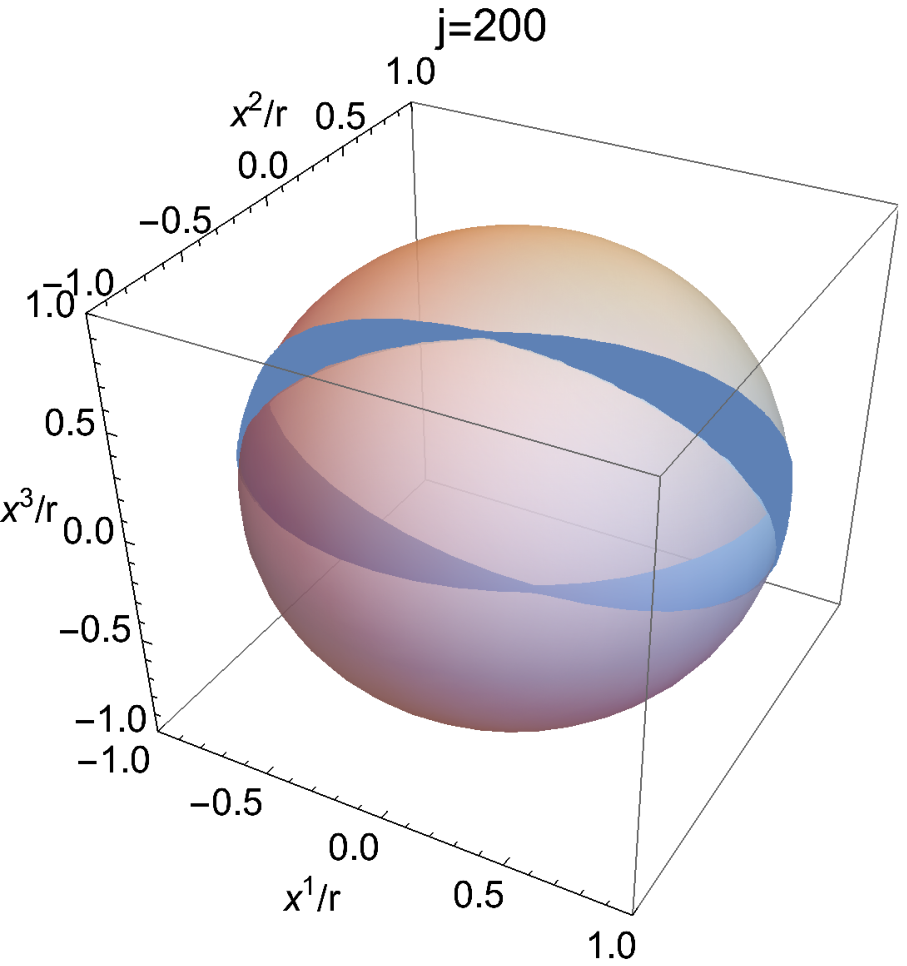} \includegraphics[width=6cm]{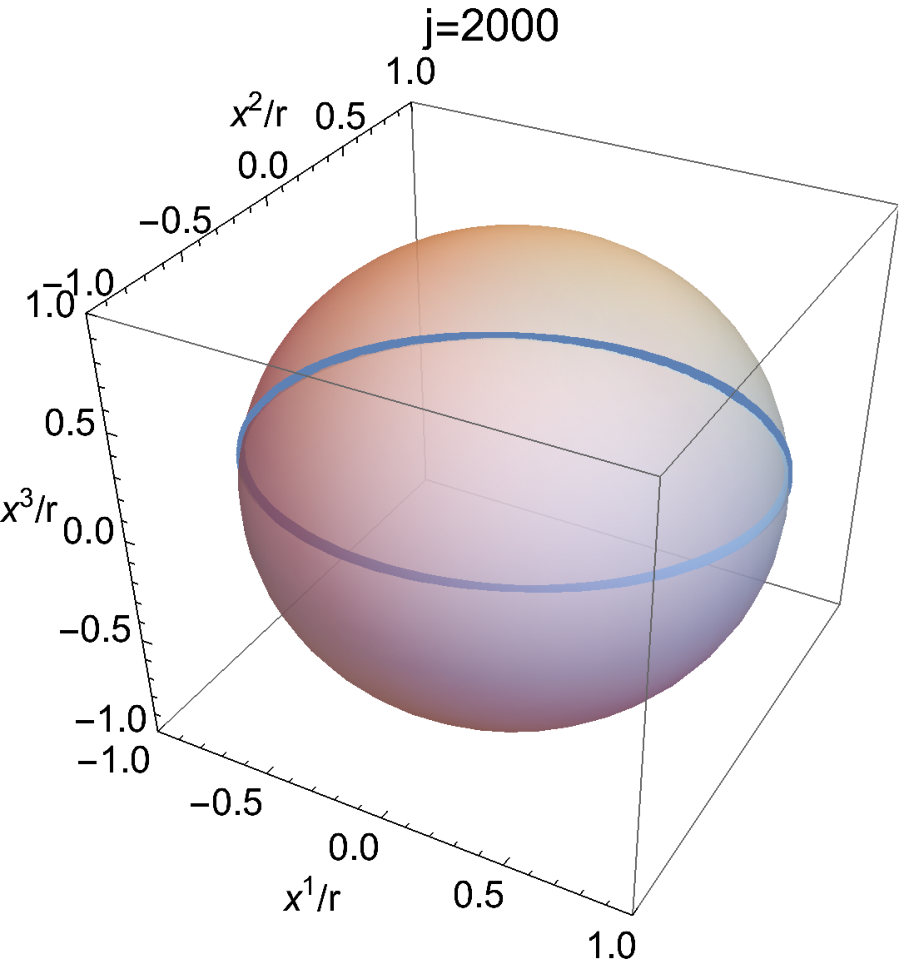} \\
  \caption{\label{geosphere} Geodesics on the sphere $M_\Lambda$ in emergent gravity of a fuzzy sphere with $\theta(0)=\frac{\pi}{4}$, $\phi(0)=0$, $\dot \theta(0)=0.8$, $\dot \phi(0)=0.8$, $t(0)=0$ and $\dot t(0)=1$ for different values of $j$.}
\end{figure}
\begin{figure}
  \includegraphics[width=4.5cm]{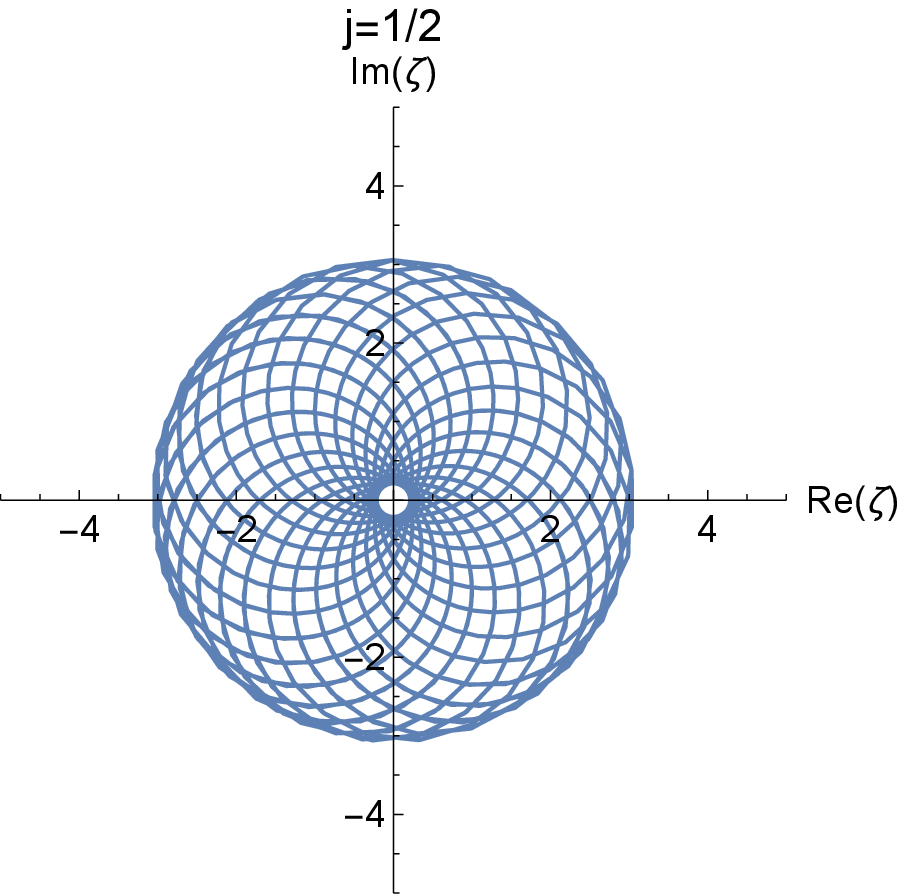} \includegraphics[width=4.5cm]{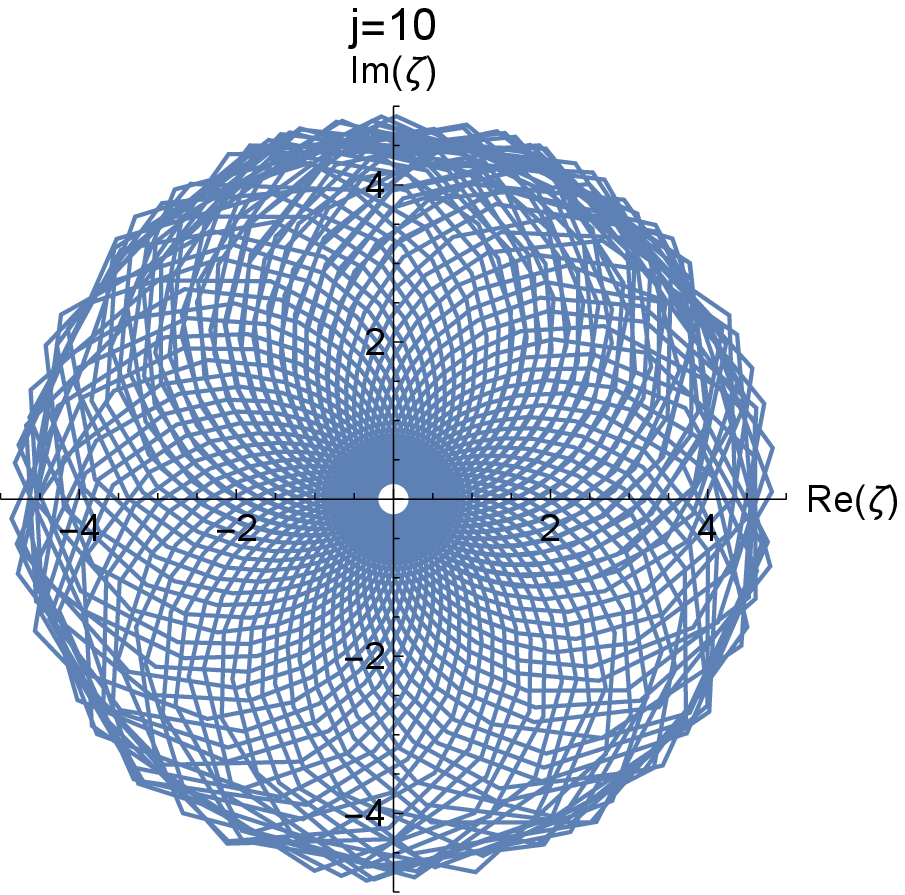}\includegraphics[width=4.5cm]{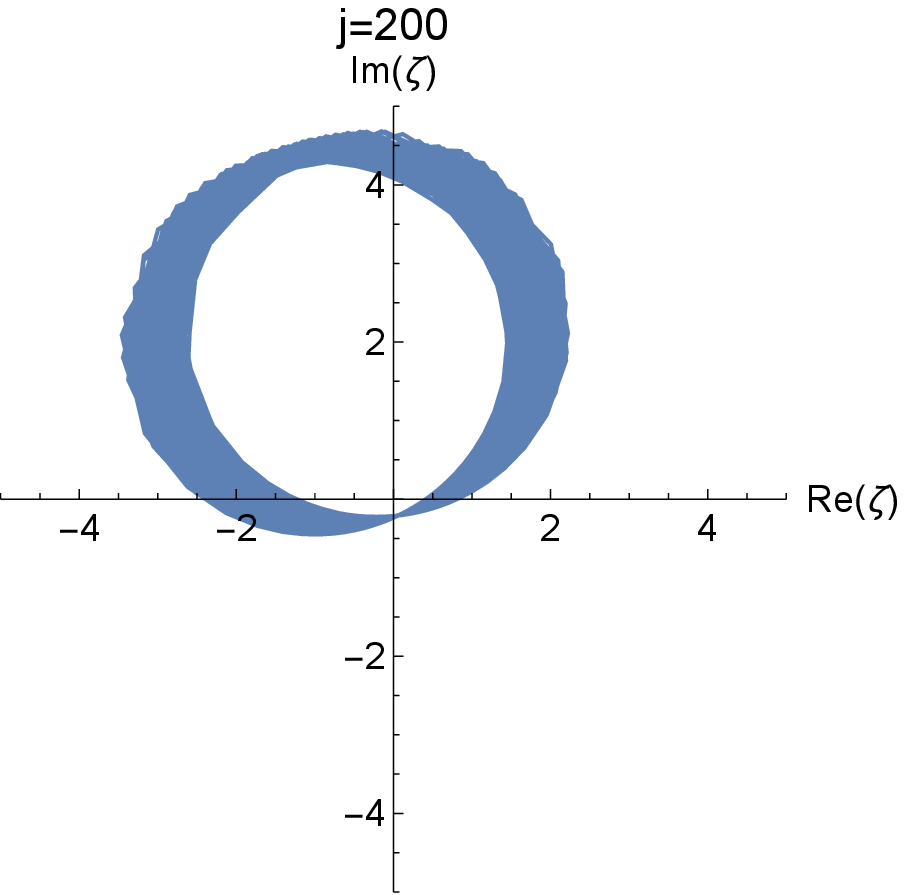}\\
  \caption{\label{geosphereproj} Same as fig. \ref{geosphere} but with a representation in the complex plane of $\zeta=e^{\imath \varphi} \tan \frac{\theta}{2}$ (stereographic projection of the sphere).}
\end{figure}
We find usual geodesics on a sphere for the thermodynamical limit $j \to +\infty$ (the number of D0-branes in the stack is $2j+1$) where the effects of the nonzero torsion are negligible. The torsion generates a precession of the classical geodesics (the speed of this precession decreasing with increasing values of $j$). In this case, the Berry curvature is equivalent to a radial magnetic field emitted by a magnetic monopole at the center of the sphere, in accordance with the observed precession (a Laplace force normal to the particle trajectory and tangent to the surface of the sphere).

\subsubsection{Excited states:} We consider now a fermion of mass $\frac{\omega}{2}$. The solutions of $\slashed D_x |\Lambda_{p,m}\rrangle = p\omega |\Lambda_{p,m} \rrangle$ with $p \in \mathbb Z^*$ are (see \ref{su2}):
$$ M_{\Lambda,p,m} = \left\{x \in \mathbb R^3, |x|=\left(2m+1\pm\sqrt{(2m+1)^2-(2j+1)^2+1+\frac{2p\omega}{r}}\right) \frac{r}{2} \right\} $$
with $m \in \{-j,...,j\}$ and $|2m+1|\geq \sqrt{2j(j+1)-\frac{p\omega}{r}}$ (the case $-\sqrt{}$ exists only if $\frac{p\omega}{2r} \leq j(j+1)$). 
\begin{equation}
  |\Lambda_{m,p} \rrangle = \Dis_{1/2\otimes j}(\vec n)(k^0_{m,p}|0\rangle \otimes |j,m\rangle + k^1_{m,p}|1\rangle \otimes |j,m+1\rangle)
\end{equation}
with
\begin{eqnarray}
  k^0_{m,p} & = & \kappa \left(\frac{p\omega}{r}+m+1-\frac{|x|}{r}\right) \\
  k^1_{m,p} & = & \kappa \sqrt{j(j+1)-m(m+1)}
  \end{eqnarray}
    ($\kappa$ being just a normalisation factor such that $\|k_{pm}\|^2 = 1$) and
\begin{equation}
  \Dis_j(\vec n) = e^{\imath \theta (\sin \varphi J^1 - \cos \varphi J^2)}
\end{equation}
$\vec n = (\sin \theta \cos \varphi, \sin \theta \sin \varphi, \cos \theta)$. $M_{\Lambda,n,p}$ is then the sphere of $\mathbb R^3$ of radius:
\begin{equation}
  r_{j,m,p} = \left(2m+1\pm\sqrt{(2m+1)^2-(2j+1)^2+1+\frac{2p\omega}{r}}\right) \frac{r}{2}
\end{equation}
 and then $d\ell_{m,p}^2 = r^2_{j,m,p}(d\theta^2+\sin^2\theta d\varphi^2)$. The density matrix of the spin is then
\begin{equation}
  \rho_{\Lambda,m,p} = \Dis_{1/2}(\vec n)\left(\begin{array}{cc} |k^0_{m,p}|^2 & 0  \\ 0 & |k^1_{m,p}|^2 \end{array} \right) \Dis_{1/2}(\vec n)^\dagger
\end{equation}
with
\begin{equation}
\Dis_{1/2}(\vec n) = \left(\begin{array}{cc} \cos \frac{\theta}{2} & -e^{-\imath \varphi} \sin \frac{\theta}{2} \\ e^{\imath \varphi} \sin \frac{\theta}{2} & \cos \frac{\theta}{2} \end{array} \right)
\end{equation}
It follows that
\begin{eqnarray}
  \llangle \Lambda_{p,m*}| & = &\llangle \Lambda_{p,m}|\rho_{\Lambda,p,m}^{-1} \\
  & = & \left(\frac{1}{k^0_{m,p}}\langle 0|\otimes \langle j,m|+\frac{1}{k^1_{m,p}}\langle 1|\otimes \langle j,m+1|\right)\Dis_{1/2\otimes j}^\dagger
\end{eqnarray}
  The Berry phase generator is then
\begin{eqnarray}
  \imath \tilde A_{p,m} & = & \llangle \Lambda_{p,m*}|d|\Lambda_{p,m}\rrangle \\
  & = & \langle j,m|\Dis_j^\dagger d\Dis_j|j,m\rangle+\langle j,m+1|\Dis_j^\dagger d\Dis_j|j,m+1\rangle \nonumber \\
  & & +\langle 0|\Dis_{1/2}^\dagger d\Dis_{1/2}|0\rangle+\langle 1|\Dis_{1/2}^\dagger d\Dis_{1/2}|1\rangle
\end{eqnarray}
\begin{equation}
  \Dis_j = e^{-\bar \zeta J^+}e^{\ln(1+|\zeta|^2)J^3}e^{\zeta J^-} = e^{\zeta J^-} e^{-\ln(1+|\zeta|^2)J^3}e^{-\bar \zeta J^+}
\end{equation}
  with $\zeta=e^{\imath \varphi} \tan \frac{\theta}{2}$ \cite{Perelomov,Puri}. It follows that
\begin{eqnarray}
  \Dis_j^\dagger \frac{\partial}{\partial \zeta} \Dis_j & = & \frac{\bar \zeta}{1+|\zeta|^2} e^{-\zeta J^-}J^3e^{\zeta J^-} + J^- \\
  & = & \frac{\bar \zeta}{1+|\zeta|^2}(J^3-\zeta J^-)+J^- \\
  \Dis_j^\dagger \frac{\partial}{\partial \bar \zeta} \Dis_j & = & -\frac{\zeta}{1+|\zeta|^2} e^{\bar \zeta J^+}J^3e^{-\bar\zeta J^+} - J^+ \\
  & = & -\frac{\zeta}{1+|\zeta|^2}(J^3-\bar \zeta J^+)-J^+
\end{eqnarray}
(because $\frac{\partial}{\partial \zeta}(e^{-\zeta J^-}J^3e^{\zeta J^-}) = e^{-\zeta J^-}[J^3,J^-]e^{\zeta J^-}=-J^-$). Finally we have
\begin{equation}
  \langle j,m|\Dis_j^\dagger d\Dis_j|j,m\rangle = m \frac{\bar \zeta d\zeta -\zeta d\bar \zeta}{1+|\zeta|^2}
\end{equation}
and
\begin{equation}
  \frac{\tilde A_{p,m}}{2} = -\frac{\imath}{2}(2m+1)\frac{\bar \zeta d\zeta -\zeta d\bar \zeta}{1+|\zeta|^2} = (2m+1)\sin^2\frac{\theta}{2} d\varphi
\end{equation}
The spacetime metric is then
\begin{eqnarray}
  ds^2 & = & \left(1-\left(\frac{2m+1}{2r_{jmp}}\right)^2 \tan^2 \frac{\theta}{2}\right)dt^2 \nonumber \\
  & & + 2(2m+1) \sin^2\frac{\theta}{2} d\varphi dt \nonumber \\
  & & - r_{jmp}^2(d\theta^2+\sin^2\theta d\varphi^2)
\end{eqnarray}

\begin{eqnarray}
  \llangle \Lambda_*|\sigma^id|\Lambda \rrangle & = & \llangle \Lambda_*(0)|\sigma^i \Dis_j^\dagger d\Dis_j|\Lambda(0)\rrangle \nonumber \\
  & & + \langle 0|\Dis_{1/2}^\dagger \sigma^i d\Dis_{1/2}|0\rangle + \langle 1|\Dis_{1/2}^\dagger \sigma^i d\Dis_{1/2}|1\rangle
\end{eqnarray}
  Let $s^i_a=\llangle \Lambda_*(0)|\sigma^i\Dis_j^\dagger\partial_a\Dis_j|\Lambda(0)\rrangle$.
\begin{eqnarray}
  s^1_a & = & \frac{k^0_{mp}}{k^1_{mp}} \langle j,m+1|\Dis^\dagger_j\partial_a\Dis_j|jm\rangle + \frac{k^1_{mp}}{k^0_{mp}}\langle jm|\Dis_j^\dagger \partial_a \Dis_j|j,m+1\rangle \\
  s^2_a & = & \imath \frac{k^0_{mp}}{k^1_{mp}} \langle j,m+1|\Dis^\dagger_j\partial_a\Dis_j|jm\rangle -\imath \frac{k^1_{mp}}{k^0_{mp}}\langle jm|\Dis_j^\dagger \partial_a \Dis_j|j,m+1\rangle \\
  s^3_a & = & \langle jm|\Dis^\dagger_j\partial_a\Dis_j|jm\rangle - \langle j,m+1|\Dis_j^\dagger \partial_a \Dis_j|j,m+1\rangle
\end{eqnarray}
with
\begin{eqnarray}
  \langle j,m+1|\Dis_j^\dagger d\Dis_j|jm\rangle & = & \frac{\zeta(\bar \zeta+1)}{1+|\zeta|^2} \sqrt{j(j+1)-m(m+1)} d\bar \zeta \\
  \langle jm|\Dis_j^\dagger d\Dis_j|j,m+1\rangle & = & \frac{\bar \zeta(1-\zeta)}{1+|\zeta|^2} \sqrt{j(j+1)-m(m+1)} d\zeta
\end{eqnarray}
and finally
\small \begin{eqnarray}
  s^1_\theta & = & \frac{\tan \frac{\theta}{2}}{2\kappa k^0_{mp}}(1+\Delta k^2_{mp} \cos \varphi \tan\frac{\theta}{2}-\imath \sin \varphi \tan \frac{\theta}{2}) \\
  s^1_\varphi & = & \frac{\imath \sin^2\frac{\theta}{2}}{\kappa k^0_{mp}}(-\Delta k^2_{mp} -\cos \varphi \tan\frac{\theta}{2}+\imath \Delta k^2_{mp} \sin \varphi \tan \frac{\theta}{2}) \\
  s^2_\theta & = & \frac{\imath \tan \frac{\theta}{2}}{2\kappa k^0_{mp}}(\Delta k^2_{mp} + \cos\varphi \tan\frac{\theta}{2}-\imath \Delta k^2_{mp} \sin \varphi \tan \frac{\theta}{2}) \\
  s^2_\varphi & = & \frac{\sin^2 \frac{\theta}{2}}{\kappa k^0_{mp}}(1+\Delta k^2_{mp} \cos \varphi \tan \frac{\theta}{2}-\imath \sin \varphi \tan \frac{\theta}{2}) \\
  s^3_\theta & = & 0 \\
  s^3_\varphi & = & -2\imath \sin^2 \frac{\theta}{2}
\end{eqnarray} \normalsize
with $\Delta k^2_{mp} = (k^0_{mp})^2-(k^1_{mp})^2$ the difference of populations. The Lorentz connection is then
\small \begin{eqnarray}
  \omega^{12} & = & 0 \\
  \omega^{23} & = & -\frac{\tan \frac{\theta}{2}}{2\kappa k^0_{mp}}(1+\Delta k^2_{mp} \cos \varphi \tan \frac{\theta}{2}) d\theta + \frac{\sin^2\frac{\theta}{2}}{\kappa k^0_{mp}} \sin \varphi \tan \frac{\theta}{2} d\varphi \\
  \omega^{31} & = & -\frac{\Delta k^2_{mp}}{2\kappa k^0_{mp}} \sin\varphi \tan^2\frac{\theta}{2} d\theta - \frac{\sin^2\frac{\theta}{2}}{\kappa k^0_{mp}}(1+\Delta k^2_{mp} \cos\varphi \tan \frac{\theta}{2}) d\varphi \\
  \omega^{01} & = & -\frac{1}{2\kappa k^0_{mp}}\sin \varphi \tan^2\frac{\theta}{2} d\theta - \frac{\sin^2\frac{\theta}{2}}{\kappa k^0_{mp}}(\Delta k^2_{mp}+\cos \varphi \tan\frac{\theta}{2}) d\varphi \nonumber \\
  & & + \sin \varphi d\theta + \cos \varphi \sin \theta d\varphi \\
  \omega^{02} & = & \frac{\tan \frac{\theta}{2}}{2\kappa k^0_{mp}}(\Delta k^2_{mp} +\cos \varphi \tan \frac{\theta}{2})d\theta -\frac{\sin^2\frac{\theta}{2}}{\kappa k^0_{mp}} \sin \varphi \tan \frac{\theta}{2} d\varphi \nonumber \\
  & &  -\cos \varphi d\theta + \sin \varphi \sin \theta d\varphi \\
  \omega^{03} & = & -4 \sin^2 \frac{\theta}{2} d\varphi
  \end{eqnarray} \normalsize

The Lorentz connection of the excited states is then very different from the one of the ground state. The usual space components of the sphere are not present and very complicated expressions depending on the population difference $\Delta k^2_{pm}$ appear and are difficult to interpret.

\section{Conclusion}
The mean geometry of a quantum spacetime can be revealed by transporting adiabatically a probe D0-brane (a test particle). The emerging geometry at the microscopic scale is defined by the push-forward of the quantum averaging map $\omega_\Lambda$ of space quantum observables in the quasi-coherent state. Its dual map, the pull-back of $\omega_\Lambda$ appears as the quantisation map. The Berry connection of the adiabatic transport is the shift vector of the emergent geometry (which appears as a foliation by spacelike surfaces). The non-abelian part of the operator-valued Berry connection of the weak adiabatic regime defines the Lorentz connection of the emergent geometry, which is not torsion free. The effects of this torsion seems to be consistent with the interpretation of the Berry curvature as a pseudo magnetic field. Generalised quasi-coherent states can be introduced as excited states of the quantum spacetime which corresponds to the Fourier modes of chiral oscillations (implying that massless particles reveal only the emergent geometry of the spacetime ground state). This paper is focused on the emergent geometry at the microscopic scale by transport of fermions, which is related to their spins as in Einstein-Cartan theory.\\
It would be interesting to compare with the emergent geometry revealed by transport of bosons. Indeed, in the case of scalar bosons the quasi-coherent state is eigenvector of the non-commutative Laplacian $\Box_x = |X-x|^2$ (see \cite{Schneiderbauer}). But we have:
\begin{equation}
  \slashed D_x^2 = \Box_x + \frac{\imath}{2} {\varepsilon_{ij}}^k \sigma_k \otimes [X^i,X^j]
\end{equation}
The analysis in \cite{Schneiderbauer} suggests that the geometries described by the eigenvector minimising the displacement energy is the same in the two cases. At least when the ground quasi-coherent state of $\slashed D_x$ is separable $|\Lambda(x) \rrangle = |0_x\rangle \otimes |\lambda_0(x) \rangle$, we can imagine that $|\lambda_0(x) \rangle$ is close to the quasi-coherent state of $\Box_x$, maybe up to a gauge change. But for excited states or for an entangled ground state, differences must arise because of the lack of spin degree of freedom in the bosonic case. In the fermionic case, this one especially ``feels'' the non-commutative character of $\mathscr M$ via $\frac{\imath}{2} {\varepsilon_{ij}}^k \sigma_k \otimes [X^i,X^j]$.\\
The emergent geometry at the macroscopic level is obtained by the semi-classical limit (which is equivalent to the thermodynamical limit of infinite number of strings) and has been studied in a lot of previous papers. We have seen that the two limits are consistent with each other for the case of the strict adiabatic limit. Is there a decoherence process inducing the transition from the emergent geometry at the microscopic level to the one at the macroscopic level? For the weak adiabatic regime, the answer to this question depends on the behaviour of $\mathfrak A$ at the semi-classical limit. This one is not obvious and need more analyses. In a previous work \cite{Viennot3} we have proved that if we consider the weak adiabatic limit of some bipartite quantum systems both with a perturbative expansion, at the second order of perturbation $\mathfrak A$ must be accompanied by several operators involving that the density matrix obeys to an effective Lindblad equation (encoding decoherence effects). Maybe such corrections occur in the semi-classical limit of the weak adiabatic transport in the present framework at order $1/N^2$ (representing the intermediate regime between the purely quantum and classical ones). But such a study needs to extend the results of \cite{Viennot3} (the structure of $\slashed D_x$ is not compatible with the assumptions used in \cite{Viennot3}).\\
In this paper we have considered only academic models for the algebra $\mathfrak X$ which are totally analytical. Future works could be dedicated to apply the present adiabatic formalism to more realistic operators $\{X^i\}_i$ issued from the integration of the equations governing them. For example in \cite{Hanada} it is shown that such operators exhibit behaviours of higher derivative corrections to the gravity. These ones can maybe then be incorporated in the non-commutative manifold $\mathscr M$ and finally in the emergent geometry defined by the quasi-coherent state. But such an analysis needs the development of specific numerical methods to compute the adiabatic entities from a numerical simulation of $\mathscr M$.\\
A last question concerns the possibility to observe the consequences of the torsion associated with the emergent geometry. At the macroscopic scale this question is related to the one concerning the thermodynamical limit of $\mathfrak A$. With another matrix model, in \cite{Steinacker2} it is shown that a torsion arises at the cosmic scale and it can be a ``source'' of dark matter. This torsion has similarities with the one introduced in this paper (see \ref{NCtorsion}). Nevertheless, in the examples treated section \ref{examples}, we have seen that the effects of the torsion emerging from the quasi-coherent state of the quantum spacetime seem disappear at large scale. For the noncommutative plane or for the noncommutative Minkoswki space, the effects of the torsion exponentially decrease with the distance to the starting point. For the Fuzzy sphere, the effects of the torsion decrease with the number of strings (at the thermodynamical limit, the geodesics become the usual geodesics of a sphere). No argument proves that this behaviour is universal, maybe the examples treated here are too simples. Because the emergent geometry considered in this paper is at the microscopic level, in the sense of being the classical geometry closest to the quantum one (by using quasi-coherent states), it might be more relevant to test the effects of the torsion at the quantum level. Unfortunately, in the treated examples, the extinction of its effects seems very fast with the increase of the scale. Anew, maybe the examples are too simples and other systems maybe can exhibit some effects continuing on a larger scale, since the existence of a torsion at large scale in a matrix model has been shown in \cite{Steinacker2}.

\appendix
\section{Weak adiabatic theorem}\label{adiabth}
In the context of the present paper, the weak adiabatic regime physically means that the transport of the probe D0-brane onto $M_\Lambda$ can modify the spin of the fermionic string but not its space degree of freedom. This is consistent with the interpretation of the quasi-coherent state $|\Lambda(x)\rrangle$ which represents a state of a strongly localised fermion at the point $x$ onto the classical manifold (corresponding to the state of the noncommutative manifold $\mathscr M$ closest to a classical one). During a localised transport, the spin of the particle can rotate (Thomas and de Sitter precessions) but the mean value of the particle position is supposed to follow the path of the transport. This assumption means that $|\psi(t) \rrangle = g(t)|\Lambda(x(t))\rrangle$ with $g \in GL(2,\mathbb C)$ (and not in $GL(\mathbb C^2 \otimes \mathcal H)$) and $x(t) \in M_\Lambda$. For the point of view of $\slashed D_x = \sigma_i \otimes (X^i -x^i) \in \mathfrak a \otimes \Env(\mathfrak X)$, this means that the transport is adiabatic with respect to the noncommutative manifold $\mathscr M$ (of observables $\Env(\mathfrak X)$) and not with respect to the spin system (of observables $\mathfrak a$). In contrast, the strong adiabatic regime consists to a transport adiabatic with respect to both the spin and the noncommutative manifold.\\
This definition of the weak adiabaticity is consistent with the one used in \cite{Viennot1,Viennot2,Viennot3} where the adiabatic transport of a system interacting with an environment is studied. Due to the interaction with the environment, the system is not sustain on an instantaneous eigenstate (as in strong adiabatic transport). A transformation acting on the system eigenstate occurs due to the effects of the environment. The weak adiabatic regime ensures that the degree of entanglement between the system and the environment is sustained. In the present context, the system is the spin, the environment is the noncommutative manifold and the transformation onto the system is $g$. The instantaneous entanglement is completely defined by the one of $|\Lambda\rrangle$.\\

We can justify the weak adiabatic transport formula by the following result:
\begin{theo}[Weak adiabatic theorem]\label{adiabtheo}
  Let $\{\lambda_n\}_n$ be the spectrum of $\slashed D_x$: $\slashed D_x|\Lambda_n \rrangle = \lambda_n|\Lambda_n \rrangle$, with $\lambda_0=0$. The eigenvalues are supposed to be non-degenerate. Let $T$ be the duration of the transport described by the Schr\"odinger-like equation $\imath |\dot \psi \rrangle = \slashed D_{x(t)} |\psi \rrangle$, with $|\psi(0)\rrangle = |\Lambda_0(0)\rrangle$. $t \mapsto \lambda_n(t)$ and $t \mapsto |\Lambda_n(t)\rrangle$ are supposed to be respectively $\mathcal C^1$ and $\mathcal C^2$. Let $g(t) = \Ted^{-\imath \int_0^t \mathfrak A dt}$ be the operator-valued geometric phase ($\mathfrak A \rho_{\Lambda 0} = -\imath \tr_{\mathcal H} |\dot \Lambda_0\rrangle \llangle \Lambda_0|$, with $\rho_{\Lambda 0} = \tr_{\mathcal H} |\Lambda_0\rrangle \llangle \Lambda_0|$) and $\tilde \lambda_0 = \llangle \Lambda_0|g^{-1} \slashed D_x g|\Lambda_0\rrangle$ be the effective ground eigenvalue.\\
If the following assumptions are satisfied:
  \begin{enumerate}
  \item $\max_{n\not=0} \sup_{t \in [0,T]} \left|\frac{\llangle \Lambda_0|(\partial_t-\imath \mathfrak A)|\Lambda_n\rrangle}{\lambda_n-\tilde \lambda_0} \right| = \mathcal O(1/T)$.
  \item $\max_{n\not=0} \sup_{t \in [0,T]} \left|\frac{\llangle \Lambda_0|g^{-1} \slashed D_x g|\Lambda_n\rrangle}{\lambda_n-\tilde \lambda_0} \right| = \mathcal O(\epsilon)$
  \end{enumerate}
  then
  \begin{equation}
    |\psi(t)\rrangle = e^{-\imath \int_0^t \tilde \lambda_0(t)dt} g(t)|\Lambda_0(t)\rrangle + \mathcal O(\max(\epsilon,1/T))
  \end{equation}
\end{theo}

\begin{proof}
  $(|\Lambda_n\rrangle)_n$ forming a basis of $\mathbb C^2 \otimes \mathcal H$ and $g$ being invertible, let $c_n \in \mathbb C$ be such that
\begin{equation}
  |\psi \rrangle = \sum_n c_n e^{-\imath \varphi_n} g|\Lambda_n\rrangle
\end{equation}
  with $\varphi_n = \int_0^t \lambda_n(t)dt$ for $n\not=0$ and $\varphi_0=\int_0^t \tilde \lambda_0(t)dt$. By injecting this expression in the Schr\"odinger-like equation and by projecting onto $\llangle \Lambda_n|g^{-1} e^{\imath \varphi_n}$ we find
  \begin{eqnarray}
    & & \imath \dot c_n + \lambda_n c_n + \imath \sum_p c_p e^{\imath(\varphi_n-\varphi_p)} \llangle \Lambda_n|g^{-1}\dot g|\Lambda_p \rrangle \nonumber \\
    & & \qquad +\imath \sum_p c_p e^{\imath (\varphi_n-\varphi_p)} \llangle \Lambda_n|\dot \Lambda_p\rrangle \nonumber \\
    & & \qquad \quad = \sum_p c_p  e^{\imath (\varphi_n-\varphi_p)} \llangle \Lambda_n|g^{-1}\slashed D_x g|\Lambda_p\rrangle
  \end{eqnarray}
  By integrating this expression, we have
  \begin{eqnarray}\label{eqC}
    c_n(t) & = & c_n(0) - \imath \sum_p \int_0^t e^{\imath (\varphi_n-\varphi_p)} \llangle \Lambda_n|(\partial_t-\imath \mathfrak A)|\Lambda_p\rrangle c_p(t)dt \nonumber \\
    & & - \imath \sum_p \int_0^t e^{\imath (\varphi_n-\varphi_p)} (\llangle \Lambda_n|g^{-1}\slashed D_x g|\Lambda_p\rrangle - \lambda_n \delta_{np}) c_p(t)dt
  \end{eqnarray}
  \begin{itemize}
  \item If $|\Lambda_0\rrangle = |0_x\rangle \otimes |\lambda_0\rangle$ is separable:
    \begin{equation}
      \mathfrak A = -\imath |\dot 0_x\rangle\langle 0_x|-\imath \langle \lambda_0|\dot \lambda_0\rangle|0_x\rangle\langle 0_x|
    \end{equation}
    and then $\llangle \Lambda_0|(\partial_t-\imath \mathfrak A)|\Lambda_0\rrangle = 0$.
 \item If $|\Lambda_0\rrangle = |0\rangle \otimes |\Lambda^0\rangle+|1\rangle \otimes |\Lambda^1\rangle$ is entangled:
\begin{eqnarray}
  \mathfrak A & = & -\imath(\langle \Lambda^0_*|\dot \Lambda^0_0\rangle|0\rangle \langle 0|+\langle \Lambda^1_*|\dot \Lambda^0_0\rangle|0\rangle \langle 1|\nonumber \\
  & & +\langle \Lambda^0_*|\dot \Lambda^1_0\rangle |1\rangle \langle 0|+\langle \Lambda^1_*|\dot \Lambda^1_0\rangle |1\rangle \langle 1| \label{structA}
\end{eqnarray}
with $\llangle \Lambda_*|=\llangle \Lambda_0|\rho_{\Lambda 0}^{-1}$; and then $\llangle \Lambda_0|(\partial_t-\imath \mathfrak A)|\Lambda_0\rrangle = 0$.
  \end{itemize}

  The summations in eq.(\ref{eqC}) for the case $n=0$ are then restricted on $p\not=0$. By integration by parts, we have for $p\not=0$
  \begin{eqnarray}
    & & \int_0^t e^{- \imath(\varphi_p-\varphi_0)} \llangle \Lambda_0|(\partial_t-\imath \mathfrak A)|\Lambda_p\rrangle c_p(t)dt \nonumber \\
    & & \quad = \mathcal O\left(\sup_t \frac{\llangle \Lambda_0|(\partial_t-\imath \mathfrak A)|\Lambda_p\rrangle}{\lambda_p-\tilde \lambda_0} \right)
  \end{eqnarray}
  \begin{eqnarray}
    & & \int_0^t e^{-\imath(\varphi_p-\varphi_0)} \llangle \Lambda_0|g^{-1}\slashed D_x g|\Lambda_p\rrangle c_p(t)dt \nonumber \\
    & & \quad = \mathcal O\left(\sup_t \frac{\llangle \Lambda_0|g^{-1}\slashed D_x g|\Lambda_p\rrangle}{\lambda_p-\tilde \lambda_0} \right)
  \end{eqnarray}
  Finally we have $c_0(t) = 1 + \mathcal O(\max(\epsilon,1/T))$. By conservation of the norm defined by the inner product of an observer comoving with the test particle (for which $g$ is unitary (see \cite{Viennot5})), this implies that $c_p(t) = \mathcal O(\max(\epsilon,1/T))$ $\forall p\not=0$.
\end{proof}

In comparison with the strong adiabatic assumption:
\begin{equation}
  \max_{n\not=0} \sup_{t \in [0,T]} \left|\frac{\llangle \Lambda_0|\partial_t|\Lambda_n\rrangle}{\lambda_n-\tilde \lambda_0} \right| = \mathcal O(1/T)
\end{equation}
the weak adiabatic assumption consists to replace in the definition of the non-adiabatic couplings, the time derivative by the covariant derivative $\partial_t-\imath \mathfrak A$. The interpretation of this fact is obvious with the structure of $\mathfrak A$ eq.(\ref{structA}), it represents transitions of spin states without transition from the ground state of $\mathfrak X$ represented by the couple $(|\Lambda_0^0\rangle,|\Lambda_0^1\rangle)$ (a transition of noncommutative space states being $(|\Lambda_0^0\rangle,|\Lambda_0^1\rangle) \to (|\Lambda_n^0\rangle,|\Lambda_n^1\rangle)$ with $n\not= 0$). The weak adiabatic assumption states then that the small quantities are the non-adiabatic couplings minus the part associated with only spin state transitions (without space state transition); and then only space state transitions are forbidden by this assumption in accordance with the discussion of the beginning of this appendix.\\
As for the strong adiabatic approximation, crossings of the ground eigenvalue with other eigenvalue are forbidden, and so we must have a gap condition $\inf_t |\lambda_n-\tilde \lambda_0|>\gamma>0$ ($\forall n\not=0$).\\
Note that the effective ground eigenvalue $\tilde \lambda_0$ is zero in the case where the quasi-coherent state is separable. Indeed
\begin{equation}
  \llangle \Lambda_0|g^{-1}\slashed D_xg|\Lambda_0\rrangle = \langle 0_x|g^{-1}\sigma_ig|0_x\rangle \langle \lambda_0|(X^i-x^i)|\lambda_0\rangle = 0
\end{equation}
since
\begin{equation}
  \llangle \Lambda_0|X^i|\Lambda_0\rrangle = x^i \Rightarrow \langle \lambda_0|X^i|\lambda_0\rangle  = x^i
\end{equation}
  So the dynamical phase is reduced to be equal to $1$.\\

  Since $\pi_\Lambda \omega_{\Lambda *}(L_{X^i}) = \imath e_i^a \partial_a$, $\slashed D_x$ plays the role of $\imath \sigma^i e_i^a\partial_a$.
\begin{equation}
  g = e^{-\frac{\imath}{2} \int_0^t \tilde A dt} \Ted^{\int_0^t \omega dt} \Rightarrow g^{-1}\slashed D_x g = \Teg^{- \int_0^t \omega dt} \slashed D_x \Ted^{\int_0^t \omega dt}
  \end{equation}
  $g^{-1}\slashed D_x g$ plays the role of $\imath U_\omega \sigma^i e_i^a \partial_a U_\omega^{-1} = \imath U_\omega \sigma^i U_\omega^{-1} e_i^a \nabla_a$ (with $\nabla_a = \partial_a + \omega_a$). So the condition $\max_{n\not=0} \sup_{t \in [0,T]} \left|\frac{\llangle \Lambda_0|g^{-1} \slashed D_x g|\Lambda_n\rrangle}{\lambda_n-\tilde \lambda_0} \right| = \mathcal O(\epsilon)$ is the noncommutative version of negligible non-adiabatic couplings defined by space covariant derivatives. The condition (ii) is then the space counterpart of the condition (i) (dealing with a time covariant derivative), ensuring an equal treatment between time and space in the adiabatic assumptions and ensuring the covariance of the adiabatic approximation in the present context.

The weak adiabatic transport formula includes specific geometric non-adiabatic transitions:
\begin{prop}
  Let $\{|\Lambda_n\rrangle\}_n$ be the eigenvectors of $\slashed D_x$. Let $\mathbb P = \{1,...,2\dim \mathcal H\}$, and $A,\mathscr A \in \mathcal L(\ell^2(\mathbb P))$ be such that $A_{np} = -\imath \llangle \Lambda_n|\dot \Lambda_p \rrangle$ and $\mathscr A_{np} = \llangle \Lambda_n|\mathfrak A|\Lambda_p \rrangle$. We have
  \begin{equation}
    \Ted^{-\imath \int_0^t\mathfrak A dt} = \sum_{n,p} \left[\Teg^{-\imath \int_0^t Adt} \Ted^{-\imath \int_0^t(\mathscr A-A)dt} \right]_{np} |\Lambda_n \rrangle \llangle \Lambda_p|
  \end{equation}
\end{prop}
Then we have
\begin{equation}
  \Ted^{-\imath \int_0^t\mathfrak A dt}|\Lambda_0\rrangle = \sum_{n} \left[\Teg^{-\imath \int_0^t Adt} \Ted^{-\imath \int_0^t(\mathscr A-A)dt} \right]_{n0} |\Lambda_n \rrangle
\end{equation}
Remark: if $\dim \mathcal H<\infty$, $\ell^2(\mathbb P)=\mathbb C^{2\dim \mathcal H}$, otherwise $\ell^2(\mathbb P) = \ell^2(\mathbb N^*)$ (space of square summable series).

\begin{proof}
  Let $X \in \mathcal L(\ell^2(\mathbb P))$ be such that $\Ted^{-\imath \int_0^t \mathfrak Adt} = \sum_{n,p} \left[\Teg^{-\imath \int_0^t Xdt}\right]_{np}|\Lambda_n\rrangle \llangle \Lambda_p|$. After the derivation with respect to $t$ of this equation and projection on the right onto $|\Lambda_p\rrangle$ and on the left onto $\llangle \Lambda_n|$, we find
  \begin{eqnarray}
    -\imath \sum_k \left[\Teg^{-\imath \int_0^t Xdt}\right]_{nk} \mathscr A_{kp} & = & - \imath \sum_k X_{nk} \left[\Teg^{-\imath \int_0^t Xdt}\right]_{kp} \nonumber \\
    & & + \sum_k \left[\Teg^{-\imath \int_0^t Xdt}\right]_{kp} \llangle \Lambda_n|\dot \Lambda_k \rrangle \nonumber \\
    & & + \sum_k \left[\Teg^{-\imath \int_0^t Xdt}\right]_{nk} \llangle \dot \Lambda_k|\Lambda_p \rrangle
  \end{eqnarray}
  $\llangle \Lambda_k|\Lambda_p \rrangle = \delta_{kp} \Rightarrow \llangle \dot \Lambda_k|\Lambda_p \rrangle = -\llangle \Lambda_k|\dot \Lambda_p \rrangle$, it follows that
  \begin{equation}
    \Teg^{-\imath \int_0^t Xdt} \mathscr A = X \Teg^{-\imath \int_0^t Xdt} + A \Teg^{-\imath \int_0^t Xdt} - \Teg^{-\imath \int_0^t Xdt} A
  \end{equation}
  and then
  \begin{equation}
    X = A + \Teg^{-\imath \int_0^t Xdt} (\mathscr A - A) (\Teg^{-\imath \int_0^t Xdt})^{-1}
  \end{equation}
  and then $\Teg^{-\imath \int_0^t Xdt} = \Teg^{-\imath \int_0^t Adt} \Ted^{-\imath \int_0^t(\mathscr A-A)dt}$.
\end{proof}
The transitions associated with the spin state changes (Thomas and de Sitter precessions) are then essentially governed by the ``matrix'' $\mathscr A$, representation of $\mathfrak A$ the operator of spin state transitions without space state transition. 

\section{About the noncommutative geometry, diffeomorphism gauge changes and torsion} \label{NCgeom}
\subsection{Noncommutative gauge potential}
Suppose that $\exists \Lambda^0(x),\Lambda^1(x) \in \Env(\mathfrak X)$ and $\exists |\Omega \rangle \in \mathcal H$ such that
\begin{equation}
  |\Lambda(x) \rrangle = |0\rangle \otimes \Lambda^0(x)|\Omega \rangle + |1\rangle \otimes \Lambda^1(x)|\Omega \rangle
\end{equation}
(this is the case for examples treated section \ref{examples}). $|0\rangle \otimes \Lambda^0 + |1\rangle \otimes \Lambda^1 \in \mathbb C^2 \otimes \Env(\mathfrak X)$ appears as a section of a noncommutative spinorial bundle over $\mathscr M$ which can be viewed as a Dirac field onto $\mathscr M$. The pull-back of the geometric phase generator:
\begin{equation}
  \omega_\Lambda^* \pi_\Lambda^* A = \langle \Omega|\Lambda_\alpha^\dagger \partial_a \Lambda^\alpha |\Omega \rangle \gamma^{ab} \delta_{ji} \frac{\partial x^j}{\partial u^b} \dnc X^i \in \Omega^1_{\Der}\mathfrak X \otimes \mathcal C^{\infty}(M_\Lambda)
\end{equation}
  can be viewed as a connection onto an associated noncommutative $U(1)$-bundle on $\mathscr M$. In accordance with the fact that $\mathscr M$ is a quantisation of a flat space, this $U(1)$-bundle is flat, indeed by using the Koszul formula we have
\begin{eqnarray}
  \dnc \omega_\Lambda^* \pi_\Lambda^* A(L_{X^i},L_{X^j}) & = & L_{X^i} \omega_\Lambda^* \pi_\Lambda^* A(L_{X^j}) - L_{X^j} \omega_\Lambda^* \pi_\Lambda^* A(L_{X^i}) \nonumber \\
  & & \quad - \omega_\Lambda^* \pi_\Lambda^* A([L_{X^i},L_{X^j}]) \\
  & = & (\omega_\Lambda^* \pi_\Lambda^* A)_k ([X^i,[X^j,X^k]]+[X^j,[X^k,X^i]] \nonumber \\
  & & \quad +[X^k,[X^i,X^j]]) \\
  & = & 0
\end{eqnarray}
by the Jacobi identity in the Lie algebra $\mathfrak X$, and so $\omega_\Lambda^* \pi_\Lambda^* A$ is $\dnc$-closed. In the same manner we have
\begin{eqnarray}
  \omega_\Lambda^* \pi_\Lambda^*(\gamma_{ab}du^adu^b) & = & \gamma_{ab} \gamma^{ac} \gamma^{bd} \delta_{ik}\delta_{jl} \frac{\partial x^k}{\partial u^c} \frac{\partial x^l}{\partial u^d} \dnc X^i \otimes \dnc X^j \\
  & = & \delta_{ij} \dnc X^i \otimes \dnc X^j
\end{eqnarray}
in accordance with the fact that $\mathscr M$ is the quantisation of a flat space.\\

\subsection{Diffeomorphism gauge changes}
Let $\varphi \in \Diff(M_\Lambda)$ be a diffeomorphism of the eigenmanifold (we consider here only diffeomorphisms leaving invariant the foliation). As expected, the dual triads transform under diffeomorphism gauge changes as $(\varphi^* \tilde e)^i_a = \tilde e^i_b(\varphi(u)) \frac{\partial \varphi^b}{\partial u^a}$ whereas $\tilde e^i_0$ is invariant. It is interesting to consider the behaviour of the shift vector (generator of the Berry phase) under a diffeomorphism gauge change. Basically, we have $\varphi^* A = A_a(\varphi(u)) \frac{\partial \varphi^a}{\partial u^b} du^b$. But suppose that their exists a displacement operator defined by:
\begin{equation}
  \forall u,v \in M_\Lambda, \Dis(v,u) \in U(\mathcal H), \text{ such that }\Dis(v,u)|\Lambda(u)\rrangle = |\Lambda(v) \rrangle
\end{equation}
(this is the case for examples treated section \ref{examples} when the quasi-coherent state is separable). The possibility of the existence of a displacement operator for quasi-coherent states is stated by analogy with the case of coherent states \cite{Perelomov,Puri}. In that case:
\begin{equation}
  A_a(v)dv^a =  A_a(u)du^a -\imath \llangle \Lambda(u)|(\Dis^{-1}d_{(2)} \Dis)|\Lambda(u) \rrangle
\end{equation}
(where $d_{(2)}$ is the exterior derivative of $M_\Lambda^2$). It follows that $\varphi^* A = A + \underline \eta$ with
\begin{equation}
  \underline \eta(u) = - \imath \llangle \Lambda(u)|(\Dis(\varphi(u),u)^\dagger d\Dis(\varphi(u),u))|\Lambda(u) \rrangle
\end{equation}
The gauge potential-transformation $\underline \eta$ is an element of the 2-connection of a categorical $U(1)$-bundle \cite{Viennot8} built over the base category $\mathcal M_\Lambda$ where the objects $\mathrm{Obj}(\mathcal M_\Lambda) = M_\Lambda$ are the points of the eigenmanifold and where the arrows are $\mathrm{Morph}(\mathcal M_\Lambda) = M_\Lambda \times \Diff(M_\Lambda)$; the source, target and identity maps being defined by:
\begin{equation}
  s(u,\varphi) = u, \quad t(u,\varphi)=\varphi(u), \quad \id_u = (u,\id_{M_\Lambda})
\end{equation}
and the composition of arrows being defined by
\begin{equation}
  (\varphi_1(u),\varphi_2) \circ (u,\varphi_1) = (u,\varphi_2 \circ \varphi_1)
\end{equation}
$\Diff(M_\Lambda)$ can be then viewed as the arrow field space over $M_\Lambda$. Since $\slashed D_y = \slashed D_x + \sigma_i(x^i-y^i)$, we have
\begin{equation}
  \llangle \Lambda(\varphi(u))|\sigma_i|\Lambda(u)\rrangle(x^i(\varphi(u))-x^i(u))= 0
\end{equation}
 $\llangle \Lambda(\varphi(u))|\sigma^i|\Lambda(u)\rrangle \partial_i$ is then a vector normal to the arrow $(u,\varphi)$ (generalising the fact that $\llangle \Lambda(u)|\sigma^i|\Lambda(u)\rrangle \partial_i$ is a vector normal to $M_\Lambda$ at $u$). The categorical structure of $\mathcal M_\Lambda$ is not external to the noncommutativity of $\mathscr M$, indeed
\begin{eqnarray}
  x^i(\varphi(u))-x^i & = & \llangle \Lambda(u)|(\Dis^\dagger X^i \Dis - X^i)|\Lambda(u) \rrangle \\
  & = & \omega_\Lambda (\Dis^\dagger[X^i,\Dis]) \\
  & = & \imath \omega_\Lambda (\mathbf{A}_{\Dis}(L_{X^i}))
\end{eqnarray}
where $\mathbf{A}_{\Dis} = -\imath \Dis^\dagger \dnc \Dis \in \Omega^1_\Der \mathfrak X$ is a Berry-like gauge potential of a noncommutative bundle over $\mathscr M$ ($\Dis \in \Env(\mathfrak X)$, $\dnc$ is anew the Koszul noncommutative derivative). This can be generalised to a noncommutative spinor $\Dis \in \mathbb C^2 \otimes \Env(\mathfrak X)$ with $\Dis = |0\rangle\langle 0|\otimes \Dis^0 + |1\rangle\langle 1|\otimes \Dis^1$ and $\Dis^a(v,u)|\Lambda^a(u)\rangle = |\Lambda^a(v)\rangle$ ($a \in \{0,1\}$) (this is the case for examples treated section \ref{examples} when the quasi-coherent state is entangled), and $\mathbf{A}_{\Dis} = -\imath \Dis_a^\dagger \dnc \Dis^a$.

\subsection{Noncommutative torsion} \label{NCtorsion}
By construction, we have
\begin{equation}
  \langle du^a, \pi_\Lambda \omega_{\Lambda *}(L_{X^i}) \rangle_M = \omega_\Lambda\left(\langle \pi^*_\Lambda \omega^*_\Lambda(du^a),L_{X^i} \rangle_{\mathfrak X} \right)
\end{equation}
\begin{equation}
  \iff e^a_i = \omega_\Lambda(E^a_i) \text{ with } E^a_i = \delta_{il} \Theta^{lj} \delta_{jk} \gamma^{ab} \frac{\partial u^k}{\partial u^b}
\end{equation}
with $\Theta^{lj} = \Theta(X^l,X^j) = -\imath[X^l,X^j]$. $\{E^a_i\}_{a,i}$ can be viewed as the ``noncommutative frame'' in $\mathscr M$ associated with the frame $\{e^a_i\}_{a,i}$ of the eigenmanifold $M_\Lambda$. $\mathscr M$ is the quantisation of a flat space, but it is not necessary torsion free. We introduce the following analogue of the Weitzenb\"ock connection \cite{Aldrovandi} by:
\begin{equation}
  \mathring{\Gamma}^a_{ij} = -\imath L_{X^j}(E^a_i)
\end{equation}
which defines an analogue of the  Weitzenb\"ock torsion:
\begin{eqnarray}
  \mathring T^a_{ij} & = & -\imath L_{X^j}(E^a_i) +\imath L_{X^i}(E^a_j) \\
  & = & -\imath ([X_j,\Theta_{im}]-[X_i,\Theta_{jm}]) \gamma^{ab} \frac{\partial x^k}{\partial u^b} \\
  & = & -\imath [X^m,\Theta_{ij}] \delta_{mk} \gamma^{ab} \frac{\partial x^k}{\partial u^b}
\end{eqnarray}
where we have used the Jacobi identity $[X_j,[X_i,X_m]]+[X_i,[X_m,X_j]]=-[X_m,[X_j,X_i]]$. $\mathring T^a_{ij} \partial_a = \pi_\Lambda(\mathring T^k_{ij} \partial_k)$ with $\mathring T^k_{ij} = -\imath [X^k,\Theta_{ij}] = -\imath [X^k,[X_i,X_j]]$ the torsion of $\mathscr M$.\\
Note the difference with the torsion introduced section \ref{torsion}, $\omega_{\Lambda}(\mathring T^k_{ij})$ is the torsion of the embedding space $\mathbb R^3$ whereas $T^a_{bc}$ is the torsion of the eigenmanifold $M_\Lambda$. We can view $T$ has a torsion intrinsic to $M_\Lambda$ (link to $\tilde \theta=F^{-1}$ and then to the Berry curvature) whereas $\mathring T$ is an extrinsic torsion (linked to $\Theta$ and then to the noncommutative character of $\mathscr M$).\\

In \cite{Steinacker} a torsion in a IKKT matrix model is introduced. A direct comparison of the two models is difficult because of two reasons. First, the time is quantised in the IKKT model and not in the BFSS model. The comparison must therefore be limited to the space part. The model presented in \cite{Steinacker2} is based on an irreducible representation of $\mathfrak{so}(4,2)$ on a Hilbert space $\mathcal H_N$ (with $N \in \mathbb N$) where $\mathcal L(\mathcal H_N)$ is interpreted as a quantised version of $\mathcal C^\infty(\mathbb CP^{1,2})$, $\mathbb CP^{1,2}$ being the total space of a fibre bundle on a classical spacetime with fibre $\mathbb S^2$ (the sphere). $\mathfrak a \otimes \Env(\mathfrak X)$ plays in the present model the same role than the quantisation of $\mathcal C^\infty(\mathbb CP^{1,2})$, but with a single spin sector (corresponding to $s=\frac{1}{2}$) whereas in the model of \cite{Steinacker2} the quantisation of $\mathcal C^\infty(\mathbb CP^{1,2})$ is a graded algebra with an infinite number of spin sectors. At the semi-classical limit $N \to \infty$, the Weitzenb\"ok torsion found in \cite{Steinacker2} is (in the notations and the conventions of the present paper) $\hat T^m_{ij} = \{x^m,\hat \Theta_{ij}\}$ where $-\imath[\cdot,\cdot] \to \{\cdot,\cdot\}$ is the Poisson bracket obtained at the semi-classical limit. This expression is very similar to our Weitzenb\"ock torsion of $\mathscr M$, except that $\hat \Theta_{ij} = -\imath [Z_i,Z_j]$ where $Z_i$ is in the quantisation of $\mathcal C^\infty(\mathbb CP^{1,2})$ (analogous of $\mathfrak a \otimes \Env(\mathfrak X)$) and is not $X_i \in \mathfrak X$. More precisely, by choosing $Z_a = \hat \tau_a + \hat {\mathfrak A}_a$, where $\hat \tau_a$ is the coordinates of the fibres $\mathbb S^2$ and $\hat {\mathfrak A}_a$ is viewed as a perturbation of the cosmic background, we have \cite{Steinacker2}
\begin{equation}
  \hat \Theta_{ab} = \kappa \theta_{ab} + \{\hat \tau_b,\hat {\mathfrak A}_a\} - \{\hat \tau_a,\hat {\mathfrak A}_b\} +\{\hat {\mathfrak A}_b,\hat {\mathfrak A}_a\}
\end{equation}
where $\kappa$ is a constant parameter of the model and $\theta_{ab} = \{x_a,x_b\}$ is the semi-classical limit of $-\imath[X_a,X_b]$. The first part of the torsion $\kappa \{x^m,\theta_{ij}\}$ is completely analogous to our Weitzenb\"ock torsion of $\mathscr M$ $\mathring T^m_{ij}$. This part can be viewed as the torsion of the cosmic background in the interpretation of \cite{Steinacker2}. In our model, the embedding space $\mathbb R^3$ (or its quantisation $\mathscr M$) are well a background in which the eigenmanifold lives. $\mathring T$ has then well a similar interpretation in our framework. The second part of the torsion $\{x^m, \{\hat \tau_b,\hat {\mathfrak A}_a\} - \{\hat \tau_a,\hat {\mathfrak A}_b\}\}$ has a structure similar to our torsion of $M_\Lambda$:
\begin{equation}
T^m_{ab} = \frac{1}{2} \mathrm{Im} \tr\left(\sigma^m([\tau_b,\mathfrak A_a]-[\tau_a,\mathfrak A_b])\right)
\end{equation}
We can think that it plays the same role, but with $T^m_{ab}$ limited to a single spin sector and with commutators not directly affected by the semi-classical limit ($\tau_a$ and $\mathfrak A_a$ are matrices of $\mathfrak a = \mathfrak M_{2 \times 2}(\mathbb C)$, their sizes do not change with the increase of $N$). We cannot see $\mathfrak A$ as a perturbation. As the geometric phase generator it characterises the geometry of the adiabatic bundle with connection in which the transport is defined. However we can note that perturbation theory and adiabatic transport have strong similarities. For example, under a perturbation $\epsilon V$ an eigenvalue becomes $\lambda_0 + \epsilon \llangle \Lambda_0|V|\Lambda_0 \rrangle + \mathcal O(\epsilon^2)$ while in adiabatic transport the phase is generated by $\lambda_0 + \frac{\imath}{T} \langle \Lambda_0 | \frac{d}{dt} |\Lambda_0\rangle_*$ (with $\mathfrak A_a \dot u^a = -\imath \langle \Lambda_0 | \frac{d}{dt} |\Lambda_0\rangle_*$). Moreover the non-adiabatic couplings $\frac{\llangle \Lambda_0|\partial_t|\Lambda_n\rrangle}{\lambda_n - \lambda_0}$ (or those for the point (ii) theorem \ref{adiabtheo}) are similar to the coefficients of a first order perturbative expansion of the eigenvector $\frac{\llangle \Lambda_n|V|\Lambda_0\rrangle}{\lambda_0-\lambda_n}$. The nonlinear last part of the torsion $\{x^m, \{\hat {\mathfrak A}_b,\hat {\mathfrak A}_a\}\}$ seems to have no analogue in our formalism. The reason is that it is a term of the second order of perturbation. But the adiabatic approximation is of first order in $1/T$. To try to find an analogue of $\{x^m, \{\hat {\mathfrak A}_b,\hat {\mathfrak A}_a\}\}$, it would be necessary to compute the first non-adiabatic corrections. The torsion in \cite{Steinacker2} seems be an analogue of the sum of our intrinsic and extrinsic torsions (a sum which cannot be directly realised without projection in our model since $\mathscr M$ and $M_\Lambda$ have not the same dimension).

\section{Higher dimension and the Koopman approach} \label{higherdim}
\subsection{Fast dynamics of $\mathscr M$}\label{fastdyn}
$\dim \mathfrak X = 3$ but the eigenmanifold $M_\Lambda$ of the noncommutative manifold $\mathscr M$ (generated by the Lie algebra $\mathfrak X$) is only 2 dimensional. We lost one dimension in the emergent geometry with respect to the dimension of the quantised flat spacetime viewed by the ideal Galilean observer. To solve this problem, we can start with an higher dimensional quantised flat spacetime as in \cite{Karczmarek}. But we have considered in this paper a dimension reduction as in \cite{Berenstein}. It seems more natural that a fifth quantum dimension (after quantisation) emerges from the analysis of $\mathscr M$ in order to recover the original four dimensions of the classical pre-quantised spacetime in the classical emergent spacetime. This fifth dimension (or more precisely a fourth coordinate operator) emerges from the dynamics of $\mathscr M$ (still generated by a 3 dimensional Lie algebra $\mathfrak X$). Indeed, until now we have considered only time independent observables $(X^i)$, but these operators obey to a noncommutative Klein-Gordon equation \cite{Zarembo,Banks}:
\begin{equation} \label{NCKGeq}
  \ddot X^i - [X_j,[X^i,X^j]] = 0
\end{equation}
But the time evolution of $(X^i)$ is not necessarily slow, and the application of the adiabatic assumption onto eq.(\ref{NCDiraceq}) can fail. To solve this problem, we can use the Schr\"odinger-Koopman approach \cite{Viennot6}. This method consists to separate slow and fast variations and applying the adiabatic assumption onto eigenvectors of an extended operator including the parameters with fast variations as new quantum variables.\\
More explicitly, let $(X^i_0)$ be the generators of $\mathfrak X$, and suppose that $X^i(t) = w^i_j(t) X^j_0$ (we suppose that $X^i \in \mathfrak X$, but generalisations with $X^i \in \Env \mathfrak X$ are possible, with for example $X^i(t) = w^i_j(t) X^j_0+w^i_{jk}(t) X^j_0X^k_0+...$). Eq.(\ref{NCKGeq}) becomes then
\begin{equation}
  \ddot w^i_k - \delta_{jl}C^{i'j'}_{m'}C^{l'm'}_{k} w^l_{l'}w^i_{i'}w^j_{j'} = 0
\end{equation}
where $C^{ij}_k$ are the structure constants of $\mathfrak X$ ($[X^i_0,X^j_0] = C^{ij}_k X^k_0$), or equivalently
\begin{eqnarray} \label{NLdyn}
  \dot w^i_k = p^i_k & \\
  \dot p^i_k = \delta_{jl}C^{i'j'}_{m'}C^{l'm'}_{k} w^l_{l'}w^i_{i'}w^j_{j'} &
\end{eqnarray}
which is a classical nonlinear dynamical system in the phase space $\Gamma = \mathbb R^{18}$. We denote by $(\xi^\alpha)_{\alpha \in \{1,...,18\}} \equiv (w^i_k,p^i_k)_{i,k\in\{1,2,3\}}$ a point of $\Gamma$. We can also consider the simplified case with $w^i_j(t) = w(t)\delta^i_j$, where $\Gamma = \mathbb R^2$ (with $\xi = (w,p)$) with $\dot w=p$ and $\dot p = \frac{1}{3} \delta_{jl} C^{ij}_{m}C^{lm}_{i} w^3$. We denote the differential equation of the nonlinear system as $\dot \xi^\alpha = \mathscr F^\alpha(\xi)$.\\
The Schr\"odinger-Koopman equation is \cite{Viennot6}:
\begin{equation} \label{SKeq}
  \imath |\dot \Psi \rrrangle = \left(\slashed D_x - \imath \mathscr F^\alpha \frac{\partial}{\partial \xi^\alpha}\right) |\Psi \rrrangle
\end{equation}
where $|\Psi \rrrangle \in \mathbb C^2 \otimes \mathcal H \otimes L^2(\Gamma,d\mu(\xi))$ (where $\mu$ is a measure on $\Gamma$), $L^2(\Gamma,d\mu(\xi))$ is the space of square integrable functions on $\Gamma$. The solutions $|\psi\rrangle$ of eq.(\ref{NCDiraceq}) are deduced from the solutions of eq.(\ref{SKeq}) by $|\psi(t)\rrangle = \langle \xi(t)|\Psi(t)\rrrangle$ where $t \mapsto \xi(t)$ is solution of eq.(\ref{NLdyn}) \cite{Viennot6}, so $|\Psi\rrrangle$ is a new representation of the quantum state including the effects of the ``noise'' associated with the fast variations. Let $X^4 \equiv - \imath \mathscr F^\alpha \frac{\partial}{\partial \xi^\alpha}$ be the generator of the Koopman operator. The time-dependence of eq.(\ref{SKeq}) being slow (because they are restricted to the time-dependence of $x$) we can apply the adiabatic assumption with $(\slashed D_x + X^4)|\tilde \Lambda \rrrangle = 0$. Let $x^4 \in \Sp(X^4)$ be a Koopman value associated with the Koopman function $f_{x^4} \in L^2(\Gamma,d\mu(\xi))$:
\begin{equation}
  X^4f_{x^4}(\xi) = x^4 f_{x^4}(\xi)
\end{equation}
  $|\Lambda\rrrangle = f_{x^4}|\tilde \Lambda \rrrangle$ is also an eigenvector of $(\slashed D_x + X^4)$ with eigenvalue $x^4$ (see \cite{Viennot6}). It follows that
\begin{equation}
  \slashed D_x^K |\Lambda \rrrangle = 0, \quad \text{with } \slashed D_x^K = \sigma_I \otimes (X^I - x^I)
\end{equation}
with $I \in \{1,2,3,4\}$ and $\sigma_4 = \sigma_0 = \id_2$. By noting that $x^4 = \langle f_{x^4}|X^4|f_{x^4}\rangle$, the situation is completely similar to the one described section \ref{sec1} with an additional dimension. $M_\Lambda = \{x \in \mathbb R^4, \text{s.t. } \det(\slashed D_x^K)=0\}$ can then be 3 dimensional. Note that:
\begin{eqnarray}
  [X^i,X^4] & = & \imath \mathscr F^\alpha(\xi) \frac{\partial X^i}{\partial x^\alpha} \\
  & = & \imath p^j_k \frac{\partial w^i_l X^l_0}{\partial w^j_k} \\
  & = & \imath p^i_k X^k_0 \in \mathfrak X
\end{eqnarray}
  Since $\sigma_4=\sigma_0$ (and is then a part of $\gamma^0$ in the Weyl representation), or in other words since $x^4$ is associated with the time of fast variations, the metric of the embedding space $\mathbb R^4$ is $d(x^1)^2+d(x^2)^2+d(x^3)^2-d(x^4)^2 \equiv -\eta_{IJ}dx^Idx^J$, i.e. the five dimensional embedding flat spacetime is Anti-de Sitter. The continuation of the discussion is similar to the main sections of this paper.\\

  With the Koopman analysis, the fourth dimension of $\mathscr M$ generated by $(X^I)_{I\in\{1,2,3,4\}}$ (the fifth dimension of the quantum spacetime) is not added to the pre-quantised flat spacetime of the classical Galilean observer (before the use of the BFSS quantisation rules $x^i \leadsto X^i$ and $\partial_i \leadsto L_{X^i}$). It emerges ``spontaneously'' from the fast evolutions of the quantum spacetime at the microscopic scale. Note that, this does not forbid emerging dynamical spacetimes. Indeed, we can have a time-dependent parametrisation as $X^i(w(t),v(t))$ with $t \mapsto w(t)$ fast evolving parameters and $t \mapsto v(t)$ slow evolving parameters. The Koopman approach is then used with $w$, but $X^i$ rests explicitly time-dependent with respect to $v$. Since the variations of $v$ are slow, we can apply the adiabatic assumption, and have a time dependent emerging manifold $M_\Lambda(v(t))$. We do not treat in more detail this case in this present paper which focus on time-independent eigenmanifold. In the present formalism, 6 compact extra dimensions can also emerge via the Koopman analysis if we consider the vacuum fluctuations perturbing eq. (\ref{NCKGeq}) as viewed in \cite{Viennot4}. The geometry of these six emergent compact dimensions is not the subject of the present paper and has been studied in \cite{Viennot4}.

\subsection{Massive test particle and chiral oscillations}\label{chiral}
At the starting point of this paper, we have considered a massless test particle to reveal the geometry. If we consider a massive fermion, the mass (denoted here by $\frac{\omega}{2}$ for convenience) induces a coupling between the two fermionic chiralities:
\begin{equation}
  \imath \gamma^0 \partial_0 |\psi \rrangle = -\gamma_i \otimes (X^i-x^i) |\psi \rrangle + \frac{\omega}{2} |\psi \rrangle
\end{equation}
\begin{equation}
  \iff \imath \partial_0 |\psi \rrangle = \left(\begin{array}{cc} 1 & 0 \\ 0 & -1 \end{array} \right) \otimes \sigma_i \otimes (X^i-x^i) |\psi \rrangle +\gamma^0 \frac{\omega}{2} |\psi \rrangle
\end{equation}
with $|\psi\rrangle \in \mathbb C^4 \otimes \mathcal H$. Let $|\tilde \psi \rrangle = e^{\imath \gamma^0 \frac{\omega}{2} t} |\psi \rrangle$ a gauge change of the fermionic state.
\begin{equation}
  e^{\imath \gamma^0 \frac{\omega}{2}t} = \left(\begin{array}{cc} \cos(\omega t/2) & -\imath \sin(\omega t/2) \\ -\imath \sin(\omega t/2) & - \cos(\omega t/2) \end{array} \right)
\end{equation}
  and we have
\begin{equation}
  \imath \partial_0 |\tilde \psi \rrangle = m(\omega t) \otimes \slashed D_x |\tilde \psi \rrangle
\end{equation}
with the mixing matrix
\begin{equation}
  m(\omega t) = \left(\begin{array}{cc} \cos(\omega t) & -\imath \sin(\omega t) \\ \imath \sin(\omega t) & -\cos(\omega t) \end{array} \right)
\end{equation}
We can solve this equation by using the Schr\"odinger-Floquet approach \cite{Viennot7} which is a particular case of the Schr\"odinger-Koopman approach for periodic dynamical systems. Let $\theta = \omega t$ and $|\Psi \rrrangle \in \mathbb C^2 \otimes \mathcal H \otimes \mathbb C^2 \otimes L^2(\mathbb S^1,\frac{d\theta}{2\pi})$ be such that
\begin{equation} \label{SchroFloEq}
  \imath |\dot \Psi \rrrangle = \left(m(\theta)\otimes \slashed D_x - \imath \omega \frac{\partial}{\partial \theta} \right) |\Psi \rrrangle
\end{equation}
$\mathbb C^2 \otimes L^2(\mathbb S^1,\frac{d\theta}{2\pi})$ is the space of square integrable functions on the circle $\mathbb S^1$ with a chirality degree of freedom. The adding operator $-\imath \omega \partial_\theta$ can be viewed as the quantum observable of a compact extra-dimension (but since the spectrum of $-\imath \omega \partial_\theta$ is discrete, this one does not involve an emerging additional dimension on the eigenmanifold). The original state is recovered by $|\tilde \psi \rrangle = \langle \theta=\omega t|\Psi \rrrangle$. We can apply the adiabatic assumption onto eq.(\ref{SchroFloEq}) with a generalised quasi-coherent state $|\Lambda_0 \rrrangle$ such that
\begin{equation}
  \imath \omega \frac{\partial}{\partial \theta} |\Lambda_0 \rrrangle = m \otimes \slashed D_x |\Lambda_0 \rrrangle
\end{equation}
Since $\imath \omega \partial_\theta e^{-\imath n \theta} = n\omega e^{-\imath \theta}$, it follows that $|\Lambda_n\rrrangle = e^{\imath n \theta}|\Lambda_0 \rrrangle$ is also an eigenvector of $m\otimes \slashed D_x - \imath \omega \partial_\theta$:
\begin{equation}
  \imath \omega \frac{\partial}{\partial \theta} |\Lambda_n \rrrangle = (m \otimes \slashed D_x -n\omega)|\Lambda_n \rrrangle \qquad n \in \mathbb Z
\end{equation}
It follows that
\begin{equation}
  |\Lambda_n(\theta)\rrrangle = e^{\imath n\theta} \Teg^{-\frac{\imath}{\omega} \int_0^\theta m(\theta) d\theta \otimes \slashed D_x} |\Lambda_n(0) \rrrangle
\end{equation}
with to ensure the continuity with $\omega \to 0$ ($n\omega$ is not negligible since $n$ can be large)
\begin{equation}
  \left(\left(\begin{array}{cc} 1 & 0 \\ 0 & -1 \end{array}\right) \otimes \slashed D_x - n\omega \right) |\Lambda_n(0)\rrrangle = 0
\end{equation}
\begin{equation}
  \iff |\Lambda_n(0)\rrrangle = \left(\begin{array}{c} |\Lambda_n \rrangle \\ |\Lambda_{-n} \rrangle \end{array}\right) \quad\text{with } \slashed D_x|\Lambda_n\rrangle = n\omega|\Lambda_n\rrangle
\end{equation}
The emergent geometry is then associated with several eigenmanifolds $M_{\Lambda,n}=\{x \in \mathbb R^3, \text{s.t. } \det(\slashed D_x-n\omega)=0\}$. The ground state $|\Lambda_0\rrangle$ is associated with the eigenmanifold $M_{\Lambda,0}$ revealed by massless fermionic test particles. With massive test particles, we have also excited states associated with eigenmanifolds $M_{\Lambda,n}$ corresponding to the Fourier modes of the chiral oscillations.\\
Note that by application of the Floquet theorem, we have $U_x(\theta)=\Teg^{-\frac{\imath}{\omega} \int_0^\theta m(\theta) d\theta \otimes \slashed D_x} = Z_x(\theta)e^{\imath M_x \theta}$ where $Z_x$ is a $2\pi$-periodic operator (with $Z_x(0)=\id$) and $M_x$ is a $\theta$-independent operator. The monodromy operator $e^{\imath M_x \theta}$ governs the general dynamics induced by the chiral oscillations (without the transient fluctuations described by $Z_x(\theta)$).\\

The emergent geometry depends on the particle mass, because this one warps the spacetime at the microscopic scale (in contrast with a test particle for the geometry at the macroscopic scale where this effect is negligible). The Berry phase generator is
\begin{eqnarray}
  A_n & = & -\imath \lllangle \Lambda_n|d|\Lambda_n \rrrangle \\
  & = & -\imath \llangle \Lambda_n|d|\Lambda_n \rrangle -\imath \llangle \Lambda_{-n}|d|\Lambda_{-n} \rrangle \nonumber \\
  & & \quad -\imath \lllangle \Lambda_n(0)| \left(\int_0^{2\pi} U_x(\theta)^{-1} d U_x(\theta) \frac{d\theta}{2 \pi} \right)|\Lambda_n(0)\rrrangle
\end{eqnarray}
(where $d$ is the exterior derivative of $M_{\Lambda,n}$ - no derivation with respect to $\theta$ -). $-\imath \llangle \Lambda_{\pm n}|d|\Lambda_{\pm n} \rrangle$  represents the general geometry modified by the mass $\omega/2$ added to the contents of the spacetime.
\begin{equation}
  A_{n,loc}=-\imath \lllangle \Lambda_n(0)| \left(\int_0^{2\pi} U_x(\theta)^{-1} d U_x(\theta) \frac{d\theta}{2 \pi} \right)|\Lambda_n(0)\rrrangle
\end{equation}
  $A_{n,loc}$ represents the local deformation of the spacetime around the localised test particle of mass $\omega/2$. For particle with small mass, this term is negligible (since $\imath \omega \partial_\theta$ is then just a perturbation operator). For large mass $\omega/2$, the local correction is
\begin{eqnarray}
  A_{n,loc} & = & \frac{1}{\omega} \int_0^{2\pi} \lllangle \Lambda_n(0)|\int_0^\theta m(\theta)d\theta\otimes \sigma_i|\Lambda_n(0)\rrrangle \frac{d\theta}{2\pi} \frac{\partial x^i}{\partial u^a} du^a \nonumber \\
  & & \qquad \qquad + \mathcal O(1/\omega^2) \\
  & = & \frac{\imath}{\omega}(\llangle \Lambda_{n}|\sigma_i|\Lambda_{-n} \rrangle - \llangle \Lambda_{-n}|\sigma_i|\Lambda_{n} \rrangle)\frac{\partial x^i}{\partial u^a} du^a + \mathcal O(1/\omega^2)
\end{eqnarray}

\section{Computation of the quasi-coherent states}
\subsection{CCR algebra}
\subsubsection{case $z=0$:}\label{CCR}
Let $\Lie(a,a^+,\id)$ be the CCR algebra. The non-commutative plane is defined by $X^1 = \frac{a+a^\dagger}{2}$, $X^2=\frac{a-a^\dagger}{2\imath}$ and $X^3=0$. By using the Block matrix determinant formula $\left|\begin{array}{cc} A & B \\ C & D \end{array} \right| = \det(A-BD^{-1}C)\det D$ \cite{Powell} we have (with $\alpha = x^1+\imath x^2$)
\begin{equation}
  \det \slashed D_x = \left|\begin{array}{cc} -x^3 & a^\dagger-\bar \alpha \\ a - \alpha & x^3 \end{array} \right| = \det(-(x^3)^2 - (a^\dagger - \bar \alpha)(a-\alpha))
\end{equation}
It follows that $\det \slashed D_x = 0$ if and only if $\exists |\Omega \rangle$ such that $(a^\dagger-\bar \alpha)(a-\alpha)|\Omega \rangle = - (x^3)^2|\Omega \rangle$. But $\langle \Omega|(a^\dagger-\bar \alpha)(a-\alpha)|\Omega \rangle= \|(a-\alpha)\Omega\|^2 \geq 0$ implies that $x^3=0$.\\
$(a^\dagger-\bar \alpha)(a-\alpha)|\Omega \rangle = 0 \Rightarrow a|\Omega\rangle = \alpha |\Omega\rangle$ or $a^+(a-\alpha)|\Omega\rangle = \bar \alpha (a-\alpha)|\Omega\rangle$. The second alternative is impossible since $\Sp(a^+) = \varnothing$. We have then $|\Omega \rangle = |\alpha \rangle$ (coherent state of the CCR algebra \cite{Perelomov,Puri}).\\

\begin{equation}
  \slashed D_x|\Lambda\rrangle = 0 \iff \left(\begin{array}{cc} 0 & a^+-\bar \alpha \\ a-\alpha & 0 \end{array} \right)\left(\begin{array}{c} |\Lambda^0\rangle \\ |\Lambda^1 \rangle \end{array} \right) = 0
\end{equation}
implying that $a|\Lambda^0\rangle = \alpha|\Lambda^0 \rangle \Rightarrow |\Lambda^0\rangle = |\alpha\rangle$ and $a^+|\Lambda^1\rangle = \bar \alpha |\Lambda^1 \rangle \Rightarrow |\Lambda^1\rangle = 0$. Finally
\begin{equation}
  |\Lambda \rrangle = |0\rangle \otimes |\alpha \rangle\qquad (\alpha=x^1+\imath x^2)
\end{equation}

\subsubsection{case $z\not=0$:}\label{CCR2}
Let $\Lie(a,a^+,\id)$ be the CCR algebra. The non-commutative manifold is defined by $X^1 = \frac{a+a^\dagger}{2}$, $X^2=\frac{a-a^\dagger}{2\imath}$ and $X^3=\xi \id$ (with $\xi \in \mathbb R^*$).
\begin{equation}
  \det \slashed D_x = \left|\begin{array}{cc} \xi-x^3 & a^\dagger-\bar \alpha \\ a - \alpha & \xi+x^3 \end{array} \right| = \det(\xi^2-(x^3)^2 - (a^\dagger - \bar \alpha)(a-\alpha))
\end{equation}
$\det \slashed D_x=0$ if and only if $\exists |\Omega \rangle$ such that $(a^\dagger-\bar \alpha)(a-\alpha)|\Omega \rangle = (\xi^2- (x^3)^2)|\Omega \rangle$. It follows that $|\xi|\geq |x^3|$. Let $b_\alpha = a-\alpha$ and $b_\alpha^+ = a^+ - \bar \alpha$ be the operators of another CCR algebra : $[b_\alpha,b^+_\alpha]=[a,a^+]=\id$. We have then $b^+_\alpha b_\alpha|\Omega \rangle = (\xi^2-(x^3)^2)|\Omega \rangle$. $|\Omega\rangle$ is then an eigenvector of $b^+_\alpha b_\alpha$, and then $\xi^2-(x^3)^2=n$ with $n \in \mathbb N$ ($b^+_\alpha b_\alpha |n\rangle_\alpha = n|n\rangle_\alpha$).
\begin{equation}
  b_\alpha|0\rangle_\alpha = 0 \iff a|0\rangle_\alpha = \alpha|0\rangle_\alpha \iff |0\rangle_\alpha = |\alpha \rangle
\end{equation}
\begin{equation}
  |n\rangle_\alpha = \frac{(b_\alpha^+)^n}{\sqrt{n!}}|0\rangle_\alpha = \frac{(a^+-\bar \alpha)^n}{\sqrt{n!}} |\alpha \rangle
\end{equation}
We fix the value $n$.
\begin{equation}
  \slashed D_x |\Lambda_n \rrangle = 0 \iff \left(\begin{array}{cc} 0 & b^+_\alpha \\ b_\alpha & 0 \end{array} \right)\left(\begin{array}{c} |\Lambda^0_n\rangle \\ |\Lambda^1_n\rangle \end{array} \right) = \left(\begin{array}{c} (\xi+x^3)|\Lambda^0_n\rangle \\ (\xi-x^3) |\Lambda^1_n\rangle \end{array} \right)
\end{equation}
We have $|\Lambda^1_n\rangle = \frac{b_\alpha}{\xi-x^3}|\Lambda^0_n\rangle$, and then $b^+_\alpha b_\alpha|\Lambda^0_n\rangle = (\xi^2-(x^3)^2)|\Lambda^0_n\rangle = n|\Lambda^0_n\rangle$. It follows that $|\Lambda^0_n \rangle = |n\rangle_\alpha$. In a same way, $|\Lambda^0_n\rangle = \frac{b_\alpha^+}{\xi+x^3}|\Lambda^1_n\rangle$, and then $b_\alpha b^+_\alpha|\Lambda^1_n\rangle = (b^+_\alpha b_\alpha +1)|\Lambda^1_n\rangle = n|\Lambda^1_n\rangle \Rightarrow |\Lambda^1_n\rangle = |n-1\rangle_\alpha$. Finally
\begin{equation}
  |\Lambda_n \rrangle = |0\rangle \otimes |n\rangle_\alpha + |1\rangle \otimes |n-1\rangle_\alpha \qquad (n \in \mathbb N^*)
\end{equation}
(with $n=0$ we find the result of the previous case). 

\subsection{$\mathfrak{su}(2)$ algebra} \label{su2}
Let $\Lie(J^i$) be the $\mathfrak{su}(2)$ algebra and implicitly its irreducible representation of dimension $2j+1$ ($\delta_{ij}J^i J^j = j(j+1)$). The Fuzzy sphere is defined by $X^i = rJ^i$ with $r \in \mathbb R^{+*}$ a parameter.
\subsubsection{Solutions at the poles:}
firstly, we search solution of eq.(\ref{NCEVeq}) with $x^1=x^2=0$. By using the Block matrix determinant formula $\left|\begin{array}{cc} A & B \\ C & D \end{array} \right| = \det(A-BD^{-1}C)\det D$ \cite{Powell} we have
\begin{eqnarray}
  \det (\slashed D_x/r-\lambda) & = & \left|\begin{array}{cc} J^3-\frac{x^3}{r}-\lambda & J^- \\ J^+ & -J^3+\frac{x^3}{r}-\lambda \end{array} \right| \\
  & = & \det(J^3-\frac{x^3}{r}-\lambda-J^-(-J^3+\frac{x^3}{r}-\lambda)^{-1}J^+) \nonumber \\
  & & \qquad \times \det(-J^3+\frac{x^3}{r}-\lambda)
\end{eqnarray}
$\lambda \in \Sp(\slashed D_x/r)$ if $\exists|\Omega\rangle$ such that $(J^3-\frac{x^3}{r}-\lambda)|\Omega\rangle=J^-(-J^3+\frac{x^3}{r}-\lambda)^{-1}J^+|\Omega\rangle$. We write $|\Omega\rangle = \sum_{m=-j}^j c_m|jm\rangle$ and then
\begin{equation}
  \sum_{m=-j}^{+j} c_m(m-\frac{x^3}{r}-\lambda)|jm\rangle = \sum_{m=-j}^{+j} \frac{c_m(j(j+1)-m(m+1))}{\frac{x^3}{r}-\lambda-m-1} |jm \rangle
\end{equation}
This implies that
\begin{eqnarray}
  & & (m-\frac{x^3}{r}-\lambda)(\frac{x^3}{r}-\lambda-m-1) = j(j+1)-m(m+1) \label{su2polychara}\\
  & \iff & \lambda^2+\lambda-j(j+1)+(2m+1)\frac{x^3}{r} - \left(\frac{x^3}{r}\right)^2 = 0
\end{eqnarray}
It follows that
\begin{equation}
  \lambda_{m\pm} = -\frac{1}{2} \pm \frac{1}{2}\sqrt{1+4j(j+1)-4(2m+1)\frac{x^3}{r}+4\left(\frac{x^3}{r}\right)^2}
\end{equation}
$\slashed D_x |\lambda_0 \rrangle = r\lambda_{M\pm}|\lambda_0 \rrangle$ if
\begin{eqnarray}
  (J^3-\frac{x^3}{r})|\lambda^0_0\rangle + J^-|\lambda^1_1\rangle & = & \lambda_{M\pm}|\lambda^0_0\rangle \\
  J^+|\lambda^0_0\rangle -(J^3-\frac{x^3}{r})|\lambda^1_0\rangle & = & \lambda_{M\pm}|\lambda^1_0\rangle
\end{eqnarray}
by writing $|\lambda_0 \rrangle = |0\rangle \otimes |\lambda^0_0\rangle + |1\rangle \otimes |\lambda^1_0\rangle$ (with $(|0\rangle,|1\rangle)$ the canonical basis of $\mathbb C^2$). We set $|\lambda^0_0\rangle = \sum_{m=-j}^{+j} a_m|jm\rangle$ and $|\lambda^1_0\rangle = \sum_{m=-j}^{+j} b_m|jm\rangle$. We have then
\begin{eqnarray}
  & & \sum_{m=-j}^{+j} a_m(m-\frac{x^3}{r}-\lambda_{M\pm})|jm\rangle \nonumber \\ & & \quad =- \sum_{m=-j+1}^{+j} b_m\sqrt{j(j+1)-m(m-1)}|jm-1\rangle  \\
  & & \sum_{m=-j}^{j-1} a_m \sqrt{j(j+1)-m(m+1)}|jm+1\rangle \nonumber \\ & & \quad = \sum_{m=-j}^{+j} b_m(m-\frac{x^3}{r}-\lambda_{M\pm})|jm\rangle 
\end{eqnarray}
and then
\begin{eqnarray}
  & & a_m(m-\frac{x^3}{r}-\lambda_{m\pm})+b_{m+1}\sqrt{j(j+1)-m(m+1)} = 0 \\ & & \qquad \qquad (m\not=j) \nonumber \\
  & & a_j(j-\frac{x^3}{r}) = 0 \\
  & & a_{m-1}\sqrt{j(j+1)-m(m-1)} - b_m(m-\frac{x^3}{r}+\lambda_{M\pm}) = 0 \\ &  & \qquad \qquad  (m\not=-j) \nonumber \\
  & & b_{-j}(-j-\frac{x^3}{r}) = 0
\end{eqnarray}
$a_j=b_{-j}=0$ and $b_{m+1} = a_m \frac{\sqrt{j(j+1)-m(m+1)}}{m+1-\frac{x^3}{r}+\lambda_{M\pm}}$, it follows that
\begin{equation}
  \scriptstyle a_m\left((m-\frac{x^3}{r}-\lambda_{M\pm})(m+1-\frac{x^3}{r}+\lambda_{M\pm})+j(j+1)-m(m+1)\right) = 0
\end{equation}
Because of eq.(\ref{su2polychara}), we have $a_m=0$ except if $m=M$; and then $b_m=0$ except if $m=M+1$.
\begin{eqnarray}
  & & a_M(M+\frac{x^3}{r}-\lambda_{M\pm})+b_{M+1}\sqrt{j(j+1)-M(M+1)} = 0 \\
  & & a_M \sqrt{j(j+1)-M(M+1)} - b_{M+1}(M+1-\frac{x^3}{r}+\lambda_{M\pm}) = 0
\end{eqnarray}
and then $a_M = M+1-\frac{x^3}{r}+\lambda_{M\pm}$ and $b_{M+1} = \sqrt{j(j+1)-M(M+1)}$.

\subsubsection{Solutions in all directions:}
We have $|\lambda_0 \rrangle = k^0|0\rangle \otimes |jm\rangle + k^1|1\rangle \otimes |j,m+1\rangle$ with $k^0 = \kappa (\lambda_{m\pm}+m+1-\frac{|x|}{r})$ and $k^1 = \kappa \sqrt{j(j+1)-m(m+1)}$ ($\kappa$ being just a normalisation factor such that $\|k\|^2 = 1$). Let $\vec n = (\sin\theta \cos \varphi, \sin \theta \sin \varphi, \cos \theta)$ be the direction vector in $\mathbb R^3$ and $\Dis_j(\vec n) = e^{\imath \theta \vec m \cdot \vec J}$ be the $\mathfrak{su}(2)$ displacement operator (see \cite{Perelomov, Puri}) with $\vec  m = (\sin \varphi, -\cos \varphi,0)$ ($\Dis_j(\vec n)J^3\Dis_j(\vec n)^\dagger = \vec n \cdot \vec J$). Let $R$ be the following rotation matrix
\begin{equation}
  R = \left(\begin{array}{ccc} \cos \theta \cos \varphi & \cos \theta \sin \varphi & - \sin \theta \\
    - \sin \varphi & \cos \varphi & 0 \\
    \sin \theta \cos \varphi & \sin \theta \sin \varphi & \cos \theta
  \end{array} \right)
\end{equation}
we have
\begin{equation}
  \Dis_j(\vec n) J^i \Dis_j(\vec n)^\dagger = {R^i}_k J^k
\end{equation}
We have then
\begin{eqnarray}
  \delta_{ik} \Dis_{1/2}(\vec n)\sigma^i\Dis_{1/2}(\vec n)^\dagger \otimes \Dis_j(\vec n)J^k\Dis_j(\vec n)^\dagger & = & \delta_{ik} {R^i}_l \sigma^l \otimes {R^k}_m J^m \\
  & = & (R\vec \sigma) \odot (R\vec J) \\
  & = & (R^T R \vec \sigma) \odot \vec J \\
  & = & \delta_{ik} \sigma^i \otimes J^k
\end{eqnarray}
($\odot$ stands for the scalar product of $\mathbb R^3$ along with the tensor product between operators). It follows that
\begin{eqnarray}
  & & \Dis_{1/2 \otimes j}(\vec n) \sigma_i \otimes (J^i-\frac{|x|}{r}\delta^{i3}) \Dis_{1/2\otimes j}(\vec n)^\dagger \nonumber \\
  & & \qquad \quad =  \sigma_i \otimes J^i - \frac{|x|}{r} \Dis_{1/2}(\vec n) \sigma_3 \Dis_{1/2}(\vec n)^\dagger \\
  & & \qquad \quad = \sigma_i \otimes (J^i - \frac{|x|n^i}{r} )
\end{eqnarray}
(with $\Dis_{1/2\otimes j} \equiv \Dis_{1/2}\otimes \Dis_j$). Finally we have
\begin{equation}
  \slashed D_x \Dis_{1/2 \otimes j}(\vec n) |\lambda_0 \rrangle = r\lambda_{m\pm} \Dis_{1/2\otimes j}|\lambda_0 \rrangle
\end{equation}
with $\vec n = \frac{\vec x}{\|\vec x\|}$ and $\lambda_{m\pm} = -\frac{1}{2}\pm\frac{1}{2}\sqrt{(2j+1)^2-4\frac{x}{r} (2m+1)+4\frac{x^2}{r^2}}$. We can have $\lambda_{m\pm}=0$ in only two cases
\begin{eqnarray}
  \lambda_{j+}=0 & \iff & |x|=rj \\ \lambda_{-j-1,+} = 0 & \iff & |x|=-rj
\end{eqnarray}
we choose $|x|=rj \iff \vec x = rj(\sin \theta \cos \varphi,\sin \theta \sin \varphi,\cos \theta)$, the quasi-coherent state is then
\begin{equation}
  |\Lambda(x)\rrangle = \Dis_{1/2\otimes j}(\vec n) |0\rangle \otimes |jj\rangle = |\zeta \rangle_{1/2} \otimes |\zeta\rangle_j
\end{equation}
with $\zeta = e^{\imath \varphi} \tan\frac{\theta}{2}$ and $|\zeta\rangle_j$ is the Perelomov $\mathfrak{su}(2)$ coherent state:
\begin{eqnarray}
  |\zeta\rangle_j & = & \frac{1}{(1+|\zeta|^2)^j} \sum_{m=-j}^{+j} \sqrt{\frac{(2j)!}{(j+m)!(j-m)!}} \zeta^{j-m} |jm\rangle \\
  |\zeta\rangle_{1/2} & = & \cos \frac{\theta}{2}|0\rangle + e^{\imath \varphi} \sin\frac{\theta}{2}|1\rangle
\end{eqnarray}

\end{document}